\documentclass[11pt]{article}

\usepackage[T1]{fontenc}
\usepackage[margin=1in]{geometry}
\usepackage{amsmath,amssymb,amsthm,mathtools,bm}
\usepackage{microtype}
\usepackage{enumitem}
\usepackage{xcolor}
\usepackage{comment}
\usepackage{float}
\usepackage{booktabs}
\usepackage{tabularx}
\usepackage{tikz}
\usetikzlibrary{arrows.meta,calc,positioning,decorations.markings}
\usepackage{aliascnt}
\usepackage[colorlinks=true,linkcolor=blue!55!black,citecolor=blue!55!black,urlcolor=blue!55!black,backref=page]{hyperref}
\usepackage[nameinlink,noabbrev]{cleveref}

\renewcommand*{\backrefalt}[4]{%
  \ifcase #1 Not cited.%
  \or page~#2.%
  \else pages~#2.%
  \fi}

\allowdisplaybreaks
\numberwithin{equation}{section}
\setlist[itemize]{leftmargin=2em,itemsep=0.25em,topsep=0.35em}
\setlist[enumerate]{leftmargin=2.3em,itemsep=0.3em,topsep=0.35em}
\newtheorem{theorem}{Theorem}[section]
\newaliascnt{lemma}{theorem}
\newtheorem{lemma}[lemma]{Lemma}
\aliascntresetthe{lemma}
\newaliascnt{proposition}{theorem}
\newtheorem{proposition}[proposition]{Proposition}
\aliascntresetthe{proposition}
\newaliascnt{corollary}{theorem}
\newtheorem{corollary}[corollary]{Corollary}
\aliascntresetthe{corollary}
\newaliascnt{definition}{theorem}
\newtheorem{definition}[definition]{Definition}
\aliascntresetthe{definition}
\newaliascnt{assumption}{theorem}

\aliascntresetthe{assumption}
\newaliascnt{remark}{theorem}
\newtheorem{remark}[remark]{Remark}
\aliascntresetthe{remark}
\newaliascnt{example}{theorem}

\aliascntresetthe{example}
\newaliascnt{fact}{theorem}
\newtheorem{fact}[fact]{Fact}
\aliascntresetthe{fact}

\newcommand{\Tr}{\operatorname{Tr}}
\newcommand{\supp}{\operatorname{supp}}
\newcommand{\diam}{\operatorname{diam}}
\newcommand{\dist}{\operatorname{dist}}
\newcommand{\id}{\mathbf 1}
\newcommand{\cH}{\mathcal H}
\newcommand{\cI}{\mathcal I}
\newcommand{\cD}{\mathcal D}
\newcommand{\str}[2]{\mathsf X_{#1:#2}}
\newcommand{\ket}[1]{\lvert #1\rangle}
\newcommand{\bra}[1]{\langle #1\rvert}
\newcommand{\braket}[2]{\langle #1\,|\,#2\rangle}

\newcommand{\norm}[1]{\left\lVert #1\right\rVert}
\newcommand{\abs}[1]{\left\lvert #1\right\rvert}
\newcommand{\eps}{\varepsilon}
\newcommand{\Dlt}{\Delta}
\newcommand{\e}{\mathrm e}
\newcommand{\poly}{\operatorname{poly}}
\newcommand{\Aone}{\texorpdfstring{\textbf{A1}}{A1}}

\hypersetup{
  pdftitle={Stability of the Modular Commutator and Hall Conductance Estimator from Small Conditional Mutual Information},
  pdfsubject={Modular Markov defects, modular structure, and approximate topological invariance of single-wave-function modular-commutator and Hall conductance estimator formulas}
}

\title{Modular commutator as a robust topological invariant and approximate Markovianity}
\author{Tai-Hsuan Yang\\[0.5em]
\small Department of Physics and Institute for Condensed Matter Theory,\\
\small University of Illinois Urbana-Champaign, Urbana, Illinois 61801, USA}
\date{\today}

\begin{document}
\maketitle

\begin{abstract}
The modular commutator provides a bulk, local, single-wave-function probe of the chiral central charge $c_-$ for gapped ground states. Its invariance under deformations was previously established under a local quantum Markov condition---namely, the conditional mutual information $I(A:C|B)$ is zero for all tripartitions of a disk into three consecutive strips A, B, and C \cite{Modular-commutator-Gapped}. However, the local quantum Markov condition also forces the probe to vanish, leaving open whether modular commutator remains robust in physically relevant states where the Markov property holds only approximately. In this paper, we resolve this tension for finite-dimensional Hilbert spaces: the approximate local quantum Markov property implies the change of the modular commutator under topology-preserving deformations vanishes asymptotically. We next prove trace-norm continuity of the modular commutator. Combining deformation invariance with trace-norm continuity, we establish that the modular commutator remains asymptotically invariant within gapped quantum phases connected by quasi-local unitary paths that preserve the approximate local quantum Markov condition. Conversely, by analyzing finite-time dynamics generated by modular Hamiltonians, we derive a quantitative lower bound on the conditional mutual information required for a non-zero modular commutator: $I(A:C|B)$ across such strip tripartitions cannot decay faster than exponentially with the width of B given that the state satisfies entanglement area law, extending the exact no-go theorem of \cite{strict-J-2024} to a quantitative finite bound . Finally, we demonstrate that those conclusions apply equally to the Hall conductance estimator \cite{FanSahayVishwanath2023}.
\end{abstract}

\clearpage
\tableofcontents
\clearpage

\section{Introduction}
\label{sec:introduction}

Two-dimensional gapped topological phases are characterized by universal data that are insensitive to local perturbations \cite{Bravyi2010,Bravyi_2011}.  These data include quasiparticle types, fusion and braiding rules, topological spins, and related modular data \cite{Kitaev2005,KitaevPreskill_2006,LevinWen2006}.  Chiral topological phases carry additional information associated with anomalous edge dynamics: the chiral central charge $c_-$ measures the imbalance between right- and left-moving edge modes and determines the universal low-temperature thermal Hall response \cite{KaneFisher1997,CappelliHuertaZemba2002}.  With an exact on-site $U(1)$ charge, a gapped two-dimensional phase can additionally be characterized by its quantized electrical Hall conductance $\sigma_{xy}$ \cite{ThoulessKohmotoNightingaleNijs1982,NiuThoulessWu1985}.

Remarkably, part of this general topological information can be extracted directly from the entanglement structure of a representative ground-state wave function locally \cite{Kim2022c-minus,Modular-commutator-Gapped,vardhan2025,Kimmodularaddivity, Huang2021knots,yang2025topologicalmixedstatesphases, Zou_2022,zhang2026extensivelongrangemagicnonabelian}.  The topological entanglement entropy isolates the total quantum dimension, while the entanglement Hamiltonian retains detailed signatures of the effective boundary theory \cite{LiHaldane2008,li2026chiralgappedstatesuniversally}; more elaborate local entropy constraints can also recover superselection and fusion data \cite{KitaevPreskill_2006,LevinWen2006,ShiKatoKim2020}.  These quantities probe static topological structure, but they are not themselves thermal or electrical response coefficients.

To access the chiral central charge from the same single-wave-function perspective, consider modular Hamiltonians.  For a region $X$ with reduced state $\rho_X$, let $K_X=-\log\rho_X$; when $\rho_X$ is not faithful, the logarithm is understood on $\supp\rho_X$.  Divide a disk into three cyclically ordered sectors $A,B,C$ as in \cref{fig:intro-geometries}(a).  The modular commutator and the Hall conductance estimator are
\begin{align}
J(A,B,C)_\rho
&:=i\Tr\rho_{ABC}[K_{AB},K_{BC}] = \frac{\pi}{3}c_-,
\label{eq:intro-J}\\
\Sigma(A,B,C)_\rho
&:=\frac{i}{2}\Tr\rho_{ABC}[K_{AB},Q_{BC}^{2}] = \sigma_{xy},
\label{eq:intro-Sigma}
\end{align}
where the second expression assumes an on-site $U(1)$ charge with regional charge $Q_{BC}$.  For 2D gapped many-body ground states studied in Refs.~\cite{Modular-commutator-Gapped,FanSahayVishwanath2023,Kim2022c-minus,strict-J-2024}, $J$ probes the chiral central charge and $\Sigma$ probes the electrical Hall conductance. Related studies have derived a universal expression for the modular commutator in $(1+1)$-dimensional conformal field theories \cite{Zou_2022} and established its geometric additivity for multipartite partitions \cite{Kimmodularaddivity}. Recent work on two-dimensional gapless systems proposes the modular commutator as a probe of the parity anomaly \cite{zeng2026parityanomalymodularcommutator}.

\begin{figure}[H]
\centering
\begin{tikzpicture}[scale=1.0,line width=0.8pt]
  \begin{scope}[xshift=-2.6cm]
    \def\R{1.35}
    \fill[blue!12] (0,0) -- (30:\R) arc (30:150:\R) -- cycle;
    \fill[red!12] (0,0) -- (150:\R) arc (150:270:\R) -- cycle;
    \fill[yellow!22] (0,0) -- (270:\R) arc (270:390:\R) -- cycle;
    \draw[thick] (0,0) circle (\R);
    \draw[thick] (0,0) -- (30:\R);
    \draw[thick] (0,0) -- (150:\R);
    \draw[thick] (0,0) -- (270:\R);
    \node at (90:0.76) {$C$};
    \node at (210:0.76) {$A$};
    \node at (330:0.76) {$B$};
    \node at (0,-1.68) {\textbf{(a)} Response partition};
  \end{scope}
  \begin{scope}[xshift=2.6cm]
    \fill[green!10] (0,0) circle (1.35);
    \fill[white] (0,0) circle (0.58);
    \draw[thick] (0,0) circle (1.35);
    \draw[thick] (0,0) circle (0.58);
    \draw[thick] (90:0.58)--(90:1.35);
    \draw[thick] (270:0.58)--(270:1.35);
    \node at (0,0) {$C$};
    \node at (-0.92,0) {$B$};
    \node at (0.92,0) {$D$};
    \node at (0,-1.68) {\textbf{(b)} \Aone{} geometry};
  \end{scope}
\end{tikzpicture}
\caption{The two geometries used in the introduction.  (a) The three-sector disk partition defining the modular commutator $J(A,B,C)$ and Hall conductance estimator $\Sigma(A,B,C)$; the cyclic order of $A,B,C$ (clockwise or counterclockwise) fixes the response orientation and hence the sign convention.  (b) The local \Aone{} geometry, in which the annular neighborhood of $C$ is split into side regions $B$ and $D$.}
\label{fig:intro-geometries}
\end{figure}
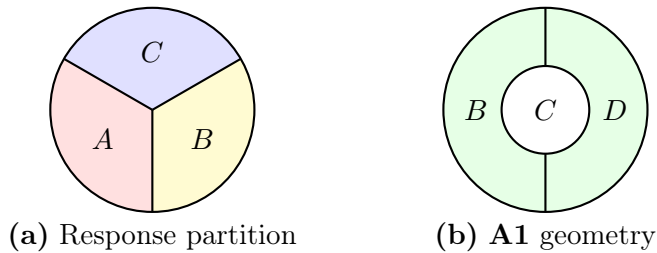

Since the choice of the partition is arbitrary while the topological data extracted is universal, $J$ and $\Sigma$ must be deformation invariant. The entropy condition relevant to deformation invariance  is formulated on the local geometry in \cref{fig:intro-geometries}(b) and one considers the following entropy combinations:
\begin{equation}
\delta_1(B,C,D)
:=S(BC)+S(CD)-S(B)-S(D).
\label{eq:intro-A1}
\end{equation}
Requiring $\delta_1(B,C,D)=0$ uniformly over the relevant local geometries is conventionally called the exact \Aone{} condition in the entanglement-bootstrap program \cite{ShiKatoKim2020} where one expect the \Aone{} violation $\delta_1$ is small uniformly for gapped quantum ground state with topological order. Exact satisfaction of this entropy condition guarantees deformation invariance of the response formulas. It allows local patches to be added, removed, or reassigned without changing either $J$ or $\Sigma$.  Thus, although the three-sector partition in \cref{fig:intro-geometries}(a) is auxiliary, the extracted value is invariant under topology- and orientation-preserving deformations of its boundaries.

The exact \Aone{} condition is, however, too restrictive for chiral phases.
The no-go theorem of Ref.~\cite{strict-J-2024} shows that exact bulk
\Aone{} forces the $c_-$ extracted from the modular commutator to vanish and,
in the presence of exact on-site $U(1)$ symmetry, likewise forces the
$\sigma_{xy}$ extracted from the Hall conductance estimator to vanish.  Moreover, exact
\Aone{} implies the existence of a local commuting-projector parent Hamiltonian with a stable gap \cite{strict-H-2024}.  A state with a
nonzero modular commutator or Hall conductance estimator must therefore exhibit \Aone{} violations
throughout the bulk. This leaves open two distinct questions.
First, for a quantum state, do the single-wave-function response formulas
remain invariant under topology-preserving deformations when
\Aone{} only holds asymptotically?  Second, after
taking the thermodynamic limit, do the extracted responses remain unchanged
between representative states in the same gapped quantum phase?  We answer
both questions affirmatively: sufficiently small \Aone{} violations make $J$
and $\Sigma$ asymptotically deformation invariant, even though exact \Aone{}
fails at finite scales.  The second question is addressed by proving the
trace-norm continuity of $J$ and $\Sigma$ and using the definition of the phase equivalence from entanglement bootstrap through the spatial domain-wall construction\cite{yang2025topologicalmixedstatesphases} from quasi-local unitary paths that preserves the $\Aone{}$ condition approximately.

Finally, we establish a complementary result in the converse direction: a finite modular commutator or Hall conductance estimator requires the \Aone{} violation to remain strictly bounded away from zero. More explicitly, under an entanglement area law or a boundary law for charge fluctuations,\footnote{Following common usage, scaling proportional to $|\partial A|$ is called an ``area law'' for entanglement entropy and a ``boundary law'' for charge fluctuations in the literatures.  Both terms refer to the same scaling with the size of the region's boundary, which is $O(L)$ in two dimensions.} the violation of \Aone{} cannot decay faster than exponentially with the buffer width. We establish this result by extending the instantaneous modular flow technique—originally used to prove the exact no-go theorem—to the approximate Markov regime. Thus, a chiral response is compatible with \Aone{} violations that vanish asymptotically, but strictly incompatible with exact \Aone{} or with decay faster than our quantitative lower bound.

The paper is organized as follows.  \Cref{sec:framework} introduces the modular Markov defect and establishes quantitative deformation stability of the modular commutator and Hall conductance estimator.  \Cref{sec:J-phase-equivalence} proves trace-norm continuity and shows that the modular commutator and the Hall response estimator are phase invariant in the thermodynamics limit under the admissible quasi-local unitary equivalences defining a gapped quantum phase.  \Cref{sec:finite-time-IMF-no-go} develops the finite-time instantaneous-modular-flow bounds and derives spectral-free quantitative lower bounds on the finite-scale \Aone{} violation required by a nonzero modular commutator or Hall conductance estimator.  \Cref{sec:discussion} summarizes the results and discusses open questions and extensions.  The appendices provide technical proofs, extensions to nonfaithful states, and supporting estimates and constructions.

\section{Deformation invariance of the modular commutator}
\label{sec:framework}

\subsection{Modular Markov defect}
\label{subsec:framework-defect}

For a tripartite state $\rho_{ABC}$, define the modular Markov defect by
\begin{equation}
\Dlt_{A:B:C}:=K_{AB}+K_{BC}-K_B-K_{ABC}.
\label{eq:modular-defect-main}
\end{equation}
For an operator $O$ and a state $\rho$, the right-GNS norm is
\begin{equation}
\norm{O}_{2,\rho}^{\mathrm R}
:=\norm{O\rho^{1/2}}_2
=\bigl(\Tr\rho O^\dagger O\bigr)^{1/2}.
\label{eq:right-GNS-main}
\end{equation}
The conditional mutual information is the expectation of the defect,
\begin{equation}
I(A:C|B)_\rho
=\Tr(\rho_{ABC}\Dlt_{A:B:C})
=S(AB)+S(BC)-S(B)-S(ABC),
\label{eq:CMI-defect-main}
\end{equation}

The following theorem is the main technical input of the paper.  It gives a quantitative stability version of the exact quantum-Markov identity by showing that small conditional mutual information forces the modular Markov defect to be small in the state-dependent right-GNS norm.

\begin{theorem}[CMI controls the modular second moment]
\label{thm:main-cmi}
Let
\[
\rho_{ABC}\in\cD(\cH_A\otimes\cH_B\otimes\cH_C),
\qquad d_X:=\dim\cH_X.
\]
Here $\cD(\cH)$ denotes the set of positive semidefinite, unit-trace operators on $\cH$.
For every finite-dimensional tripartite state, with
$u:=I(A:C|B)_\rho$,
\begin{equation}
\boxed{
\norm{\Dlt_{A:B:C}\rho_{ABC}^{1/2}}_2
\le\Lambda_{A:B:C}(u).
}
\label{eq:main-cmi-bound}
\end{equation}
where, for an ordered tripartition $X|Y|Z$, the error function is
\begin{equation}
\Lambda_{X:Y:Z}(u)
:=
\begin{cases}
0,
&u=0,\\[0.2em]
C_0\sqrt{u\left[1+\log(d_{XY}d_{YZ})+\log(1/u)\right]},
&0<u\le\frac12,\\[0.2em]
\log(d_{XY}d_{YZ}d_Yd_{XYZ}),
&u>\frac12,
\end{cases}
\label{eq:Lambda-main}
\end{equation}
Here we take $C_0 = 14$.
\end{theorem}
The proof of \cref{thm:main-cmi} rests on the following three auxiliary lemmas.  Their proofs are given in \cref{app:cmi-proof}. In the regime considered throughout this paper, we assume every region $W$ appearing below has $d_W\ge4$.  Let $p_i$ be the eigenvalues of $\rho_W$, a direct optimization of $\sum_i p_i(\log p_i)^2$ gives
\begin{equation}
\norm{K_W\rho_W^{1/2}}_2
\le \log d_W.
\label{eq:universal-modular-moment-main}
\end{equation}
Equality is achieved by the uniform spectrum $p_i=1/d_W$ for every $i$.

\begin{lemma}[Depolarization smoothing]
\label{lem:smoothing-app}
Let $\sigma$ be a density operator on a $d$-dimensional Hilbert space and, for $0<\eta\le\frac12$, define
\[
\sigma^{(\eta)}:=(1-\eta)\sigma+\eta\frac{\id}{d}.
\]
Then
\begin{equation}
\Tr\sigma\bigl(\log\sigma-\log\sigma^{(\eta)}\bigr)^2
\le2\eta.
\label{eq:smoothing-app}
\end{equation}
\end{lemma}

\begin{lemma}[Bounded relative-surprisal moment]
\label{lem:scalar-moment-app}
Let $X$ be a real random variable satisfying
\[
\abs X\le L,
\qquad
\mathbb E X=u,
\qquad
\mathbb E\e^{-X}\le1.
\]
Then
\begin{equation}
\mathbb E X^2\le2(1+L)u.
\label{eq:scalar-moment-app}
\end{equation}
\end{lemma}

\begin{lemma}[Finite-dimensional CMI continuity]\cite{Audenaert_2007,Winter_2016}
\label{lem:cmi-continuity-app}
If $\frac12\norm{\rho_{ABC}-\sigma_{ABC}}_1\le\eta\le\frac12$, then
\begin{equation}
\abs{I(A:C|B)_\rho-I(A:C|B)_\sigma}
\le4\eta\log d_*+4h_2(\eta),
\qquad d_*:=\min\{d_A,d_C\},
\label{eq:cmi-continuity-app}
\end{equation}
where $h_2$ is the binary entropy.
\end{lemma}

\begin{proof}[Proof of \cref{thm:main-cmi}]
To display the dependence of the modular Markov defect on the state, we use
the abbreviation
\begin{equation}
\Dlt_\rho
:=\Dlt_{A:B:C}(\rho)
=K_{AB}(\rho)+K_{BC}(\rho)-K_B(\rho)-K_{ABC}(\rho),
\label{eq:state-dependent-defect-main}
\end{equation}
and use the same convention for any other state. 

The endpoint and large-CMI cases are immediate.  If $u=0$, the exact quantum-Markov structure \cite{HaydenJozsaPetzWinter2004} gives
$\Dlt_\rho=0$ on $\supp\rho_{ABC}$ and hence
$\Dlt_\rho\rho_{ABC}^{1/2}=0$
\cite[Eq.~(12)]{Modular-commutator-Gapped}.  If $u>\frac12$, the triangle inequality and \eqref{eq:universal-modular-moment-main} give
\begin{equation}
\norm{\Dlt_\rho\rho_{ABC}^{1/2}}_2
\le\log(d_{AB}d_{BC}d_Bd_{ABC}),
\label{eq:large-cmi-fallback-main}
\end{equation}
which is the last branch of \eqref{eq:Lambda-main}.

Now for the branch where $0<u<1/2$, introduce the globally depolarized state of $\rho$: 
\[
\rho^{(\eta)}=(1-\eta)\rho+\eta\frac{\id}{d_{ABC}},
\qquad 0<\eta\le\frac12,
\].
Every marginal is depolarized with the same parameter.  Applying
\cref{lem:smoothing-app} separately to the $ABC$, $AB$, $BC$, and $B$
logarithms, and then using the triangle inequality, gives
\[
\norm{(\Dlt_\rho-\Dlt_{\rho^{(\eta)}})\rho^{1/2}}_2
\le4\sqrt{2\eta}.
\]
Moreover, $\rho^{(\eta)}\ge(1-\eta)\rho$, so changing the GNS weight from $\rho$ to $\rho^{(\eta)}$ costs only
\[
\norm{\Dlt_{\rho^{(\eta)}}\rho^{1/2}}_2
\le(1-\eta)^{-1/2}
\norm{\Dlt_{\rho^{(\eta)}}(\rho^{(\eta)})^{1/2}}_2.
\]
It therefore remains to bound the modular second moment of the full-rank state $\rho^{(\eta)}$. Let
\[
\tau_\eta
:=\exp\!\left(
\log\rho_{AB}^{(\eta)}+
\log\rho_{BC}^{(\eta)}-
\log\rho_B^{(\eta)}
\right).
\]
$\Tr\tau_\eta\le1$ is derived in appendix \ref{app:identities}.  Directly from the definition of $\tau_\eta$, we have
\[
\Dlt_{\rho^{(\eta)}}=\log\rho^{(\eta)}-\log\tau_\eta,
\qquad
u_\eta:=I(A:C|B)_{\rho^{(\eta)}}
=\Tr\rho^{(\eta)}
\bigl(\log\rho^{(\eta)}-\log\tau_\eta\bigr).
\]
Diagonalize the two generally noncommuting operators in their respective eigenbases:
\[
\rho^{(\eta)}=\sum_i p_i\ket i\bra i,
\qquad
\tau_\eta=\sum_j q_j\ket j\bra j.
\]
The overlaps of these bases define a probability distribution
\[
P_{ij}:=p_i\abs{\braket{i}{j}}^2,
\qquad
\sum_{i,j}P_{ij}=1.
\]
For a pair $(i,j)$ sampled according to $P_{ij}$, define
\[
X_{ij}:=\log p_i-\log q_j.
\]
Expanding in the two eigenbases shows directly that
\begin{align}
\mathbb E X
&=u_\eta, \qquad 
\mathbb E X^2
=\Tr\rho^{(\eta)}
\bigl(\log\rho^{(\eta)}-\log\tau_\eta\bigr)^2
=\norm{\Dlt_{\rho^{(\eta)}}(\rho^{(\eta)})^{1/2}}_2^2,
\notag\\
\mathbb E\e^{-X}
&=\sum_{i,j}P_{ij}\frac{q_j}{p_i}
=\sum_jq_j
=\Tr\tau_\eta\le1.
\label{eq:proof-sketch-moments-main}
\end{align}
Thus the noncommutative modular second moment is exactly the classical second moment of $X$ under the coupling $P$.  By \cref{lem:depolarized-markov-floor-app}, depolarization supplies the spectral floors
\begin{equation}
\rho^{(\eta)}\ge\frac{\eta}{d_{ABC}}\id,
\qquad
\tau_\eta\ge\frac{\eta^2}{d_{AB}d_{BC}}\id.
\label{eq:spectral-floors-main}
\end{equation}
Together with $p_i,q_j\le1$, \eqref{eq:spectral-floors-main} implies, for every pair $(i,j)$,
\begin{equation}
\abs{X_{ij}}
=\abs{\log p_i-\log q_j}
\le L_\eta
:=2\log\frac1\eta+\log(d_{AB}d_{BC}).
\label{eq:relative-surprisal-bound-main}
\end{equation}
Indeed, $X_{ij}\le-\log q_j\le L_\eta$, while
$-X_{ij}\le-\log p_i\le\log(d_{ABC}/\eta)\le L_\eta$.
Applying \cref{lem:scalar-moment-app} to the random variable just constructed therefore yields 
\begin{equation}
\norm{\Dlt_{\rho^{(\eta)}}(\rho^{(\eta)})^{1/2}}_2^2
\le2(1+L_\eta)u_\eta.
\label{eq:regularized-second-main}
\end{equation}

Finally, \cref{lem:cmi-continuity-app} compares the smoothed and original conditional mutual informations:
\begin{equation}
u_\eta\le u+4\eta\log d_*+4h_2(\eta),
\qquad
d_*:=\min\{d_A,d_C\}.
\label{eq:smoothed-cmi-main}
\end{equation}
Set
\begin{equation}
\Gamma:=1+\log(d_{AB}d_{BC})+\log\frac1u,
\qquad
\eta:=\frac{u}{\Gamma}.
\label{eq:eta-choice-main}
\end{equation}
The elementary constant estimates collected in
\cref{app:cmi-constant-evaluation} give
$u_\eta\le9u$, $L_\eta\le3\Gamma$, and $\eta\le3/10$.
Combining these bounds with \eqref{eq:regularized-second-main} and the
preceding smoothing and change-of-weight bounds gives
\begin{equation}
\norm{\Dlt_\rho\rho^{1/2}}_2
\le4\sqrt{2\eta}
+(1-\eta)^{-1/2}\sqrt{2(1+L_\eta)u_\eta}
\le14\sqrt{u\Gamma},
\label{eq:combined-cmi-bound-main}
\end{equation}
which is \eqref{eq:main-cmi-bound} for $0<u\le\frac12$ and $C_0=14$. In particular, no lower spectral bound on the original reduced density matrices is assumed.
\end{proof}

\begin{corollary}[Regular many-body regions]
\label{cor:regular-regions-main}
Consider fixed-local-dimension regions $A,B,C$ of comparable linear scale $\ell$ in $s$ spatial dimensions, with $B$ separating $A$ from $C$.  If
\[
|A|,|B|,|C|=O(\ell^s),
\qquad
I(A:C|B)_\rho\le C\ell^s\e^{-c\ell},
\]
then
\begin{equation}
\norm{\Dlt_{A:B:C}\rho_{ABC}^{1/2}}_2
\le C'\ell^s\e^{-c\ell/2}.
\label{eq:regular-defect-main}
\end{equation}
\end{corollary}

\begin{proof}
For local dimension $q$,
\[
\log(d_{AB}d_{BC})
=(|A|+2|B|+|C|)\log q
=O(\ell^s).
\]
Substitution into \cref{thm:main-cmi} proves \eqref{eq:regular-defect-main}.
\end{proof}

\subsection{Approximate \Aone{} condition}
\label{subsec:A1-main}

Physically relevant ground states of gapped local Hamiltonians are expected to satisfy \Aone{} approximately at scales much larger than the correlation length. Therefore, it motivates us to consider the following state family.

\begin{definition}[Uniformly \Aone-admissible state family]
\label{def:A1-admissible-state}
A family of pure states $\boldsymbol\rho=\{\rho_\Lambda\}_\Lambda$ on closed
lattices is \emph{uniformly \Aone-admissible} if there exist a constant
$c_0>0$ and a polynomial $p_1$, independent of $\Lambda$ and of the partition,
such that every \Aone{} violation on the geometry of \cref{fig:intro-geometries}(b) satisfies the following bound.
Here $\ell$ denotes its characteristic length scale, with all relevant
separating-buffer widths taken to be of order $\ell$:
\begin{equation}
0 \le \delta_1^{\rho_\Lambda}(B,C,D)
:=
S_{\rho_\Lambda}(BC)+S_{\rho_\Lambda}(CD)
-S_{\rho_\Lambda}(B)-S_{\rho_\Lambda}(D)
\le p_1(|C|)\e^{-c_0\ell}.
\label{eq:A1-admissible-bound}
\end{equation}
uniformly in the system size, position, shape, and scale.
\end{definition}

For the Hall conductance, we further need to assume the $U(1)$ symmetry, so the state family we consider becomes:
\begin{definition}[Uniformly $U(1)$-symmetric \Aone-admissible state family]
\label{def:U1-A1-admissible-state-phase}
Fix an on-site charge
\begin{equation}
  Q_\Lambda:=\sum_{v\in\Lambda}Q_v,
  \qquad
  \sup_v\norm{Q_v}_\infty\le q_0<\infty.
  \label{eq:onsite-charge-phase}
\end{equation}
A uniformly \Aone-admissible pure-state family
$\boldsymbol\rho=\{\ket{\Psi_\Lambda}\!\bra{\Psi_\Lambda}\}_\Lambda$ is
\emph{uniformly $U(1)$-symmetric} if every representative has sharp total
charge,
\begin{equation}
  Q_\Lambda\ket{\Psi_\Lambda}
  =q_\Lambda\ket{\Psi_\Lambda}.
  \label{eq:sharp-total-charge-phase}
\end{equation}
Equivalently, every reduced state obeys
$[\rho_{\Lambda,X},Q_X]=0$.
\end{definition}

The quantum states considered below are uniformly \Aone-admissible families for the modular commutator and uniformly $U(1)$-symmetric \Aone-admissible families for the Hall conductance estimator.  In particular, any family with a persistent $O(1)$ \Aone{} violation on some local geometry lies outside this class; \cref{app:spurious-A1} gives an explicit cluster-chain example that is common in demonstrating the spurious contribution of topological entanglement entropy and modular commutator \cite{Zou2016,Williamson2019, Gass_2024}. The role of \Aone{} in the deformation argument is to control the CMI of the following  tripartition shown in \cref{fig:principal-Markov-partition}.

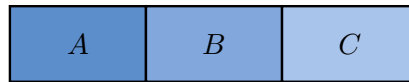
\begin{figure}[H]
\centering
\begin{tikzpicture}[line width=0.9pt, font=\normalsize]
  \definecolor{colorA}{HTML}{5C8ECC}
  \definecolor{colorB}{HTML}{81A7DD}
  \definecolor{colorC}{HTML}{A9C4EA}

  \draw[fill=colorA] (0,0) rectangle (1.8,1.0) node[pos=0.5] {$A$};
  \draw[fill=colorB] (1.8,0) rectangle (3.6,1.0) node[pos=0.5] {$B$};
  \draw[fill=colorC] (3.6,0) rectangle (5.4,1.0) node[pos=0.5] {$C$};
\end{tikzpicture}
\caption{The principal Markov partition $A|B|C$.  For any subsystem $A$ of
the purifying complement of the \Aone{} geometry, approximate \Aone{}
controls the conditional mutual information $I(A:C|B)$.}
\label{fig:principal-Markov-partition}
\end{figure}

To see explicitly how \Aone{} supplies the local Markov relation, let $E$ purify the physical regions $BCD$ in \cref{fig:intro-geometries}(b).  The purification gives
\begin{align}
\delta_1(B,C,D)
&=S(BC)+S(CD)-S(B)-S(D)
\notag\\
&=S(BC)+S(BE)-S(B)-S(BCE)
=I(C:E|B).
\label{eq:A1-purification-main}
\end{align}
If $A$ is any subsystem of the purifying complement $E$, tracing out
$E\setminus A$ and applying data processing for conditional mutual
information gives
\begin{equation}
I(A:C|B)\le \delta_1(B,C,D)
\le p_1(|C|)\e^{-c_0\ell}.
\label{eq:A1-controls-CMI-main}
\end{equation}
Thus approximate \Aone{} supplies the small-CMI input used in the deformation
bounds below.

\subsection{Deformation invariance of the modular commutator}
\label{subsec:J-main}

Before analyzing the individual deformation moves, we establish a general
bound showing that small CMI implies a small modular commutator.  The
deformation proof below will apply this bound to the tripartitions
associated with each local move.

\begin{lemma}[Small CMI implies a small modular commutator]
\label{lem:small-J-main}
Let $\rho_{ABC}$ be a finite-dimensional state and set
$u:=I(A:C|B)_\rho$.  Then
\begin{equation}
\abs{J(A,B,C)_\rho}
\le2\min\{\log d_{AB},\log d_{BC}\}\,\Lambda_{A:B:C}(u).
\label{eq:small-J-main}
\end{equation}
\end{lemma}

\begin{proof}
The following form of commutator expectation value can be bounded with the triangle inequality and Cauchy--Schwarz for the Hilbert--Schmidt inner product:
\begin{align}
\abs{i\Tr\rho[P,Q]}
&\le
\abs{\Tr\bigl[(P\rho^{1/2})^\dagger(Q\rho^{1/2})\bigr]}
+\abs{\Tr\bigl[(Q\rho^{1/2})^\dagger(P\rho^{1/2})\bigr]}
\notag
\\
&\le2\norm{P\rho^{1/2}}_2\norm{Q\rho^{1/2}}_2.
\label{eq:weighted-comm-main}
\end{align}
Here $P,Q$ are Hermitian.  The universal modular second-moment bound
\eqref{eq:universal-modular-moment-main} gives
$\norm{K_X\rho^{1/2}}_2\le\log d_X$.  Substitute
$K_{AB}=K_{ABC}-K_{BC}+K_B+\Dlt_{A:B:C}$ into the definition of $J$.  The terms containing $K_{ABC}$, $K_{BC}$, and $K_B$ have zero commutator expectation by cyclicity and nested support, leaving
\begin{equation}
J(A,B,C)_\rho
=i\Tr\rho[\Dlt_{A:B:C},K_{BC}]
=i\Tr\rho[K_{AB},\Dlt_{A:B:C}].
\label{eq:J-defect-main}
\end{equation}
The second equality follows by instead substituting
$K_{BC}=K_{ABC}-K_{AB}+K_B+\Dlt_{A:B:C}$; the additional terms again
have zero commutator expectation.
Applying \eqref{eq:weighted-comm-main},
\eqref{eq:universal-modular-moment-main}, and \cref{thm:main-cmi} to the
two expressions in \eqref{eq:J-defect-main} gives the bounds with
prefactors $2\log d_{BC}$ and $2\log d_{AB}$, respectively.
Taking their minimum proves \eqref{eq:small-J-main}.
\end{proof}

We next pass to a finite deformation.  The proof below decomposes the
deformation into elementary moves and controls each change using the
weighted commutator inequality underlying \cref{lem:small-J-main}.
Summing these one-step errors gives the desired bound.

\begin{theorem}[Approximate invariance of the modular commutator]
\label{thm:main-J}
Let $\ket\psi$ be a finite-dimensional pure state satisfying approximate \Aone{} uniformly on all local geometries used in a topology-preserving deformation between tripartitions $\mathcal G_0$ and $\mathcal G_N$.  Resolve the deformation into endpoint moves, middle-region moves, and exact complementary-region flips, as illustrated in \cref{fig:J-deformation-moves}.  Then
\begin{equation}
\abs{J(\mathcal G_N)-J(\mathcal G_0)}
\le\sum_{j=1}^{N}E_j,
\label{eq:J-global-main}
\end{equation}
where each $E_j$ is the explicit one-step error in \cref{eq:endpoint-main,eq:middle-main}.  If the \Aone{} violations obey \eqref{eq:A1-admissible-bound}, region volumes and the number of moves are polynomial in $\ell$, then
\begin{equation}
\boxed{
\abs{J(\mathcal G_N)-J(\mathcal G_0)}
\le p_J(\ell)\e^{-c_0\ell/2}
}
\label{eq:J-exp-main}
\end{equation}
for a polynomial $p_J$.
\end{theorem}

\begin{figure}[ht]
\centering
\begin{tikzpicture}[line cap=round,line join=round]
  \def\Jdisk{%
    \def\R{1.15}%
    \fill[blue!12] (0,0)--(30:\R) arc (30:150:\R)--cycle;
    \fill[red!11] (0,0)--(150:\R) arc (150:270:\R)--cycle;
    \fill[yellow!22] (0,0)--(270:\R) arc (270:390:\R)--cycle;
    \draw[thick] (0,0) circle (\R);
    \draw (0,0)--(30:\R);
    \draw (0,0)--(150:\R);
    \draw (0,0)--(-90:\R);
    \node at (90:0.55) {$B$};
    \node at (-30:0.55) {$A$};
    \node at (-150:0.55) {$C$};
  }

  \begin{scope}[xshift=-4.5cm,yshift=1.65cm]
    \filldraw[fill=orange!32,draw=orange!70!black,thick]
      (-30:1.15) circle (0.28);
    \Jdisk
    \node at (-30:1.31) {$a$};
    \node[align=center,font=\scriptsize] at (0,-1.55)
      {\textbf{(a)} Endpoint move\\$A\to aA$};
  \end{scope}

  \begin{scope}[yshift=1.65cm]
    \filldraw[fill=orange!32,draw=orange!70!black,thick]
      (-150:1.15) circle (0.28);
    \Jdisk
    \node at (-150:1.31) {$c$};
    \node[align=center,font=\scriptsize] at (0,-1.55)
      {\textbf{(b)} Endpoint move\\$C\to Cc$};
  \end{scope}

  \begin{scope}[xshift=4.5cm,yshift=1.65cm]
    \filldraw[fill=orange!32,draw=orange!70!black,thick]
      (90:1.15) circle (0.28);
    \Jdisk
    \node at (90:1.31) {$b$};
    \node[align=center,font=\scriptsize] at (0,-1.55)
      {\textbf{(c)} Middle move\\$B\to Bb$};
  \end{scope}

  \begin{scope}[xshift=-2.3cm,yshift=-2.15cm]
    \filldraw[fill=orange!32,draw=orange!70!black,thick]
      (30:1.15) circle (0.30);
    \Jdisk
    \node at (30:1.32) {$b$};
    \node[align=center,font=\scriptsize] at (0,-1.55)
      {\textbf{(d)} Junction move \\$B\to Bb$};
  \end{scope}

  \begin{scope}[xshift=2.3cm,yshift=-2.15cm]
    \def\R{1.15}
    \fill[blue!12] (0,0)--(30:\R) arc (30:150:\R)--cycle;
    \fill[red!11] (0,0)--(150:\R) arc (150:270:\R)--cycle;
    \fill[yellow!22] (0,0)--(270:\R) arc (270:390:\R)--cycle;
    \begin{scope}
      \clip (0,0)--(-90:\R)
        arc[start angle=-90,end angle=30,radius=\R]--cycle;
      \filldraw[fill=orange!32,draw=orange!70!black,thick]
        (-30:0.08) circle (0.48);
    \end{scope}
    \draw[thick] (0,0) circle (\R);
    \draw (0,0)--(30:\R);
    \draw (0,0)--(150:\R);
    \draw (0,0)--(-90:\R);
    \node at (90:0.55) {$B$};
    \node at (-150:0.55) {$C$};
    \node at (-30:0.83) {$A$};
    \node at (-30:0.25) {$b$};
    \node[align=center,font=\scriptsize] at (0,-1.55)
      {\textbf{(e)} Interior move \\$B\to Bb$};
  \end{scope}
\end{tikzpicture}
\caption{Representative local geometries in the deformation proof for the modular commutator.  Panels (a) and (b) add an endpoint patch while keeping $BC$ or $AB$, respectively, fixed.  Panel (c) adds a patch to the middle region, buffered from both endpoints.  Junction patches such as (d) are converted by an exact complementary-region flip into an endpoint or middle geometry represented by (e).  The orange patch lies in the complementary region $D$ before reassignment; cyclic rotations cover the remaining cases.}
\label{fig:J-deformation-moves}
\end{figure}
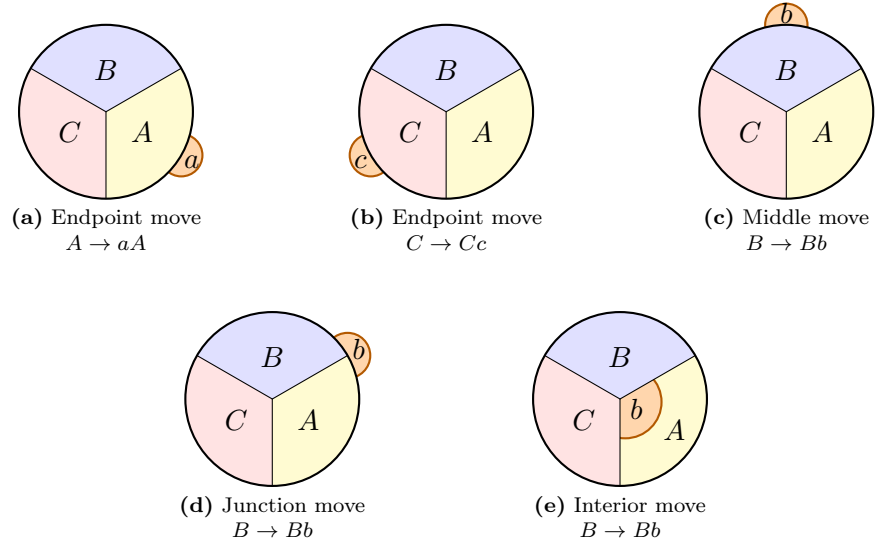

\begin{proof}
For the endpoint move in \cref{fig:J-deformation-moves}(a), attach a patch $a$ to $A$ in a chain-like geometry $a|A|B|C$. We have the following decomposition
\begin{equation}
K_{aAB}=K_{aA}+K_{AB}-K_A-\Dlt_{a:A:B}.
\end{equation}
Since both $K_{aA}$ and $K_A$ have support disjoint from $BC$, substituting this decomposition into $J(aA,B,C)$ gives the exact difference
\begin{equation}
J(aA,B,C)-J(A,B,C)
=-i\langle[\Dlt_{a:A:B},K_{BC}]\rangle.
\label{eq:endpoint-identity-main}
\end{equation}
The weighted commutator inequality \eqref{eq:weighted-comm-main}, the modular
second-moment bound \eqref{eq:universal-modular-moment-main}, and
\cref{thm:main-cmi} therefore give
\begin{equation}
\abs{J(aA,B,C)-J(A,B,C)}
\le2\log d_{BC}\Lambda_{a:A:B}\!\left(I(a:B|A)\right).
\label{eq:endpoint-main}
\end{equation}

For the middle move in \cref{fig:J-deformation-moves}(c), add a patch $b$ to the separating region $B$ and assume that it is buffered from both endpoints.  Set
\begin{equation}
u_L:=I(A:b|B),
\qquad
u_R:=I(b:C|B).
\label{eq:middle-u-main}
\end{equation}
The modular Markov defects for the two tripartitions satisfy the exact
operator identity
\begin{equation}
\Dlt_{A:B:b}-\Dlt_{b:B:C}
=-K_{ABb}+K_{AB}+K_{BbC}-K_{BC}.
\label{eq:delta-difference-main}
\end{equation}
Expanding the two modular commutators and using
$\langle[K_X,K_Y]\rangle=0$ whenever $X\subseteq Y$ gives
\begin{equation}
J(A,Bb,C)-J(A,B,C)
=-i\left\langle
[\Dlt_{A:B:b}-\Dlt_{b:B:C},K_{AB}+K_{BbC}]
\right\rangle.
\label{eq:middle-identity-main}
\end{equation}
The same identity may be evaluated with $K_{ABb}+K_{BC}$ in place of $K_{AB}+K_{BbC}$ after the nested-support terms are removed. Applying \eqref{eq:weighted-comm-main} separately to $\Dlt_{A:B:b}$ and $\Dlt_{b:B:C}$, and then applying \cref{thm:main-cmi}, gives
\begin{equation}
\begin{aligned}
\abs{J(A,Bb,C)-J(A,B,C)}
\le{}&2\min\{\log(d_{AB}d_{BbC}),\,\log (d_{ABb}d_{BC})\}
\Bigl[
\Lambda_{A:B:b}(u_L)+\Lambda_{b:B:C}(u_R)
\Bigr].
\end{aligned}
\label{eq:middle-main}
\end{equation}

For a pure global state, complementary modular Hamiltonians act identically on the state.  This yields the exact support flips illustrated in \cref{fig:J-deformation-moves}(d) and (e), such as
\begin{equation}
J(A,B,C)=J(D,C,B)=J(C,D,A),
\label{eq:J-flip-main}
\end{equation}
which reduce the remaining junction moves to endpoint or middle moves.  Applying \eqref{eq:endpoint-main} or \eqref{eq:middle-main} at each step and telescoping the scalar differences gives \eqref{eq:J-global-main}.  Under \eqref{eq:A1-admissible-bound}, each local $\Lambda$ is bounded by a polynomial times $\e^{-c_0\ell/2}$; polynomially many moves and polynomial modular second moments therefore give \eqref{eq:J-exp-main}.
\end{proof}

\subsection{Deformation invariance of the Hall conductance estimator}
\label{subsec:Sigma-main}

For the Hall conductance estimator, assume exact on-site $U(1)$ symmetry:
\begin{equation}
Q=\sum_v Q_v,
\qquad
Q\ket\Psi=q\ket\Psi,
\qquad
Q_X:=\sum_{v\in X}Q_v.
\label{eq:U1-main}
\end{equation}

\begin{theorem}[Approximate invariance of the Hall conductance estimator]
\label{thm:main-Sigma}
Let $\ket\Psi$ be a finite-dimensional pure state satisfying the exact
on-site $U(1)$ symmetry in \eqref{eq:U1-main} and approximate \Aone{}
uniformly on all local geometries used in a topology- and
orientation-preserving deformation between tripartitions $\mathcal G_0$ and
$\mathcal G_N$.  Resolve the deformation into endpoint moves, middle-region
moves, and exact complementary-region flips, as illustrated in
\cref{fig:J-deformation-moves}.  Then
\begin{equation}
\abs{\Sigma(\mathcal G_N)-\Sigma(\mathcal G_0)}
\le\sum_{j=1}^{N}E_j^\Sigma,
\label{eq:Sigma-global-main}
\end{equation}
where, for the ordered tripartition $X_j|Y_j|Z_j$ associated with move $j$,
\begin{equation}
E_j^\Sigma
:=
\Lambda_{X_j:Y_j:Z_j}\!\left(I(X_j:Z_j|Y_j)\right)
\norm{Q_j\ket\Psi}_2,
\label{eq:Sigma-one-step-error}
\end{equation}
and $Q_j$ is the corresponding quadratic regional-charge polynomial.  If
the \Aone{} violations obey \eqref{eq:A1-admissible-bound}, the on-site charges are
uniformly bounded, and region volumes and the number of moves are polynomial
in $\ell$, then
\begin{equation}
\boxed{
\abs{\Sigma(\mathcal G_N)-\Sigma(\mathcal G_0)}
\le p_\Sigma(\ell)\e^{-c_0\ell/2}.
}
\label{eq:Sigma-exp-main}
\end{equation}
for a polynomial $p_\Sigma$.
\end{theorem}

\begin{proof}
The charge and modular identities used here are collected in
\cref{fact:basic} of \cref{app:Hall}.  In particular, expanding
$Q_{BC}^2=(Q_B+Q_C)^2$, the
$Q_B^2$ and $Q_C^2$ terms have vanishing commutator expectation, so
\begin{equation}
\Sigma(A,B,C)
=i\langle[K_{AB},Q_BQ_C]\rangle.
\label{eq:Sigma-cross-main}
\end{equation}
Substituting the modular Markov defect gives the exact identity
\begin{equation}
\Sigma(A,B,C)
=i\langle[\Dlt_{A:B:C},Q_BQ_C]\rangle.
\label{eq:Hall-Markov-main}
\end{equation}
Therefore
\begin{equation}
\abs{\Sigma(A,B,C)}
\le2\Lambda_{A:B:C}\!\left(I(A:C|B)\right)
\norm{Q_BQ_C\ket\Psi}_2.
\label{eq:Hall-small-main}
\end{equation}
This identity is not intended to show that the Hall conductance estimator vanishes in the topological geometry: the CMI entering a local Markov move is associated with a buffered auxiliary tripartition, whereas the response partition itself can retain a nonzero topological value.

The exact Hall conductance estimator deformation proof uses the same patch geometries as the modular-commutator proof: endpoint and buffered middle moves are represented by \cref{fig:J-deformation-moves}(a)--(c), while junction moves are reduced by the complementary-region flips in panels (d) and (e).  Every one-step difference has the form
\begin{equation}
\Sigma(\mathcal G')-\Sigma(\mathcal G)
=-\frac{i}{2}\langle[\Dlt_m,R_m]\rangle,
\label{eq:Sigma-generic-main}
\end{equation}
where $\Dlt_m$ is a local modular Markov defect and $R_m$ is a quadratic polynomial in regional charges.  Consequently,
\begin{equation}
\abs{\Sigma(\mathcal G')-\Sigma(\mathcal G)}
\le\Lambda_m(u_m)\norm{R_m\ket\Psi}_2.
\label{eq:Sigma-step-main}
\end{equation}

The six representative charge polynomials are listed and derived in \cref{app:Hall}.  If $\norm{Q_v}_\infty\le q_0$, then $\norm{Q_X}_\infty\le q_0|X|$, so every $\norm{R_m}_\infty$ grows at most polynomially with the region sizes.  Approximate \Aone{} controls each local CMI $u_m$; summing \eqref{eq:Sigma-step-main} proves \cref{thm:main-Sigma}.
\end{proof}

\section{Stability of modular commutator and Hall conductance estimator within a quantum phase}
\label{sec:J-phase-equivalence}

The approximate deformation invariance shown in \cref{thm:main-J,thm:main-Sigma} compare different partitions of one fixed state.  To promote these results to phase invariants, one must also compare the modular responses of two different representative states.  Two additional ingredients are needed.  First, local trace-norm closeness must imply closeness of both the modular commutator and the Hall conductance estimator without any lower bound on the reduced-state spectra. Second, the phase relation must admit spatially truncated evolutions that create a bubble of one representative in the background of the other while
preserving approximate \Aone{} uniformly, including at the bubble wall; for the Hall conductance estimator, the full and truncated evolutions must additionally preserve the on-site $U(1)$ symmetry exactly. This is analagous to the topological mixed state phase equivalence definition but adapted to quasi-local unitaries \cite{yang2025topologicalmixedstatesphases}.  We formulate and prove these statements below.

\subsection{Small trace distance implies small changes of $J$ and $\Sigma$}
\label{subsec:J-trace-continuity-phase}

We first rewrite both responses as imaginary parts of inner products.  Let $\ket\psi$ be a purification of $\rho_{ABC}$.  For operators supported on $ABC$, their expectation in $\ket\psi$ equals their expectation in $\rho_{ABC}$.  Since the modular Hamiltonians are Hermitian,
\begin{align}
J(A,B,C)_\rho
&=i\left(
\braket{K_{AB}(\rho)\psi}{K_{BC}(\rho)\psi}
-
\braket{K_{BC}(\rho)\psi}{K_{AB}(\rho)\psi}
\right)
\notag\\
&=-2\,\operatorname{Im}
\braket{K_{AB}(\rho)\psi}{K_{BC}(\rho)\psi}.
\label{eq:J-quadratic-form-phase}
\end{align}
For the Hall conductance estimator, \eqref{eq:Sigma-cross-main} gives
\begin{align}
\Sigma(A,B,C)_\rho
&=i\left(
\braket{K_{AB}(\rho)\psi}{Q_BQ_C\psi}
-
\braket{Q_BQ_C\psi}{K_{AB}(\rho)\psi}
\right)
\notag\\
&=-2\,\operatorname{Im}
\braket{K_{AB}(\rho)\psi}{Q_BQ_C\psi}.
\label{eq:Sigma-quadratic-form-phase}
\end{align}
To decompose the purification $\ket{\psi}$ on $X$ and $\bar X$, we define orthonormal bases $\{\ket{i}_X\}$ and $\{\ket{\alpha}_{\bar X}\}$ and write
\[
\ket\psi=\sum_{i,\alpha}T_{i\alpha}\ket{i}_X\ket{\alpha}_{\bar X}.
\]
$T=(T_{i\alpha})$ can be viewed as a linear map from $\cH_{\bar X}$ to $\cH_X$, is called the coefficient matrix of $\ket\psi$ across the cut $X:\bar X$.  Since $\rho_X=TT^\dagger$, we can identify the vector norm of $\ket\psi$ with the Hilbert--Schmidt norm of $T$.

\begin{theorem}[Spectrum-independent continuity of a modular vector]
\label{thm:modular-vector-continuity-phase}
Let $\ket\psi$ and $\ket\phi$ be normalized vectors on the same bipartite
Hilbert space $X\bar X$, with reduced states $\rho_X$ and $\sigma_X$. Denote
\begin{equation}
  \delta:=\norm{\ket\psi-\ket\phi}_2\le2.
\end{equation}
and
\begin{equation}
  \vartheta_X(\delta)
  :=
  \begin{cases}
    0, & \delta=0,\\[0.2em]
    \delta\left[\log d_X+2\log(1/\delta)+3\right],
    &0<\delta\le2.
  \end{cases}
  \label{eq:vartheta-J-phase}
\end{equation}
Then
\begin{equation}
  \norm{K_X(\rho)\ket\psi-K_X(\sigma)\ket\phi}_2
  \le \vartheta_X(\delta).
  \label{eq:modular-vector-continuity-phase}
\end{equation}
No minimum-eigenvalue bound is assumed.
\end{theorem}

\begin{proof}
Identify $\ket\psi$ and $\ket\phi$ across the cut $X:\bar X$ with coefficient matrices $T$ and $S$ with the same orthonormal basis on $X$ and $\bar X$.  Then
\begin{equation}
  TT^\dagger=\rho_X,
  \qquad
  SS^\dagger=\sigma_X,
  \qquad
  \norm{T-S}_2=\delta.
  \label{eq:coefficient-matrices-phase}
\end{equation}
For $L\ge2$, define the truncated logarithm
\begin{equation}
  k_L(p):=\min\{-\log p,L\},
  \qquad 0\le p\le1,
\end{equation}
and
\begin{equation}
  F_L(T):=k_L(TT^\dagger)T.
\end{equation}
If $s\in[0,1]$ is a singular value of $T$, the associated scalar singular-value
function is
\begin{equation}
  f_L(s):=s\min\{-2\log s,L\},
  \qquad f_L(0):=0.
\end{equation}
The two auxiliary lemmas proved in
\cref{lem:fL-Lipschitz-app,lem:truncated-modular-vector-app} show that the
odd extension of $f_L$ is $L$-Lipschitz and that the induced map $F_L$ is
also $L$-Lipschitz in the Hilbert--Schmidt norm, namely,
\begin{equation}
  \norm{F_L(T)-F_L(S)}_2\le L\norm{T-S}_2=L\delta.
  \label{eq:truncated-vector-Lipschitz-phase}
\end{equation}

It remains to control the discarded logarithmic tail.  If $p_i$ are the
eigenvalues of $\rho_X$:
\begin{align}
  &\norm{\bigl(K_X(\rho)-k_L(\rho_X)\bigr)\ket\psi}_2^2 =
  \sum_{p_i<\e^{-L}}p_i\bigl(-\log p_i-L\bigr)^2
  \le \frac{4}{\e^2}d_X\e^{-L}.
  \label{eq:modular-tail-phase}
\end{align}
The same bound holds for $\sigma_X$ and $\ket\phi$.  For $\delta>0$, choose
\begin{equation}
  L=\log d_X+2\log(1/\delta)+2.
\end{equation}
Since $d_X\ge4$ and $\delta\le2$, this choice satisfies $L\ge2$.
Each tail norm in \eqref{eq:modular-tail-phase} is then at most
$2\e^{-2}\delta$. 
Under the coefficient-matrix identification, we have
\begin{equation}
  \norm{k_L(\rho_X)\ket\psi-k_L(\sigma_X)\ket\phi}_2
  =
  \norm{F_L(T)-F_L(S)}_2
  \le L\delta.
\end{equation}
The triangle inequality, followed by the two tail bounds, therefore gives
\begin{align}
  \norm{K_X(\rho)\ket\psi-K_X(\sigma)\ket\phi}_2
  &\le
  \norm{\bigl(K_X(\rho)-k_L(\rho_X)\bigr)\ket\psi}_2
  \notag\\
  &\quad+
  \norm{k_L(\rho_X)\ket\psi-k_L(\sigma_X)\ket\phi}_2
  \notag\\
  &\quad+
  \norm{\bigl(k_L(\sigma_X)-K_X(\sigma)\bigr)\ket\phi}_2
  \notag\\
  &\le
  2\e^{-2}\delta+L\delta+2\e^{-2}\delta
  \notag\\
  &\le\delta(L+1)
  \notag\\
  &=\delta\left[\log d_X+2\log(1/\delta)+3\right],
\end{align}
where $4\e^{-2}<1$ is used in the third inequality. The case $\delta=0$ is immediate.
\end{proof}

\begin{theorem}[Trace-norm continuity of the modular commutator]
\label{thm:J-trace-continuity-phase}
Let $\rho_{ABC}$ and $\sigma_{ABC}$ be arbitrary finite-dimensional density matrices and set
\begin{equation}
  \eps:=\frac12\norm{\rho_{ABC}-\sigma_{ABC}}_1,
  \qquad
  \Omega_J(\delta;A,B,C)
  :=2\log d_{BC}\,\vartheta_{AB}(\delta)
  +2\log d_{AB}\,\vartheta_{BC}(\delta).
  \label{eq:trace-distance-phase}
\end{equation}
Then
\begin{equation}
  \abs{J(A,B,C)_\rho-J(A,B,C)_\sigma}
  \le
  \Omega_J\!\left(\sqrt{2\eps};A,B,C\right).
  \label{eq:J-trace-continuity-phase}
\end{equation}
In particular, for some universal constant $C=24\sqrt{2}$,
\begin{equation}
  \abs{J(A,B,C)_\rho-J(A,B,C)_\sigma}
  \le
  C\sqrt\eps
  \left[
    1+\log(d_{AB}d_{BC})+\log\frac1\eps
  \right]^2.
  \label{eq:J-trace-continuity-rough-phase}
\end{equation}
\end{theorem}

\begin{proof}
By Uhlmann's theorem and the Fuchs--van de Graaf inequality, as shown in \cref{lem:purification-trace-distance-app}, there are purifications $\ket\psi$ of $\rho_{ABC}$ and $\ket\phi$ of $\sigma_{ABC}$ on a common auxiliary system such that
\begin{equation}
  \norm{\ket\psi-\ket\phi}_2\le\sqrt{2\eps}.
\end{equation}
Write
\begin{equation}
  v_X(\rho):=K_X(\rho)\ket\psi,
  \qquad
  v_X(\sigma):=K_X(\sigma)\ket\phi.
\end{equation}
Using \eqref{eq:J-quadratic-form-phase}, we have:
\begin{align}
  J_\rho-J_\sigma
  &=-2\operatorname{Im}
  \left[
    \braket{v_{AB}(\rho)}{v_{BC}(\rho)}
    -
    \braket{v_{AB}(\sigma)}{v_{BC}(\sigma)}
  \right]
  \notag\\
  &=-2\operatorname{Im}
  \left[
    \braket{v_{AB}(\rho)-v_{AB}(\sigma)}
            {v_{BC}(\rho)}
    +
    \braket{v_{AB}(\sigma)}
            {v_{BC}(\rho)-v_{BC}(\sigma)}
  \right].
  \label{eq:J-inner-product-difference-phase}
\end{align}
The second equality is obtained by adding and subtracting
$\braket{v_{AB}(\sigma)}{v_{BC}(\rho)}$.  Since
$\abs{\operatorname{Im}z}\le\abs z$, the triangle inequality and Cauchy--Schwarz give
\begin{align}
  \abs{J_\rho-J_\sigma}
  &\le
  2\norm{v_{AB}(\rho)-v_{AB}(\sigma)}_2
   \norm{v_{BC}(\rho)}_2
  +
  2\norm{v_{AB}(\sigma)}_2
   \norm{v_{BC}(\rho)-v_{BC}(\sigma)}_2.
  \label{eq:J-vector-difference-phase}
\end{align}
The modular second-moment bound \eqref{eq:universal-modular-moment-main} gives $\norm{v_X(\rho)}_2,\ \norm{v_X(\sigma)}_2\le\log d_X$ and \cref{thm:modular-vector-continuity-phase} controls the two differences in \eqref{eq:J-vector-difference-phase}. This proves \eqref{eq:J-trace-continuity-phase}.

For the rough form \eqref{eq:J-trace-continuity-rough-phase}, the case $\eps=0$ is immediate, so assume $\eps>0$ and
set
\begin{equation}
  M:=1+\log(d_{AB}d_{BC})+\log(1/\eps).
\end{equation}
For $X=AB,BC$, one has
\begin{equation}
  \log d_X\le M.
\end{equation}
Moreover,
\begin{align}
  \vartheta_X\!\left(\sqrt{2\eps}\right)
  &=\sqrt{2\eps}
    \left[
      \log d_X+\log\frac{1}{2\eps}+3
    \right]
\le 3\sqrt{2\eps}\,M.
\end{align}
Substituting these two bounds into
\eqref{eq:J-trace-continuity-phase} gives
\begin{equation}
  \abs{J_\rho-J_\sigma}
  \le 24\sqrt{2}\,\sqrt{\eps}\,M^2,
\end{equation}
which proves \eqref{eq:J-trace-continuity-rough-phase}.
\end{proof}

We next prove the corresponding statement for the Hall conductance estimator.

\begin{theorem}[Trace-norm continuity of the Hall conductance estimator]
\label{thm:Sigma-trace-continuity-phase}
Let $\rho_{ABC}$ and $\sigma_{ABC}$ be arbitrary finite-dimensional density
matrices, and assume $\norm{Q_v}_\infty\le q_0$ for every site $v$.  Set
\begin{equation}
  \eps:=\frac12\norm{\rho_{ABC}-\sigma_{ABC}}_1\le\frac12.
  \label{eq:Sigma-trace-distance-phase}
\end{equation}
For $0\le\delta\le2$, define
\begin{equation}
  \Omega_\Sigma(\delta;A,B,C)
  :=2\norm{Q_BQ_C}_\infty
  \left[
    \vartheta_{AB}(\delta)+\log d_{AB}\,\delta
  \right].
  \label{eq:Omega-Sigma-phase}
\end{equation}
Then
\begin{equation}
  \abs{\Sigma(A,B,C)_\rho-\Sigma(A,B,C)_\sigma}
  \le
  \Omega_\Sigma\!\left(\sqrt{2\eps};A,B,C\right).
  \label{eq:Sigma-trace-continuity-phase}
\end{equation}
In particular, for some constant $C=10\sqrt{2}$,
\begin{equation}
  \abs{\Sigma(A,B,C)_\rho-\Sigma(A,B,C)_\sigma}
  \le
  Cq_0^2|B||C|\sqrt\eps
  \left[
    1+\log d_{AB}+\log\frac1\eps
  \right].
  \label{eq:Sigma-trace-continuity-rough-phase}
\end{equation}
\end{theorem}

\begin{proof}
By Uhlmann's theorem and the Fuchs--van de Graaf inequality, as shown in \cref{lem:purification-trace-distance-app}, there exist purifications $\ket\psi$ and $\ket\phi$ on a common auxiliary system such that
\begin{equation}
  \norm{\ket\psi-\ket\phi}_2\le\sqrt{2\eps},
\end{equation}
where $\ket\psi$ purifies $\rho_{ABC}$ and $\ket\phi$ purifies
$\sigma_{ABC}$.  Define
\begin{equation}
  v_\rho:=K_{AB}(\rho)\ket\psi,
  \qquad
  v_\sigma:=K_{AB}(\sigma)\ket\phi.
\end{equation}
Using \eqref{eq:Sigma-quadratic-form-phase} and following the derivation of the $J(A,B,C)$ counterpart:
\begin{align}
  \abs{\Sigma_\rho-\Sigma_\sigma}
  &\le
  2\abs{
    \braket{v_\rho}{Q_BQ_C\psi}
    -\braket{v_\sigma}{Q_BQ_C\phi}
  }
  \notag\\
  &\le
  2\norm{v_\rho-v_\sigma}_2\norm{Q_BQ_C\ket\psi}_2
  +2\norm{v_\sigma}_2
  \norm{Q_BQ_C(\ket\psi-\ket\phi)}_2.
  \label{eq:Sigma-vector-difference-phase}
\end{align}
By \cref{thm:modular-vector-continuity-phase} and the universal modular second-moment bound
\begin{equation}
  \norm{v_\rho-v_\sigma}_2
  \le\vartheta_{AB}\!\left(\sqrt{2\eps}\right), \quad \norm{v_\sigma}_2\le\log d_{AB}.
\end{equation}
Moreover,
\begin{equation}
\norm{Q_BQ_C}_\infty
\le\norm{Q_B}_\infty\norm{Q_C}_\infty
\le q_0^2|B||C|.
\end{equation}
Hence
\begin{equation}
\norm{Q_BQ_C\ket\psi}_2\le q_0^2|B||C|,
\qquad
\norm{Q_BQ_C(\ket\psi-\ket\phi)}_2
\le q_0^2|B||C|\sqrt{2\eps}.
\end{equation}
Substitution into \eqref{eq:Sigma-vector-difference-phase} proves
\eqref{eq:Sigma-trace-continuity-phase}.

For the rough form, the case $\eps=0$ is immediate, so assume $\eps>0$ and
set
\begin{equation}
  M_\Sigma:=1+\log d_{AB}+\log(1/\eps).
\end{equation}
As in the modular-commutator bound,
\begin{equation}
  \log d_{AB}\le M_\Sigma, \quad \vartheta_{AB}\!\left(\sqrt{2\eps}\right)
  \le3\sqrt{2\eps}\,M_\Sigma.
\end{equation}
Substituting these bounds into \eqref{eq:Sigma-trace-continuity-phase} gives
\begin{align}
  \abs{\Sigma_\rho-\Sigma_\sigma}
  &\le
  2q_0^2|B||C|
  \left[
    3\sqrt{2\eps}\,M_\Sigma
    +2\sqrt{2\eps}\,M_\Sigma
  \right]
  \notag\\
  &=10\sqrt{2}\,
  q_0^2|B||C|\sqrt{\eps}\,M_\Sigma.
\end{align}
\end{proof}

\subsection{Quasi-local unitaries, \Aone{} preservation, and symmetry-protected phase equivalence}
\label{subsec:QLU-phase-equivalence}

We now specify the quasi-local unitary paths used to compare the admissible state families introduced in\cref{def:A1-admissible-state,def:U1-A1-admissible-state-phase}.
We follow the time-dependent interaction and propagator conventions in \cite[Sec.~4.2, Eqs.~(4.16)--(4.19)]{BachmannMichalakisNachtergaeleSims2012}. We use their interaction norm with polynomial decay functions of arbitrarily high degree, uniformly in the finite volume; this specifies the faster-than-any-power class needed below.

\begin{definition}[Quasi-local unitary path]
\label{def:QLU-path-phase}
Let $\Gamma$ be a metric lattice with uniformly polynomial volume growth,
$\sup_{x\in\Gamma}|B_r(x)|\le c(1+r)^{d_{\mathrm g}}$, with fixed volume-growth exponent $d_{\mathrm g}$, and let $\Lambda\subset\Gamma$
range over finite volumes.  A family of unitary paths
\begin{equation}
  U_\Lambda(s),
  \qquad 0\le s\le1,
  \qquad U_\Lambda(0)=\id,
\end{equation}
is a \emph{quasi-local unitary path} if
\begin{equation}
  \partial_sU_\Lambda(s)
  =-iG_\Lambda(s)U_\Lambda(s),
  \qquad
  G_\Lambda(s)=\sum_{Z\subseteq\Lambda}\Phi_{\Lambda,s}(Z),
  \label{eq:QLU-generator-phase}
\end{equation}
where $\Phi_{\Lambda,s}(Z)=\Phi_{\Lambda,s}(Z)^\dagger$ is supported on $Z$
and norm-continuous in $s$.  For every $p>d_{\mathrm g}$, set $F_p(r)=(1+r)^{-p}$
and require
\begin{equation}
  \norm{\Phi}_{F_p}
  :=\sup_\Lambda\sup_{x,y\in\Lambda}
  \frac{1}{F_p(\dist(x,y))}
  \sum_{\substack{Z\subseteq\Lambda\\x,y\in Z}}
  \sup_{s\in[0,1]}\norm{\Phi_{\Lambda,s}(Z)}
  <\infty.
  \label{eq:QLU-F-norm-phase}
\end{equation}
This is the finite-volume-uniform version of the interaction norm in
Ref.~\cite[Sec.~4.1, Eq.~(4.7)]{BachmannMichalakisNachtergaeleSims2012}.
The endpoint $U_\Lambda(1)$ is called a quasi-local unitary.
We write $U_\Lambda(s,t)=U_\Lambda(s)U_\Lambda(t)^\dagger$; the associated
Heisenberg evolution is $O\mapsto U_\Lambda(s,t)^\dagger O U_\Lambda(s,t)$.
\end{definition}

Since \eqref{eq:QLU-F-norm-phase} holds for arbitrarily large $p$, polynomial
volume growth gives, for every $n\in\mathbb N$,
\begin{equation}
  \sup_{\Lambda}\sup_{s\in[0,1]}\sup_{x\in\Lambda}
  \sum_{Z\ni x}
  |Z|(1+\diam Z)^n\norm{\Phi_{\Lambda,s}(Z)}
  <\infty.
  \label{eq:QLU-almost-local-phase}
\end{equation}
For the argument below, we only need spatial decay uniform in
$s,t\in[0,1]$.  Fix an integer $n\ge1$ and choose
$p>\max\{d_{\mathrm g},n\}$.  The time-dependent Lieb--Robinson bound of
Ref.~\cite[Sec.~4.2, Theorem~4.6]{BachmannMichalakisNachtergaeleSims2012}
contains an exponential factor in $|s-t|$ and the spatial sum
$\sum_{x\in X,y\in Y}F_p(\dist(x,y))$.
Since $|s-t|\le1$, the time factor is absorbed into a constant, while the
sum is at most $|X||Y|(1+\dist(X,Y))^{-p}$.  Consequently,
\begin{equation}
  \norm{[U_\Lambda(s,t)^\dagger O_XU_\Lambda(s,t),O_Y]}
  \le
  C_n|X||Y|\norm{O_X}\norm{O_Y}
  \dist(X,Y)^{-n},
  \label{eq:QLU-LR-phase}
\end{equation}
with $C_n$ independent of $\Lambda$, $s,t\in[0,1]$, and the supports.
Throughout the locality estimates, $n$ denotes an arbitrary positive
integer decay exponent.  Inverse powers of distance are understood as
$+\infty$ at zero distance, where the bounds are vacuous.

The canonical example is quasi-adiabatic continuation along a uniformly
gapped path of local Hamiltonians.  It produces a quasi-local unitary path
that transports the ground-space projector and therefore connects the ground
spaces of the endpoint Hamiltonians; for nondegenerate ground states, it maps
one ground-state vector to the other up to a phase
\cite{Hasting_quasiadiabaticcontinuation_2005,
hastings2010localityquantumsystems,BachmannMichalakisNachtergaeleSims2012}.

For a region $D\subseteq\Lambda$, define the restricted generator and its
unitary path by
\begin{align}
  G_{\Lambda,D}(s)
  &:=\sum_{Z\subseteq D}\Phi_{\Lambda,s}(Z),
  \label{eq:QLU-restricted-generator-phase}\\
  \partial_sU_{\Lambda,D}(s)
  &:=-iG_{\Lambda,D}(s)U_{\Lambda,D}(s),
  \qquad U_{\Lambda,D}(0)=\id.
  \label{eq:QLU-restricted-unitary-phase}
\end{align}

\begin{lemma}[Automatic spatial truncation]
\label{lem:QLU-spatial-truncation-phase}
For every $n\in\mathbb N$, there is a constant $C_n$, independent of
$\Lambda$, $D$, and $s$, such that, for $X\subseteq D$,
\begin{equation}
  \norm{
    U_\Lambda(s)^\dagger O_XU_\Lambda(s)
    -U_{\Lambda,D}(s)^\dagger O_XU_{\Lambda,D}(s)
  }
  \le
  C_n|X|\norm{O_X}
  \dist(X,D^c)^{-n}.
  \label{eq:QLU-interior-truncation-phase}
\end{equation}
If
\begin{align}
  \rho_{\Lambda,s}
  :=U_\Lambda(s)\rho_\Lambda U_\Lambda(s)^\dagger,\qquad
  \omega_{\Lambda,D,s}
  :=U_{\Lambda,D}(s)\rho_\Lambda U_{\Lambda,D}(s)^\dagger,
\end{align}
then
\begin{equation}
  \frac12\norm{
    (\omega_{\Lambda,D,s})_X-(\rho_{\Lambda,s})_X
  }_1
  \le
  C_n|X|
  \dist(X,D^c)^{-n},
  \qquad X\subseteq D,
  \label{eq:QLU-interior-state-phase}
\end{equation}
and
\begin{equation}
  (\omega_{\Lambda,D,s})_Y=(\rho_\Lambda)_Y, 
  \qquad U_{\Lambda,D}(s)^\dagger O_YU_{\Lambda,D}(s)=O_Y,
  \qquad Y\subseteq D^c.
  \label{eq:QLU-exterior-state-phase}
\end{equation}
\end{lemma}

\begin{proof}
Let $U(s,t)$ and $U_D(s,t)$ be the full and restricted propagators. By Duhamel formula:
\begin{align}
  &U(s,0)^\dagger O_XU(s,0)
  -U_D(s,0)^\dagger O_XU_D(s,0)
  \notag\\
  &\quad=
  i\int_0^s
  U(t,0)^\dagger
  \bigl[
    G(t)-G_D(t),
    U_D(s,t)^\dagger O_XU_D(s,t)
  \bigr]
  U(t,0)\,dt.
  \label{eq:QLU-Duhamel-phase}
\end{align}
 where $G(t)-G_D(t)$ is the sum of interaction terms $\Phi_{\Lambda,t}(Z)$ with $Z\not\subseteq D$.  Set $d:=\dist(X,D^c)$.  Since the restricted interaction is obtained by deleting terms from the full interaction, it satisfies the same almost-local norm bound as in \eqref{eq:QLU-almost-local-phase}; hence
$U_D(s,t)$ obeys the same type of Lieb--Robinson bound, uniformly in $D$.

For every omitted support $Z\not\subseteq D$, the elementary geometric bound
\begin{equation*}
  d\le\dist(X,Z)+\diam Z
\end{equation*}
shows that either $Z$ is at least distance $d/2$ from $X$, or its diameter is
at least $d/2$.  In the first case, the Lieb--Robinson bound suppresses its
commutator with $U_D(s,t)^\dagger O_XU_D(s,t)$.  In the second case, the
large-diameter tail in \eqref{eq:QLU-almost-local-phase} suppresses the
interaction term itself.  Summing the two contributions, using bounded
geometry and choosing the available moment exponent larger than $n$, gives
\begin{align*}
  &\sum_{Z\not\subseteq D}
  \norm{[
    \Phi_{\Lambda,t}(Z),
    U_D(s,t)^\dagger O_XU_D(s,t)
  ]}\le
  C_n|X|\norm{O_X}
  d^{-n},
\end{align*}
uniformly for $0\le t\le s\le1$.  Unitary conjugation by $U(t,0)$ does not
change the norm.  Integrating \eqref{eq:QLU-Duhamel-phase} over an interval
of length at most one proves
\eqref{eq:QLU-interior-truncation-phase}.

For the state bound, trace-norm duality rewrites the left-hand side of \eqref{eq:QLU-interior-state-phase} as the supremum, over $\norm{O_X}\le1$, of the difference between the full and restricted Heisenberg evolutions of $O_X$.  Equation \eqref{eq:QLU-interior-truncation-phase} therefore proves \eqref{eq:QLU-interior-state-phase}, after absorbing the factor $1/2$ into $C_n$.  Subsequently, $U_D(s)$ commutes with every observable supported in $D^c$, so the expectation values of all observables on $Y\subseteq D^c$ are unchanged.  This proves \eqref{eq:QLU-exterior-state-phase}.
\end{proof}

\begin{definition}[$U(1)$-symmetric quasi-local unitary path]
\label{def:U1-symmetric-QLU-phase}
A quasi-local unitary path in the sense of
\cref{def:QLU-path-phase} is \emph{$U(1)$-symmetric} if its almost-local
interaction decomposition can be chosen so that every interaction term
conserves the charge on its support:
\begin{equation}
  [\Phi_{\Lambda,s}(Z),Q_Z]=0,
  \qquad
  Q_Z:=\sum_{v\in Z}Q_v,
  \label{eq:termwise-U1-symmetry-phase}
\end{equation}
for every $\Lambda$, $s$, and $Z\subseteq\Lambda$.
\end{definition}

The termwise formulation is needed because the bubble construction truncates
the generator spatially.  It implies
\begin{equation}
  [G_\Lambda(s),Q_\Lambda]=0,
  \qquad
  [G_{\Lambda,D}(s),Q_D]=0,
  \label{eq:full-and-truncated-generator-symmetry-phase}
\end{equation}
and therefore
\begin{equation}
  [U_\Lambda(s),Q_\Lambda]=0,
  \qquad
  [U_{\Lambda,D}(s),Q_\Lambda]=0.
  \label{eq:full-and-truncated-unitary-symmetry-phase}
\end{equation}
For the second identity, use that $U_{\Lambda,D}(s)$ is supported on $D$ and
commutes with $Q_D$.  Thus every spatially truncated bubble state produced
from a sharp-charge state remains in the same total-charge sector.

\begin{definition}[\Aone-preserving quasi-local unitary]
\label{def:A1-preserving-QLU-phase}
A quasi-local unitary path $U_\Lambda(s)$ is \emph{\Aone-preserving} if, for
every uniformly \Aone-admissible state family $\boldsymbol\rho$, every region $D\subseteq\Lambda$ including $D=\Lambda$, and every $s\in[0,1]$, both
\begin{equation}
  U_{\Lambda,D}(s)\rho_\Lambda U_{\Lambda,D}(s)^\dagger
  \quad\text{and}\quad
  U_{\Lambda,D}(s)^\dagger\rho_\Lambda U_{\Lambda,D}(s)
  \label{eq:A1-preserving-two-directions-phase}
\end{equation}
are uniformly \Aone-admissible, including on geometries that cross the truncation wall.\footnote{The endpoint comparison below only requires this admissibility condition at $s=0$ and $s=1$.  We retain the full interval $s\in[0,1]$ to match the quasi-adiabatic picture: along a uniformly gapped Hamiltonian path, the transported ground state at every intermediate $s$ belongs to the same gapped quantum phase as the initial state at $s=0$.}
\end{definition}

\begin{definition}[$U(1)$-symmetric and \Aone-preserving quasi-local unitary]
\label{def:U1-A1-preserving-QLU-phase}
A quasi-local unitary path is \emph{$U(1)$-symmetric and
\Aone-preserving} if it satisfies both
\cref{def:U1-symmetric-QLU-phase,def:A1-preserving-QLU-phase}.  In
particular, every full or spatially truncated intermediate state obtained
from a uniformly $U(1)$-symmetric \Aone-admissible representative remains
exactly $U(1)$ symmetric and uniformly \Aone-admissible, including on
geometries that cross the truncation wall.
\end{definition}

The endpoint admissibility requirement in
\cref{def:A1-preserving-QLU-phase,def:U1-A1-preserving-QLU-phase}
excludes states with persistent spurious contributions that violate uniform
\Aone{}.  In particular, the cluster-based constructions underlying spurious
modular commutators \cite{Gass_2024} also produce nondecaying \Aone{} defects:
\cref{app:spurious-A1} gives $\delta_1=\log2$ for the ordinary cluster chain,
and \cref{app:modified-cluster-A1} calculates a strictly positive,
size-independent defect for the modified chain with nonzero spurious $J$.
Thus these examples fail the small-\Aone{} requirement on the endpoint, thus is excluded from the quantum states we consider for phase equivalence.
\footnote{We do
not prove a general statement that \Aone{} preservation along the gapped path in the bulk alone rules out every
possible spurious contribution.  The explicit cluster-chain counterexamples
discussed here can nevertheless be excluded by their persistent \Aone{}
defects; this is not a classification of all possible counterexamples. More details will be discussed in Ref.~\cite{Bowen2025}.}

Small \Aone{} violation on geometries crossing the truncation wall is a
\emph{separate assumption}, not a consequence of uniform bulk \Aone{} along
the full state trajectory.  As shown in \cref{app:boundary-A1-counterexample},
a finite-range commuting generator can leave an initial product state
unchanged for every $s\in[0,1]$, while spatial truncation removes cancelling
terms and creates a cluster chain along the cut.  Its wall-crossing
\Aone{} defect is $2\log2$ at arbitrarily large annular width.  Quasi-locality
controls the truncation error away from the wall, but does not supply the
boundary admissibility required by the bubble construction.

\begin{definition}[Phase equivalence from entanglement bootstrap]
\label{def:A1-phase-equivalence}
Two uniformly \Aone-admissible representative-state families
$\boldsymbol\rho^\alpha$ and $\boldsymbol\rho^\beta$ are in the same
\emph{entanglement-bootstrap phase}, written
\begin{equation}
  \boldsymbol\rho^\alpha\sim_{A_1}\boldsymbol\rho^\beta,
\end{equation}
if there exists a single \Aone-preserving quasi-local unitary path
$U_\Lambda(s)$ such that
\begin{equation}
  \rho_\Lambda^\beta
  =U_\Lambda(1)\rho_\Lambda^\alpha U_\Lambda(1)^\dagger
  \label{eq:A1-phase-endpoint-phase}
\end{equation}
for every $\Lambda$.
\end{definition}

\begin{definition}[$U(1)$-protected phase equivalence]
\label{def:U1-A1-phase-equivalence}
Two uniformly $U(1)$-symmetric \Aone-admissible representative-state
families $\boldsymbol\rho^\alpha$ and $\boldsymbol\rho^\beta$, defined with
respect to the same on-site charge representation, are in the same
\emph{$U(1)$-protected entanglement-bootstrap phase}, written
\begin{equation}
  \boldsymbol\rho^\alpha
  \sim_{A_1,U(1)}
  \boldsymbol\rho^\beta,
\end{equation}
if there exists a single $U(1)$-symmetric and \Aone-preserving quasi-local
unitary path $U_\Lambda(s)$ such that
\begin{equation}
  \rho_\Lambda^\beta
  =U_\Lambda(1)\rho_\Lambda^\alpha U_\Lambda(1)^\dagger
  \label{eq:U1-A1-phase-endpoint-phase}
\end{equation}
for every $\Lambda$.
\end{definition}

\subsection{Phase stability of the modular commutator and Hall conductance estimator}
\label{subsec:J-phase-stability-final}

We now combine the trace-norm continuity theorem with the phase equivalence definition in \cref{def:A1-phase-equivalence} to show the stability of the modular commutator within a topological phase.  For a topology-preserving
deformation $\mathcal G\to\mathcal G'$ in a state $\rho$, denote by
\begin{equation}
  \mathcal E_J(\rho;\mathcal G\to\mathcal G')
\end{equation}
the right-hand side of \eqref{eq:J-global-main}.  By \cref{thm:main-J}, if the
relevant \Aone{} defects satisfy \eqref{eq:A1-admissible-bound}, then
\begin{equation}
  \mathcal E_J(\rho;\mathcal G\to\mathcal G')
  \le p_J(\ell)\e^{-c_0\ell/2}
  \label{eq:EJ-exponential-phase}
\end{equation}
whenever the geometry and the number of elementary moves grow polynomially in
$\ell$.  The comparison geometry used below is shown schematically in
\cref{fig:J-bubble-comparison}.

\begin{figure}[H]
\centering
\resizebox{0.42\textwidth}{!}{%
\begin{tikzpicture}[x=1cm, y=1cm, line cap=round, line join=round]

  \begin{scope}
    \clip[rounded corners=20pt] (0,0) rectangle (9,9);
    
    \fill[white] (0,0) rectangle (2.5,9);
    
    \shade[left color=white, right color=red!14] (2.5,0) rectangle (6.8,9);
    
    \fill[red!14] (6.8,0) rectangle (9,9);
  \end{scope}

  \draw[line width=3.5pt, rounded corners=20pt] (0,0) rectangle (9,9);

  \newcommand{\defect}[2]{%
    \begin{scope}[shift={(#1,#2)}, scale=0.7]
      \draw[thick, fill=orange!15] 
        (0,-0.6) 
        to[out=60, in=-120] (0.2,-0.1)
        to[out=60, in=-60] (0.05,0.4)
        to[out=120, in=-90] (-0.05,0.7)
        to[out=-70, in=110] (0.08,0.25)
        to[out=-70, in=60] (-0.1,-0.25)
        to[out=-120, in=80] (0,-0.6) -- cycle;
    \end{scope}
  }


  \draw[line width=4.5pt, draw={rgb,255:red,13;green,110;blue,220},
        postaction={decorate},
        decoration={markings,
          mark=at position 0.55 with {\arrow{Stealth[length=5mm,width=4.2mm]}}
        },
        ->, >={Stealth[length=6.5mm, width=5.5mm]}]
    (2.9, 2.3) to[out=85, in=240] (6.0, 5.7);

  \newcommand{\diskPartition}[2]{%
    \begin{scope}[shift={(#1,#2)}]
      \def\R{1.35}
      \fill[blue!12]   (0,0) -- (30:\R)  arc (30:150:\R)  -- cycle;
      \fill[red!12]    (0,0) -- (150:\R) arc (150:270:\R) -- cycle;
      \fill[yellow!22] (0,0) -- (270:\R) arc (270:390:\R) -- cycle;

      \draw[line width=1.5pt] (0,0) -- (30:\R);
      \draw[line width=1.5pt] (0,0) -- (150:\R);
      \draw[line width=1.5pt] (0,0) -- (270:\R);
      \draw[line width=2.4pt] (0,0) circle (\R);

      \node at (90:0.76)  {\large $C$};
      \node at (210:0.76) {\large $A$};
      \node at (330:0.76) {\large $B$};
    \end{scope}
  }

  \diskPartition{1.8}{1.8}
  \diskPartition{7.2}{7.2}

  \node[font=\large, fill=white, fill opacity=0.8, text opacity=1,
        rounded corners=2pt, inner sep=2pt]
    at (1.8,3.55) {$\mathcal G_{\mathrm{out}}$};
  \node[font=\large, fill=red!7, fill opacity=0.85, text opacity=1,
        rounded corners=2pt, inner sep=2pt]
    at (7.2,5.35) {$\mathcal G_{\mathrm{in}}$};

\end{tikzpicture}
}
\caption{Bubble comparison used in
\cref{thm:J-phase-stability-final}.  The partition is deformed from
$\mathcal G_{\mathrm{out}}$, where the bubble state agrees with
$\rho^\alpha$, to $\mathcal G_{\mathrm{in}}$, where it is locally close to
$\rho^\beta$.  These two comparisons produce the errors $\mathcal E_J$ and
$\Omega_J$, respectively.}
\label{fig:J-bubble-comparison}
\end{figure}
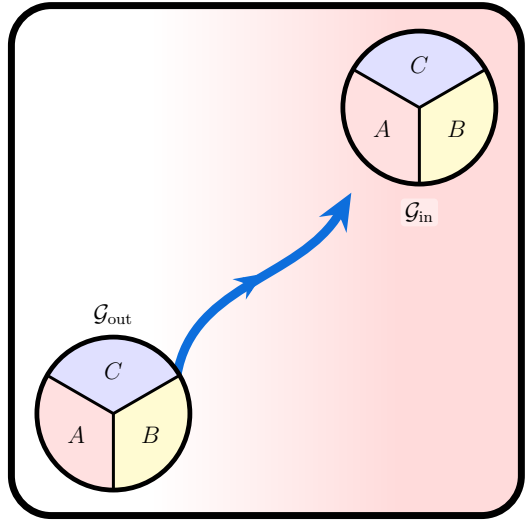

\begin{theorem}[Stability of $J$ under \Aone-preserving quasi-local equivalence]
\label{thm:J-phase-stability-final}
Suppose
$\boldsymbol\rho^\alpha\sim_{A_1}\boldsymbol\rho^\beta$ through the
\Aone-preserving quasi-local unitary path $U_\Lambda(s)$.  Let
$D\subseteq\Lambda$ be a sufficiently large disk and define the endpoint
bubble state
\begin{equation}
  \omega_{\Lambda,D}
  :=U_{\Lambda,D}(1)\rho_\Lambda^\alpha
    U_{\Lambda,D}(1)^\dagger.
  \label{eq:endpoint-bubble-state-phase}
\end{equation}
Choose topology- and orientation-equivalent response partitions
$\mathcal G_{\mathrm{out}}$ and $\mathcal G_{\mathrm{in}}$ of characteristic
scale $\ell$ such that their total supports satisfy
\begin{equation}
  R_{\mathrm{out}}\subseteq D^c,
  \qquad
  R_{\mathrm{in}}\subseteq D,
  \qquad
  r:=\dist(R_{\mathrm{in}},D^c).
  \label{eq:inside-outside-partitions-phase}
\end{equation}
Assume that $\mathcal G_{\mathrm{out}}$ can be deformed into
$\mathcal G_{\mathrm{in}}$ through admissible local moves.  Define
\begin{equation}
  \eta_U(r,R_{\mathrm{in}})
  :=C_n|R_{\mathrm{in}}|
  r^{-n},
  \label{eq:eta-U-phase}
\end{equation}
where $n$ may be chosen arbitrarily and $C_n$ is the constant in
\eqref{eq:QLU-interior-state-phase}.  For $r$ large enough that
$\eta_U\le1/2$,
\begin{equation}
  \begin{aligned}
  &\abs{
    J(\mathcal G_{\mathrm{in}})_{\rho^\beta}
    -J(\mathcal G_{\mathrm{out}})_{\rho^\alpha}
  }
  \le
  \Omega_J\!\left(
    \sqrt{2\eta_U(r,R_{\mathrm{in}})};\mathcal G_{\mathrm{in}}
  \right)
  +\mathcal E_J\!\left(
    \omega_{\Lambda,D};
    \mathcal G_{\mathrm{out}}\to\mathcal G_{\mathrm{in}}
  \right).
  \end{aligned}
  \label{eq:J-bubble-stability-phase}
\end{equation}
Here $\Omega_J(\delta;\mathcal G)$ is the modulus in
\eqref{eq:trace-distance-phase} evaluated on the tripartition $\mathcal G$.
\end{theorem}

\begin{proof}
Because $U_{\Lambda,D}(1)$ is supported on $D$,
\eqref{eq:QLU-exterior-state-phase} gives
\begin{equation}
  (\omega_{\Lambda,D})_{R_{\mathrm{out}}}
  =(\rho_\Lambda^\alpha)_{R_{\mathrm{out}}},
\end{equation}
and therefore
\begin{equation}
  J(\mathcal G_{\mathrm{out}})_{\omega_{\Lambda,D}}
  =J(\mathcal G_{\mathrm{out}})_{\rho^\alpha}.
  \label{eq:J-outside-exact-phase}
\end{equation}
On the other hand, the full endpoint of the quasi-local path is
$\rho_\Lambda^\beta$.  Hence \eqref{eq:QLU-interior-state-phase} gives
\begin{equation}
  \frac12\norm{
    (\omega_{\Lambda,D})_{R_{\mathrm{in}}}
    -(\rho_\Lambda^\beta)_{R_{\mathrm{in}}}
  }_1
  \le\eta_U(r,R_{\mathrm{in}}).
  \label{eq:bubble-inside-trace-phase}
\end{equation}
Applying \cref{thm:J-trace-continuity-phase} to
\eqref{eq:bubble-inside-trace-phase} bounds the difference between the two
values on $\mathcal G_{\mathrm{in}}$ by $\Omega_J$.

By \cref{def:A1-preserving-QLU-phase}, the wall state
$\omega_{\Lambda,D}$ is uniformly \Aone-admissible on all geometries,
including those intersecting $\partial D$.  Therefore
\cref{thm:main-J} applies to the deformation from
$\mathcal G_{\mathrm{out}}$ to $\mathcal G_{\mathrm{in}}$.  Inserting the two
wall-state values between the endpoint values and applying the triangle
inequality proves \eqref{eq:J-bubble-stability-phase}.
\end{proof}

\paragraph{Hall conductance estimator deformation error.}
For a pure sharp-charge state $\rho=\ket\Psi\!\bra\Psi$ and a
symmetry-preserving deformation $\mathcal G\to\mathcal G'$, denote by
\begin{equation}
  \mathcal E_\Sigma(\rho;\mathcal G\to\mathcal G')
  :=\sum_{j=1}^{N}\Lambda_j(u_j)\norm{Q_j\ket\Psi}_2
  \label{eq:ESigma-definition-phase}
\end{equation}
the right-hand side of \eqref{eq:Sigma-global-main}.  By
\cref{thm:main-Sigma}, if the relevant \Aone{} defects satisfy
\eqref{eq:A1-admissible-bound}, the on-site charges are uniformly bounded, and the
geometry and number of elementary moves grow polynomially in $\ell$, then
\begin{equation}
  \mathcal E_\Sigma(\rho;\mathcal G\to\mathcal G')
  \le p_\Sigma(\ell)\e^{-c_0\ell/2}.
  \label{eq:ESigma-exponential-phase}
\end{equation}

\begin{theorem}[Stability of the Hall conductance estimator under symmetry-protected quasi-local equivalence]
\label{thm:Sigma-phase-stability-final}
Suppose
\begin{equation}
  \boldsymbol\rho^\alpha
  \sim_{A_1,U(1)}
  \boldsymbol\rho^\beta
\end{equation}
through a $U(1)$-symmetric and \Aone-preserving quasi-local unitary path.
Let $D$, the bubble state $\omega_{\Lambda,D}$, the response partitions
$\mathcal G_{\mathrm{out}}$ and $\mathcal G_{\mathrm{in}}$, and the
truncation error $\eta_U(r,R_{\mathrm{in}})$ be as in \eqref{eq:endpoint-bubble-state-phase},
\eqref{eq:inside-outside-partitions-phase}, and \eqref{eq:eta-U-phase}.  For $\eta_U\le1/2$,
\begin{equation}
  \begin{aligned}
  &\abs{
    \Sigma(\mathcal G_{\mathrm{in}})_{\rho^\beta}
    -\Sigma(\mathcal G_{\mathrm{out}})_{\rho^\alpha}
  }
  \le
  \Omega_\Sigma\!\left(
    \sqrt{2\eta_U(r,R_{\mathrm{in}})};\mathcal G_{\mathrm{in}}
  \right)
  +\mathcal E_\Sigma\!\left(
    \omega_{\Lambda,D};
    \mathcal G_{\mathrm{out}}\to\mathcal G_{\mathrm{in}}
  \right).
  \end{aligned}
  \label{eq:Sigma-bubble-stability-phase}
\end{equation}
Here $\Omega_\Sigma(\delta;\mathcal G)$ is the modulus in
\eqref{eq:Omega-Sigma-phase} evaluated on the tripartition
$\mathcal G$.
\end{theorem}

\begin{proof}
The termwise symmetry condition
\eqref{eq:termwise-U1-symmetry-phase} implies
$[U_{\Lambda,D}(1),Q_\Lambda]=0$.  Hence the bubble state
$\omega_{\Lambda,D}$ is a pure state with the same sharp total charge as
$\rho_\Lambda^\alpha$, and the full endpoint $\rho_\Lambda^\beta$ lies in
the same charge sector.  Thus the exact symmetry hypothesis of
\cref{thm:main-Sigma} holds for the two bulks and for the wall state.

Because the truncated unitary is supported on $D$,
\eqref{eq:QLU-exterior-state-phase} gives
\begin{equation}
  (\omega_{\Lambda,D})_{R_{\mathrm{out}}}
  =(\rho_\Lambda^\alpha)_{R_{\mathrm{out}}},
\end{equation}
and therefore
\begin{equation}
  \Sigma(\mathcal G_{\mathrm{out}})_{\omega_{\Lambda,D}}
  =\Sigma(\mathcal G_{\mathrm{out}})_{\rho^\alpha}.
  \label{eq:Sigma-outside-exact-phase}
\end{equation}
Deep inside the bubble,
\eqref{eq:QLU-interior-state-phase} gives
\begin{equation}
  \frac12\norm{
    (\omega_{\Lambda,D})_{R_{\mathrm{in}}}
    -(\rho_\Lambda^\beta)_{R_{\mathrm{in}}}
  }_1
  \le\eta_U(r,R_{\mathrm{in}}).
  \label{eq:Sigma-bubble-inside-trace-phase}
\end{equation}
Both reduced states commute with the regional charge, so
\cref{thm:Sigma-trace-continuity-phase} bounds their Hall conductance estimators on
$\mathcal G_{\mathrm{in}}$ by $\Omega_\Sigma$.

By \cref{def:U1-A1-preserving-QLU-phase}, the wall state is uniformly
\Aone-admissible on every geometry encountered while moving the response
partition through $\partial D$, and it remains exactly $U(1)$ symmetric.
Therefore \cref{thm:main-Sigma} applies to the wall deformation
$\mathcal G_{\mathrm{out}}\to\mathcal G_{\mathrm{in}}$.  Inserting the two
wall-state values between the endpoint values and applying the triangle
inequality proves \eqref{eq:Sigma-bubble-stability-phase}.
\end{proof}

\begin{corollary}[Phase invariance of the limiting modular responses]
\label{cor:response-phase-invariance-final}
Assume fixed on-site dimension with two-dimensional response partitions
of scale $\ell$, a polynomial number of elementary deformation moves, and
comparison bubbles with $r\ge c\ell$.  Suppose the corresponding
thermodynamic response limits exist.  If
$\boldsymbol\rho^\alpha\sim_{A_1}\boldsymbol\rho^\beta$, then
\begin{equation}
  \boxed{
  J^{\mathrm{top}}(\boldsymbol\rho^\alpha)
  =J^{\mathrm{top}}(\boldsymbol\rho^\beta).
  }
  \label{eq:J-topological-phase-invariance-final}
\end{equation}
If, in addition, the on-site charges are uniformly bounded and
$\boldsymbol\rho^\alpha\sim_{A_1,U(1)}\boldsymbol\rho^\beta$, then
\begin{equation}
  \boxed{
  \Sigma^{\mathrm{top}}(\boldsymbol\rho^\alpha)
  =\Sigma^{\mathrm{top}}(\boldsymbol\rho^\beta).
  }
  \label{eq:Sigma-topological-phase-invariance-final}
\end{equation}
\end{corollary}

\begin{proof}
For $r\ge c\ell$, the quasi-local truncation error $\eta_U$ tends to zero
superpolynomially.  The dimension factors in $\Omega_J$ and, under the
bounded-charge assumption, the dimension and charge factors in
$\Omega_\Sigma$ grow only polynomially.  Hence the continuity terms in
\eqref{eq:J-bubble-stability-phase} and
\eqref{eq:Sigma-bubble-stability-phase} vanish as $\ell\to\infty$.
Uniform \Aone{} admissibility of the bubble-wall states similarly gives
$\mathcal E_J,\mathcal E_\Sigma\to0$ through
\cref{eq:EJ-exponential-phase,eq:ESigma-exponential-phase}.  The two
domain-wall bounds therefore imply that the finite-size response differences
vanish.  Taking the thermodynamic limit proves both equalities.
\end{proof}

\section{Finite-time instantaneous modular flow and quantitative \Aone{} obstructions}
\label{sec:finite-time-IMF-no-go}

Instantaneous modular flow (IMF) was introduced in
Ref.~\cite{strict-J-2024} as a central tool for proving the exact no-go theorem: a nonzero modular commutator is incompatible with exact \Aone{} uniformly.  For a normalized vector \(\ket\phi\), write
\begin{equation}
\rho_X(\phi):=\Tr_{\bar X}\ket\phi\!\bra\phi
\end{equation}
for its reduced density matrix on $X$, and define
\begin{equation}
\cI_X(t)\ket\phi
:=\rho_X(\phi)^{it}\ket\phi.
\label{eq:finite-IMF-definition}
\end{equation}
The defining feature of IMF, distinguishing it from ordinary modular flow, is that in a product of maps the reduced density matrix in each factor is recomputed from that factor's \emph{current input state}, rather than held
fixed at its initial value.  Products of IMF maps are composed from right to left.  IMF obeys the flipping and commutation properties
\cite{strict-J-2024}:
\begin{align}
\cI_X(t)\ket\phi
&=\cI_{\overline X}(t)\ket\phi,
\label{eq:finite-flip-move}\\
\cI_X(s)\cI_Y(t)\ket\phi
&=\cI_Y(t)\cI_X(s)\ket\phi,
\qquad X\subseteq Y.
\label{eq:finite-commutation-move}
\end{align}
The flipping identity follows directly from the Schmidt decomposition across $X|\overline X$: the reduced states on the two sides have the same nonzero Schmidt eigenvalues, so their imaginary powers act identically on $\ket\phi$.  If $X\subseteq Y$, flipping the $Y$ flow replaces it by a flow on $\overline Y$, which is disjoint from $X$.  The two flows then leave each other's reduced state unchanged and therefore commute; flipping back gives \eqref{eq:finite-commutation-move}.
For nonfaithful reduced states, imaginary powers are understood on the support and may be extended by the identity on the kernel when an ordinary unitary representative is convenient.

The exact no-go argument uses exact Markov decompositions of IMF to show that the modular commutator and Hall conductance estimator remain constant along the orbit
\(
\ket{\Psi(t)}=\cI_{AB}(t)\ket\Psi
\).
Along this orbit, the entropy and regional charge fluctuation will change linearly in time if $\Aone{}$ holds exactly, which is incompatible with their finite local bounds in the long time.  In this section, we extend this argument to finite times under approximate Markov conditions.  The main new input is a finite-time stability theorem for IMF that depends only on the initial conditional mutual information and requires no lower bound on any reduced-state eigenvalue.  We begin by reviewing the exact no-go argument of Ref.~\cite{strict-J-2024} and identifying the steps that require quantitative control.

\subsection{Review of the exact no-go theorem and strategy of the quantitative extension}
\label{subsec:exact-no-go-review}

Before proving the finite lower bound on the $\Aone{}$ violation, we review the logic of the exact no-go theorem of Ref.~\cite{strict-J-2024}. Fix a partition $A,B,C$ and the finite collection of auxiliary annular geometries needed in the proof. If the state has finite local Hilbert-space dimension and satisfies exact bulk \Aone{} on all of these geometries, then
\begin{equation}
J(A,B,C)_\Psi=0.
\label{eq:exact-no-go-local-statement}
\end{equation}
Equivalently, its contrapositive says that a nonzero finite modular commutator forces at least one local \Aone{} violation to be strictly positive. The original theorem is qualitative: it rules out the possibility that $\Aone{}$ holds exactly for every partition, but it does not estimate how large one of the defects must be.  The purpose of the present section is to promote the result to a quantitative, spectral-free lower bound.

\subsubsection{The exact entropy-drift argument}
\label{subsubsec:exact-entropy-drift-review}
Consider the instantaneous modular-flow orbit and its $BC$ reduced state and entropy:
\begin{align}
\ket{\Psi(t)}:=\cI_{AB}(t)\ket\Psi,
\label{eq:exact-review-Psi-t}\qquad
\rho_{BC}(t):=
\Tr_{\overline{BC}}\!\left(\ket{\Psi(t)}\!\bra{\Psi(t)}\right),
\qquad
S_{BC}(t)&:=S\!\left(\rho_{BC}(t)\right).
\end{align}
Write $K_R(t):=-\log\rho_R(t)$.  The IMF generator gives
\begin{equation}
\frac{d}{dt}\ket{\Psi(t)}=-iK_{AB}(t)\ket{\Psi(t)},
\qquad
\dot\rho_{BC}(t) =-i\Tr_{\overline{BC}} \left[K_{AB}(t),\ket{\Psi(t)}\!\bra{\Psi(t)}\right].
\end{equation}
Therefore, using $dS(\rho)/dt=\Tr(\dot\rho K)$ and cyclicity of the trace, we get the identity that relates the entropy production rate with the modular commutator:
\begin{align}
\frac{d}{dt}S_{BC}(t)
&=\Tr_{BC}\!\left[\dot\rho_{BC}(t)K_{BC}(t)\right]\notag\\
&=i\bra{\Psi(t)}[K_{AB}(t),K_{BC}(t)]\ket{\Psi(t)}.\notag \\
&= J(A,B,C)_{\Psi(t)}. 
\label{eq:exact-review-entropy-derivative}
\end{align}
Moreover, we can write the modular commutator in terms of the mixed infinitesimal generator of the overlap between the two IMFs.  For any normalized $\ket\phi$,
\begin{equation}
J(A,B,C)_\phi
=
-2\operatorname{Im}
\left.
\frac{\partial^2}{\partial x\,\partial y}
\braket{\cI_{AB}(x)\phi}{\cI_{BC}(y)\phi}
\right|_{x=y=0}.
\label{eq:J-IMF-overlap-generator}
\end{equation}
as the mixed derivative equals $\bra\phi K_{AB}K_{BC}\ket\phi$, whose imaginary part gives $i\bra\phi[K_{AB},K_{BC}]\ket\phi$ with the factor above.

The nontrivial part of the proof is to show that exact \Aone{} makes this slope constant:
\begin{equation}
J(A,B,C)_{\Psi(t)} = J(A,B,C)_\Psi
\qquad \text{for every }t\in\mathbb R.
\label{eq:exact-review-J-constant}
\end{equation}
We recall the two steps behind this statement.  Split $C=C'M$, let $D$ be an
outer collar, and set $E=\overline{ABCD}$.  The basic geometry and the refined
partition used below are shown in \cref{fig:finite-no-go-geometry}.

\begin{figure}[H]
\centering
\begin{minipage}[t]{0.48\linewidth}
\centering
\resizebox{\linewidth}{!}{%
\begin{tikzpicture}[line cap=round,line join=round]
  \path[use as bounding box] (-3.05,-2.65) rectangle (3.25,3.05);
  \fill[yellow!22] (0,0) circle (2.35);
  \path[fill=green!18]
    (30:1.92) arc[start angle=30,end angle=150,radius=1.92]
    -- (150:2.35) arc[start angle=150,end angle=30,radius=2.35] -- cycle;
  \path[fill=blue!12]
    (20:2.35) arc[start angle=20,end angle=160,radius=2.35]
    -- (160:2.72) arc[start angle=160,end angle=20,radius=2.72] -- cycle;
  \draw[thick] (0,0) circle (2.35);
  \draw[thick] (30:1.92) arc[start angle=30,end angle=150,radius=1.92];
  \draw[thick] (20:2.72) arc[start angle=20,end angle=160,radius=2.72];
  \draw[thick] (20:2.35)--(20:2.72);
  \draw[thick] (160:2.35)--(160:2.72);
  \draw[thick] (0,0)--(150:2.35);
  \draw[thick] (0,0)--(30:2.35);
  \draw[thick] (0,0)--(-90:2.35);
  \draw[densely dotted] (150:2.35) arc[start angle=150,end angle=30,radius=2.35];
  \node at (0,2.54) {$D$};
  \node at (0,2.13) {$M$};
  \node at (0,0.82) {$C'$};
  \node at (-0.86,-0.76) {$A$};
  \node at (0.86,-0.76) {$B$};
  \node at (3.02,0.05) {$E$};
\end{tikzpicture}%
}
\par\smallskip\textbf{(a)}
\end{minipage}\hfill
\begin{minipage}[t]{0.48\linewidth}
\centering
\resizebox{\linewidth}{!}{%
\begin{tikzpicture}[line cap=round,line join=round]
  \path[use as bounding box] (-3.05,-2.65) rectangle (3.25,3.05);
  \path[fill=green!18]
    (20:2.43) arc[start angle=20,end angle=160,radius=2.43]
    -- (160:2.8) arc[start angle=160,end angle=20,radius=2.8] -- cycle;
  \draw[line width=0.9pt] (20:2.43) -- (20:2.8)
    arc[start angle=20,end angle=160,radius=2.8] -- (160:2.43);
  \fill[yellow!22] (0,0) circle (2.43);
  \path[fill=green!18]
    (30:2) arc[start angle=30,end angle=150,radius=2]
    -- (150:2.43) arc[start angle=150,end angle=30,radius=2.43] -- cycle;
  \draw[line width=0.9pt] (30:2) arc[start angle=30,end angle=150,radius=2];
  \draw[line width=0.9pt]
    (150:2.43) arc[start angle=150,end angle=270,radius=2.43]
    arc[start angle=270,end angle=390,radius=2.43];
  \draw[line width=0.9pt] (0,0) -- (150:2.43);
  \draw[line width=0.9pt] (0,0) -- (30:2.43);
  \draw[line width=0.9pt] (0,0) -- (-90:2.43);
  \draw[line width=0.7pt,densely dotted]
    (150:2.43) arc[start angle=150,end angle=30,radius=2.43];
  \draw[red,line width=1.2pt]
    (150:1.22) arc[start angle=150,end angle=30,radius=1.22];
  \draw[red,line width=1.2pt]
    (30:0.78)
    .. controls (0.66,-0.58) and (1.32,-1.12)
    .. ({2.43*cos(-25)},{2.43*sin(-25)});
  \node at (0,2.67) {$D$};
  \node at (0,2.22) {$M$};
  \node at (0,1.42) {$C'_2$};
  \node at (0,0.58) {$C'_1$};
  \node at (-0.88,-0.78) {$A$};
  \node at (0.53,-1.18) {$B_1$};
  \node at (1.30,-0.20) {$B_2$};
  \node at (2.92,0.02) {$E$};
\end{tikzpicture}%
}
\par\smallskip\textbf{(b)}
\end{minipage}
\caption[Geometries used in the exact no-go proof]{Geometries used in the exact no-go proof.  (a) The original response
region is $C=C'M$; $D$ is an outer collar and
$E=\overline{ABCD}$.  The two Markov decompositions are applied to the
ordered tripartitions $E|D|C$ and $C'|M|D$.  (b) The refined partition with
$B=B_1B_2$ and $C'=C'_1C'_2$.\protect\footnotemark}
\label{fig:finite-no-go-geometry}
\end{figure}
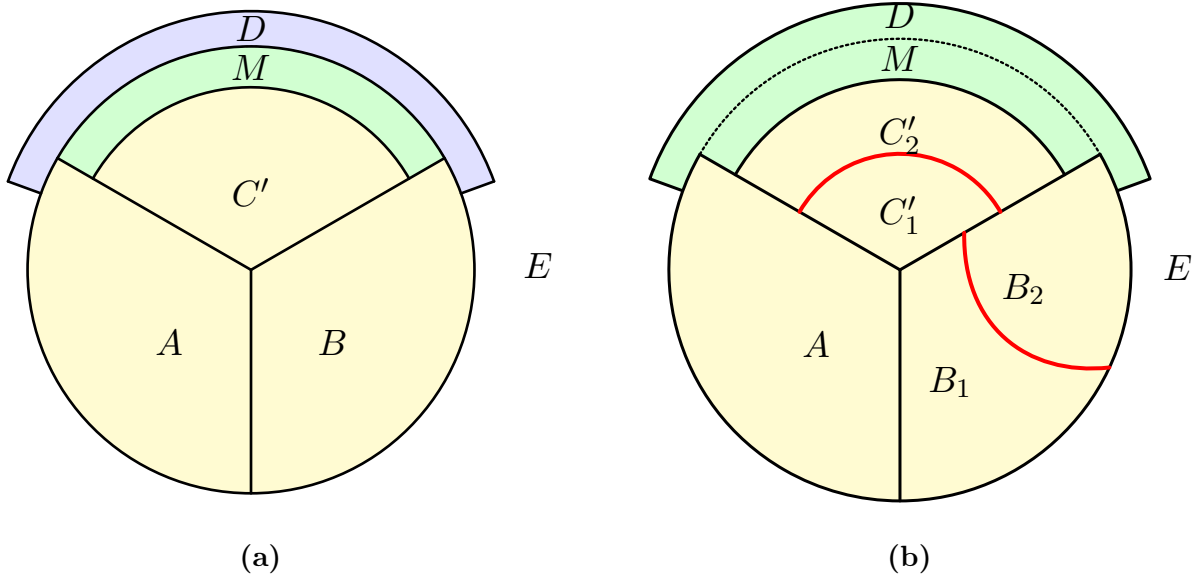
\footnotetext{The figure is adapted from \emph{Strict Area Law Entanglement versus
Chirality}~\cite{strict-J-2024} and it was created in TikZ
with generative AI assistance.}

Introduce the auxiliary orbit
\begin{equation}
\ket{\Psi'(t)}:=\cI_{MD}(t)\ket\Psi.
\label{eq:exact-review-Psi-prime}
\end{equation}
Exact \Aone{} implies
\begin{equation}
I(C:E|D)_\Psi=0,
\qquad
I(D:C'|M)_\Psi=0,
\label{eq:exact-review-two-CMIs}
\end{equation}
Furthermore, for an exact quantum Markov chain $X|Y|Z$, instantaneous modular flow obeys the exact decomposition move
\begin{equation}
I(X:Z|Y)_\phi=0
\quad\Longrightarrow\quad
\cI_{XYZ}(t)\ket\phi
=
\cI_{XY}(t)\cI_{YZ}(t)\cI_Y(-t)\ket\phi.
\label{eq:exact-review-Markov-move}
\end{equation}
Since $\overline{AB}=CDE$, the flipping identity
\eqref{eq:finite-flip-move} first gives
\begin{equation}
\ket{\Psi(t)}=\cI_{CDE}(t)\ket\Psi.
\end{equation}
Applying \eqref{eq:exact-review-Markov-move} to the Markov chain $E|D|C$
and then using nested commutation \eqref{eq:finite-commutation-move} yields
\begin{align}
\ket{\Psi(t)}
&=\cI_{DE}(t)\cI_{CD}(t)\cI_D(-t)\ket\Psi
\notag\\
&=\cI_D(-t)\cI_{DE}(t)\cI_{CD}(t)\ket\Psi.
\label{eq:exact-review-first-decomposition}
\end{align}
For the second Markov chain $C'|M|D$, recall that $C=C'M$.  Another
application of \eqref{eq:exact-review-Markov-move}, followed by nested
commutation, gives
\begin{align}
\cI_{CD}(t)\ket\Psi
&=\cI_C(t)\cI_{MD}(t)\cI_M(-t)\ket\Psi
\notag\\
&=\cI_M(-t)\cI_C(t)\cI_{MD}(t)\ket\Psi
\notag\\
&=\cI_M(-t)\cI_C(t)\ket{\Psi'(t)}.
\label{eq:exact-review-second-decomposition}
\end{align}
Substituting \eqref{eq:exact-review-second-decomposition} into
\eqref{eq:exact-review-first-decomposition} gives
\begin{equation}
\ket{\Psi(t)}
=
\cI_D(-t)\cI_{DE}(t)\cI_M(-t)\cI_C(t)
\ket{\Psi'(t)}.
\label{eq:exact-review-two-decompositions}
\end{equation}

The IMF flipping and nested-commutation moves allow the probe flows
$\cI_{AB}(x)$ and $\cI_{BC}(y)$ to be pulled through the four factors in
\eqref{eq:exact-review-two-decompositions}.  The remaining factors are
implemented by the same unitary
\begin{equation}
V(t)
=
\bigl[\rho_D(\Psi'(t))\bigr]^{-it}
\bigl[\rho_{DE}(\Psi'(t))\bigr]^{it}
\bigl[\rho_M(\Psi'(t))\bigr]^{-it}
\bigl[\rho_C(\Psi'(t))\bigr]^{it},
\label{eq:exact-review-V-definition}
\end{equation}
Its action on the two probe states gives
\begin{align}
\cI_{AB}(x)\ket{\Psi(t)}
&=
V(t)\cI_{AB}(x)\ket{\Psi'(t)},
\label{eq:exact-review-common-unitary-AB}\\
\cI_{BC}(y)\ket{\Psi(t)}
&=
V(t)\cI_{BC}(y)\ket{\Psi'(t)}.
\label{eq:exact-review-common-unitary-BC}
\end{align}
Consequently, $V(t)$ cancels from the overlap of the two probe states.  By
\eqref{eq:J-IMF-overlap-generator}, taking the mixed derivative at $x=y=0$
yields the first equality
\begin{equation}
J(A,B,C)_{\Psi(t)} = J(A,B,C)_{\Psi'(t)}.
\label{eq:exact-review-step-one}
\end{equation}

For the second step, refine $B=B_1B_2$ and $C'=C'_1C'_2$ as in
\cref{fig:finite-no-go-geometry}(b).  These refinements resolve the junction
where $B$ meets the removed patch $M$: a direct one-defect argument would
require $I(B:M\mid C')=0$, which is not supplied by the buffered local
\Aone{} condition.  Instead, we use the four Markov conditions
\begin{align*}
I(B_1C'_1:M\mid B_2C'_2)&=0,
& I(A:B_2\mid B_1)&=0,\\
I(B_1:C'_2M\mid B_2)&=0,
& I(B_1:C'_2\mid B_2)&=0.
\end{align*}
They remain exact on $\ket{\Psi'(t)}$: the CMIs on regions disjoint from
$MD$ are unchanged, while those involving $M$ are controlled by enlarging
the leg containing $M$ to include $D$.  For example,
\[
I(B_1C'_1:M\mid B_2C'_2)_{\Psi'(t)}
\le I(B_1C'_1:MD\mid B_2C'_2)_{\Psi'(t)}
= I(B_1C'_1:MD\mid B_2C'_2)_{\Psi}=0,
\]
by data processing and invariance under the unitary $MD$ flow.
The first two Markov conditions eliminate the two modular-defect terms in
the algebraic identity derived below in \eqref{eq:J-deformation-identity};
the last two make its residual commutator expectations vanish by
\eqref{eq:J-defect-main}.  Thus
\begin{equation}
J(A,B,C)_{\Psi'(t)} = J(A,B,C')_{\Psi'(t)}.
\label{eq:exact-review-step-two}
\end{equation}
Finally, the $MD$ flow does not change the $ABC'$ reduced state, and exact \Aone{} gives deformation invariance on the initial state.  Hence
\begin{equation}
J(A,B,C')_{\Psi'(t)}
=
J(A,B,C')_\Psi
=
J(A,B,C)_\Psi.
\label{eq:exact-review-final-chain}
\end{equation}
Equations
\eqref{eq:exact-review-step-one}--\eqref{eq:exact-review-final-chain} prove
\eqref{eq:exact-review-J-constant}.

Substitution into the exact drift identity gives
\begin{equation}
S_{BC}(t)
=
S_{BC}(0)+tJ(A,B,C)_\Psi.
\label{eq:exact-review-linear-entropy}
\end{equation}
On the other hand, all states in the orbit have their $BC$ marginal in the
same finite-dimensional Hilbert space, so
\begin{equation}
0\le S_{BC}(t)\le\log d_{BC}
\qquad
\text{for every }t.
\label{eq:exact-review-entropy-range}
\end{equation}
A nonzero constant slope in \eqref{eq:exact-review-linear-entropy} is incompatible with \eqref{eq:exact-review-entropy-range} for large $t$.  

The Hall conductance estimator no-go theorem has the same structure.
Under exact on-site $U(1)$ symmetry, consider the orbit
$\partial_t\ket{\Psi(t)}=-iK_{AB}(t)\ket{\Psi(t)}$.
First, the mean is constant: $[K_{AB}(t),Q_C]=0$ by disjoint support, so
\[
\frac{d}{dt}\langle Q_{BC}\rangle_{\Psi(t)}
=i\Tr\rho_{AB}(t)[K_{AB}(t),Q_B]
=i\Tr[\rho_{AB}(t),K_{AB}(t)]Q_B=0.
\]
Next, differentiating the second moment gives
\begin{align}
\frac{d}{dt}\langle Q_{BC}^{2}\rangle_{\Psi(t)}
&=i\bra{\Psi(t)}K_{AB}(t)Q_{BC}^2\ket{\Psi(t)}
-i\bra{\Psi(t)}Q_{BC}^2K_{AB}(t)\ket{\Psi(t)}\notag\\
&=i\langle[K_{AB}(t),Q_{BC}^2]\rangle_{\Psi(t)}
=2\Sigma(A,B,C)_{\Psi(t)}.
\label{eq:exact-review-Hall-drift}
\end{align}
Consequently, centering the charge leaves the drift unchanged:
\begin{equation}
\frac{d}{dt}\operatorname{Var}_{\Psi(t)}(Q_{BC})
=\frac{d}{dt}\bigl(\langle Q_{BC}^2\rangle_{\Psi(t)}
-\langle Q_{BC}\rangle_{\Psi(t)}^2\bigr)
=2\Sigma(A,B,C)_{\Psi(t)}.
\label{eq:mean-charge-constant}
\end{equation}
Exact \Aone{} makes $\Sigma(A,B,C)_{\Psi(t)}$ constant by the same two-orbit comparison and deformation argument.  A nonzero constant Hall conductance estimator would then make the second moment grow linearly without bound, contradicting the bounded spectrum of $Q_{BC}$ in a finite local Hilbert space.  For details, see Ref.~\cite{strict-J-2024}.

\subsubsection{Overview of the extension to finite lower bound of \texorpdfstring{\Aone{}}{A1}}
\label{subsubsec:finite-A1-overview}

Suppose now that the finitely many \Aone{} defects used by the no-go geometry
are not zero.  Their bounds depend on the response, since the deformation
geometries differ: write $\eps_{A_1}^{(J)}$ and $\eps_{A_1}^{(\Sigma)}$
as defined in \eqref{eq:uniform-A1-error}.  We illustrate the argument
below for $J$; the Hall argument uses $\eps_{A_1}^{(\Sigma)}$ instead.
The exact drift identities
\eqref{eq:exact-review-entropy-derivative} and
\eqref{eq:exact-review-Hall-drift} remain unchanged.  What must be stabilized
is the constancy of the response along the IMF orbit.  The quantitative proof
replaces the three exact ingredients above as follows.

\begin{enumerate}
\item \emph{Approximate Markov decomposition at finite time.}
Exact \Aone{} supplied zero CMIs and hence the exact identity
\eqref{eq:exact-review-Markov-move}.  For finite \Aone{}, purification and data processing give $I(X:Z|Y)\le\eps_{A_1}^{(J)}$
for every tripartition used in the two decompositions.  We prove below that the corresponding IMF factorization error
\begin{equation}
\mathfrak E_{\phi}^{X:Y:Z}(t)
:=\cI_{XYZ}(t)\ket\phi - \cI_{XY}(t)\cI_{YZ}(t)\cI_Y(-t)\ket\phi
\label{eq:exact-review-flow-error-preview}
\end{equation}
is controlled directly by the initial CMI through:
\begin{equation}
\pi\int_{-\infty}^{\infty}
\frac{
\norm{\mathfrak E_{\phi}^{X:Y:Z}(t)}_2^2
}{
\ell(t)^2[\cosh(2\pi t)-1]
}\,dt
\le
I(X:Z|Y)_\phi.
\label{eq:exact-review-weighted-preview}
\end{equation}
where $\ell(t):=\sqrt{2(1+4t^2)}$.  No evolved CMI and no lower spectral bound are required.

\item \emph{Approximate comparison of the two IMF orbits.}
Apply \eqref{eq:exact-review-weighted-preview} to the same ordered
tripartitions $E|D|C$ and $C'|M|D$ as in the exact proof.  After propagating
the second error through the remaining IMF factors, one constructs the state
\begin{equation}
\ket{\widetilde\Psi(t)}
:=
\cI_D(-t)\cI_{DE}(t)\cI_M(-t)\cI_C(t)
\ket{\Psi'(t)},
\label{eq:exact-review-Psi-tilde-preview}
\end{equation}
which satisfies a weighted estimate of the form
\begin{equation}
\pi\int_{-\infty}^{\infty}
\frac{
r(t)^2
}{
\ell(t)^6[\cosh(2\pi t)-1]
}\,dt
\le4\eps_{A_1}^{(J)},
\qquad
r(t):=\norm{\ket{\Psi(t)}-\ket{\widetilde\Psi(t)}}_2.
\label{eq:exact-review-r-preview}
\end{equation}
The common-unitary cancellation remains exact between
$\ket{\widetilde\Psi(t)}$ and $\ket{\Psi'(t)}$; only the comparison between
$\ket{\Psi(t)}$ and $\ket{\widetilde\Psi(t)}$ is approximate.

\item \emph{Approximate response constancy.}
The spectrum-independent continuity theorem converts $r(t)$ into a change of the modular commutator, while \cref{thm:main-J} controls the two deformations $C\leftrightarrow C'$ by a modulus $\mathcal D_J(\eps_{A_1}^{(J)})=O\!\left(L^2\sqrt{\eps_{A_1}^{(J)}[L^2+\log(1/\eps_{A_1}^{(J)})]}\right)$ for regular two-dimensional regions of linear size $L$ and fixed on-site dimension.  The exact
chain \eqref{eq:exact-review-step-one}--\eqref{eq:exact-review-final-chain} is then replaced by
\begin{equation}
\abs{
J(A,B,C)_{\Psi(t)}-J(A,B,C)_\Psi
}
\le
\Omega_J(r(t);A,B,C)+2\mathcal D_J(\eps_{A_1}^{(J)}),
\label{eq:exact-review-J-orbit-preview}
\end{equation}
where $\Omega_J(r;A,B,C)=O\!\left(rL^2[L^2+\log(1/r)]\right)$ as $r\to0$ for regular two-dimensional regions of linear size $L$ and fixed on-site dimension.  The Hall conductance estimator obeys an analogous bound with $\Omega_\Sigma$ and $\mathcal D_\Sigma$.

\item \emph{Invert the finite drift budget.}
Define the distance of the initial entropy to the nearer endpoint of its
allowed interval by
\begin{equation}
R
:=
\min\{S_{BC}(0),\log d_{BC}-S_{BC}(0)\}.
\label{eq:exact-review-R-definition}
\end{equation}
Assume $J(A,B,C)_\Psi\ne0$ and choose a time-orientation sign
$\kappa\in\{-1,1\}$ so that the initial entropy drift along
$\Psi(\kappa t)$ points toward the nearer endpoint: choose
$\kappa J(A,B,C)_\Psi<0$ if that endpoint is $0$, and
$\kappa J(A,B,C)_\Psi>0$ if it is $\log d_{BC}$
(either endpoint may be chosen in case of a tie).
Combining the exact entropy identity with
\eqref{eq:exact-review-J-orbit-preview} gives, for every
$T>R/\abs{J(A,B,C)_\Psi}$,
\begin{equation}
\boxed{
\abs{J(A,B,C)_\Psi}T-R
\le
\int_0^T\Omega_J(r(\kappa t);A,B,C)\,dt
+
2T\mathcal D_J(\eps_{A_1}^{(J)}).
}
\label{eq:exact-review-master-preview}
\end{equation}
The weighted estimate \eqref{eq:exact-review-r-preview} bounds the first term
on the right.  Since both error moduli vanish with $\eps_{A_1}^{(J)}$, the
right-hand side of \eqref{eq:exact-review-master-preview} would be too small
to accommodate a fixed positive left-hand side if $\eps_{A_1}^{(J)}$ were
arbitrarily small.  Inverting this inequality gives a strictly positive,
finite-dimensional, spectral-free lower bound on the required \Aone{}
violation.  The Hall proof is identical after replacing the entropy budget by
the finite regional charge-fluctuation budget and replacing
$\abs{J(A,B,C)_\Psi}$ by
$\abs{\sigma_{xy}}$.
\end{enumerate}

The remaining subsections implement these four steps.  In particular, they prove the finite-time IMF estimate, construct the common comparison path, and optimize the drift time.  Under the entanglement area law assumption: $R(L)=O(L)$, the obstruction obtained has the following form
\begin{equation}
\eps_{A_1}^{(J)}(L)
\ge
\poly(L,\abs{c_-}^{-1})^{-1}
\exp\!\left[-\frac{6R(L)}{\abs{c_-}}\right],
\end{equation}
The precise form is stated in
\cref{cor:thermal-exp-L}.  Thus the quantitative extension keeps the original
contradiction mechanism---a finite local budget versus a persistent response
drift---while replacing every exact Markov equality by a controlled,
spectral-free error.

\subsection{Global Lipschitz continuity of one IMF map}

Although $\cI_X(t)$ is nonlinear, it is globally Lipschitz in the state vector:
\begin{lemma}[Spectrum-free IMF Lipschitz bound]
\label{lem:IMF-Lipschitz}
For arbitrary normalized vectors $\ket\phi,\ket\chi$,
\begin{equation}
\norm{\cI_X(t)\ket\phi-\cI_X(t)\ket\chi}_2
\le
\ell(t)\norm{\ket\phi-\ket\chi}_2,
\qquad
\ell(t):=\sqrt{2(1+4t^2)}.
\label{eq:IMF-Lipschitz}
\end{equation}
\end{lemma}

\begin{proof}
Identify a vector across $X:\bar X$ with its coefficient matrix.  Write
\[
\ket\phi\leftrightarrow T,
\qquad
\ket\chi\leftrightarrow S,
\qquad
\rho_X(\phi)=TT^\dagger,
\qquad
\rho_X(\chi)=SS^\dagger.
\]
Then
\[
\norm{\ket\phi-\ket\chi}_2=\norm{T-S}_2.
\]
If $T=U\operatorname{diag}(s_j)V^\dagger$, the coefficient matrix of the
flowed vector is
\begin{equation}
(TT^\dagger)^{i t}T
=
U\operatorname{diag}(s_j^{1+2i t})V^\dagger.
\label{eq:IMF-singular-transform}
\end{equation}
Introduce the odd function
\[
f_t(x)=
\begin{cases}
\operatorname{sgn}(x)|x|^{1+2i t},&x\ne0,\\
0,&x=0.
\end{cases}
\]
On each half-axis, $|f_t'(x)|=\sqrt{1+4t^2}$; for opposite signs,
$|f_t(x)-f_t(y)|\le|x|+|y|=|x-y|$.  Hence
\[
\operatorname{Lip}(f_t)\le\sqrt{1+4t^2}.
\]
Equation \eqref{eq:IMF-singular-transform} is the singular-value transform
$f_t^\diamond(T)$ in the notation of
\cref{lem:truncated-modular-vector-app}.  Applying its complex-valued bound
\eqref{eq:complex-singular-transform-bound-app} to the restriction of $f_t$
to $[0,1]$ gives
\begin{align*}
\norm{\cI_X(t)\ket\phi-\cI_X(t)\ket\chi}_2
&=\norm{f_t^\diamond(T)-f_t^\diamond(S)}_2\\
&\le\sqrt{2(1+4t^2)}\norm{T-S}_2\\
&=\ell(t)\norm{\ket\phi-\ket\chi}_2.
\end{align*}
This proves \eqref{eq:IMF-Lipschitz}.
\end{proof}

\subsection{Weighted finite-time theorem}

For an ordered tripartition $X|Y|Z$, define
\begin{equation}
\mathfrak E_{\phi}^{X:Y:Z}(t)
:=
\cI_{XYZ}(t)\ket\phi
-
\cI_{XY}(t)\cI_{YZ}(t)\cI_Y(-t)\ket\phi.
\label{eq:general-flow-error}
\end{equation}
All factors on the right are genuine current-state IMF maps.  The following
is the finite-time stability theorem used in both no-go arguments.

\begin{theorem}[Weighted finite-time Markov decomposition]
\label{thm:weighted-IMF}
\begin{equation}
\pi\int_{-\infty}^{\infty}
\frac{
\norm{\mathfrak E_{\phi}^{X:Y:Z}(t)}_2^2
}{
\ell(t)^2[\cosh(2\pi t)-1]
}\,dt
\le I(X:Z|Y)_\phi.
\label{eq:weighted-IMF}
\end{equation}
\end{theorem}

\begin{proof}
First suppose that $\rho:=\rho_{XYZ}$ is faithful, and define the
Hilbert--Schmidt distance between the two matrix amplitudes
\begin{equation}
d_\phi(t)
:=\norm{
\rho_{XY}^{-it}\rho^{1/2+it}
-\rho_Y^{-it}\rho_{YZ}^{it}\rho^{1/2}
}_2.
\label{eq:dphi-definition}
\end{equation}
Set $A_t:=\rho_{XY}^{-it}\rho^{it}$ and
$B_t:=\rho_Y^{-it}\rho_{YZ}^{it}$.  Since the $XYZ$ flow leaves the $XY$
marginal unchanged, and the $YZ$ flow leaves the $Y$ marginal unchanged,
\begin{equation}
A_t\ket\phi=\cI_{XY}(-t)\cI_{XYZ}(t)\ket\phi,
\qquad
B_t\ket\phi=\cI_Y(-t)\cI_{YZ}(t)\ket\phi.
\end{equation}
The inverse and nested commutation identities therefore give
\begin{equation}
\mathfrak E_{\phi}^{X:Y:Z}(t)
=\cI_{XY}(t)A_t\ket\phi-\cI_{XY}(t)B_t\ket\phi.
\end{equation}
Applying \cref{lem:IMF-Lipschitz} to the two normalized inputs with $M=A_t-B_t$, we obtain
\begin{align}
\norm{\mathfrak E_{\phi}^{X:Y:Z}(t)}_2
&\le\ell(t)\norm{(A_t-B_t)\ket\phi}_2\notag\\
&=\ell(t)\norm{(A_t-B_t)\rho^{1/2}}_2
=\ell(t)d_\phi(t).
\label{eq:E-by-d}
\end{align}

To control $d_\phi(t)$, introduce the operator-valued analytic function
\begin{equation}
\mathcal W(z)
:=
\rho_{YZ}^{z}\rho_Y^{-z}\rho_{XY}^{z}\rho^{1/2-z},
\qquad
0\le\operatorname{Re}z\le\frac12.
\label{eq:interpolation-function-definition}
\end{equation}
The boundary estimates and the three-lines interpolation argument are given
in \cref{lem:interpolation-contraction-app}.  They imply
\begin{equation}
\norm{\mathcal W(z)}_2\le1
\qquad
(0\le\operatorname{Re}z\le\tfrac12).
\label{eq:interpolation-contraction}
\end{equation}
Set
\[
f(z):=\Tr[\rho^{1/2}\mathcal W(z)],
\qquad
h(z):=1-\operatorname{Re}f(z).
\]
Since $f$ is analytic, $h=1-\operatorname{Re}f$ is harmonic on the strip.
By \eqref{eq:interpolation-contraction}, it is also nonnegative and bounded,
and $h(0)=0$.  To differentiate \eqref{eq:interpolation-function-definition}, use
\begin{equation}
\left.\frac{d}{dz}A^z\right|_{z=0}=\log A,
\qquad
\frac{d}{dz}\rho^{1/2-z}
=-(\log\rho)\rho^{1/2-z}.
\end{equation}
At \(z=0\), the first three factors in \(\mathcal W(z)\) are identities.
The product rule therefore gives
\begin{equation}
\mathcal W'(0)
=
\left(
\log\rho_{YZ}-\log\rho_Y+\log\rho_{XY}-\log\rho
\right)\rho^{1/2}.
\label{eq:interpolation-derivative-origin}
\end{equation}
Since \(f(z)=\Tr[\rho^{1/2}\mathcal W(z)]\), cyclicity of the trace yields
\begin{align}
f'(0)
&=
\Tr\rho
\left(
\log\rho_{YZ}-\log\rho_Y+\log\rho_{XY}-\log\rho
\right)
\notag\\
&=
\Tr\rho_{YZ}\log\rho_{YZ}
-\Tr\rho_Y\log\rho_Y
+\Tr\rho_{XY}\log\rho_{XY}
-\Tr\rho\log\rho. 
\notag\\
&=-I(X:Z|Y)_\rho = -\partial_x h(0).
\label{eq:f-derivative-entropies}
\end{align}
Here $\partial_x$ denotes differentiation with respect to $x=\operatorname{Re}z$; at the left boundary $x=0$, it is the inward normal
derivative of the strip.

The Poisson integral formula on a strip and its boundary differentiation are
stated and derived in \cref{lem:strip-Poisson-normal-app}.  Applying that
lemma with strip width $a=1/2$ gives
\begin{align}
\partial_xh(0)
={}&
2\pi\int_{-\infty}^{\infty}
\frac{h(i s)}{\cosh(2\pi s)-1}\,ds +
2\pi\int_{-\infty}^{\infty}
\frac{h(\tfrac12+i s)}{\cosh(2\pi s)+1}\,ds.
\label{eq:Poisson-normal}
\end{align}
The second integral is nonnegative.  Left multiplication by the unitary
$\rho_{YZ}^{-is}\rho_Y^{is}$ in \eqref{eq:dphi-definition} gives
\begin{equation}
d_\phi(s)
=
\norm{\mathcal W(-i s)-\rho^{1/2}}_2.
\label{eq:d-interpolation-function}
\end{equation}
Both matrices have unit Hilbert--Schmidt norm, and hence
\begin{align}
h(i s)
&=1-\operatorname{Re}\Tr[\rho^{1/2}\mathcal W(i s)]\notag\\
&=\frac12\left(
\norm{\mathcal W(i s)}_2^2+\norm{\rho^{1/2}}_2^2
-2\operatorname{Re}\Tr[\rho^{1/2}\mathcal W(i s)]
\right)\notag\\
&=\frac12\norm{\mathcal W(i s)-\rho^{1/2}}_2^2
=\frac12d_\phi(-s)^2.
\label{eq:h-d}
\end{align}
Substitution into \eqref{eq:Poisson-normal}, followed by dropping the
nonnegative upper-boundary integral, yields
\begin{equation}
\pi\int_{-\infty}^{\infty}
\frac{d_\phi(s)^2}{\cosh(2\pi s)-1}\,ds
\le I(X:Z|Y)_\rho.
\label{eq:weighted-d}
\end{equation}
Combining \eqref{eq:E-by-d} and \eqref{eq:weighted-d} proves
\eqref{eq:weighted-IMF} for faithful states.

The faithful-state assumption is removed in
\cref{app:nonfaithful-IMF}.  There, a faithful depolarizing regularization,
the global IMF Lipschitz bound, and Fatou's lemma extend
\eqref{eq:weighted-IMF} to arbitrary finite-dimensional states without
changing its right-hand side or introducing any spectral dependence.
\end{proof}

The finite-time factorization error of the IMF also admits an explicit bound, not just the integral bound above: for fixed time and subsystem dimensions, it is $O\!\left(I(X:Z|Y)_\phi^{1/3}\right)$. The precise statement and proof are given in \cref{app:pointwise-IMF}. The proof of the $\Aone{}$ violation lower-bound below uses \cref{thm:weighted-IMF} directly.

\subsection{Approximate no-go geometry and common comparison path}
\label{sec:shared-geometry}

We now use the geometry of the exact no-go proofs for the modular commutator
and Hall conductance estimator.  Let
$A,B,C$ be the three-sector response partition, split $C=C'M$,
let $D$ be the neighboring outer collar, and set $E:=\overline{ABCD}$.
This geometry is illustrated in \cref{fig:finite-no-go-geometry}.

Define
\begin{equation}
\ket{\Psi(t)}:=\cI_{AB}(t)\ket\Psi,
\qquad
\ket{\Psi'(t)}:=\cI_{MD}(t)\ket\Psi.
\label{eq:two-trajectories}
\end{equation}

With the refinements $B=B_1B_2$ and $C'=C'_1C'_2$ in
\cref{fig:finite-no-go-geometry}(b), let $\mathcal G_{\rm flow}$ be the
finite family of \Aone{} geometries controlling the two IMF Markov
decompositions on $E|D|C$ and $C'|M|D$.  This family is common to both
responses.  Let $\mathcal G_{\rm def}^{(J)}$ be the family controlling the
modular-commutator deformation tripartitions
\begin{align*}
B_1C'_1\,|\,B_2C'_2\,|\,MD, \qquad
A\,|\,B_1\,|\,B_2,\qquad
B_1\,|\,B_2\,|\,C'_2MD, \qquad
B_1\,|\,B_2\,|\,C'_2,
\end{align*}
whereas $\mathcal G_{\rm def}^{(\Sigma)}$ controls only
$A|B_1|B_2$ and $B_1|A|MD$, the tripartitions needed for
the Hall conductance estimator deformation.
The tripartitions containing $MD$ are obtained by enlarging $M$ to $MD$
in the corresponding deformation tripartitions.  Since the flow generating
$\ket{\Psi'(t)}$ acts only on $MD$, the enlarged CMIs are unchanged by the flow; tracing out $D$ bounds the CMIs needed for the deformations as we will demonstrate below.
The CMIs for $A|B_1|B_2$ and $B_1|B_2|C'_2$ are unchanged because their regions are disjoint from $MD$.
For each response $\mathsf P\in\{J,\Sigma\}$, define
\begin{equation}
\mathcal G_{\rm ng}^{(\mathsf P)}
:=\mathcal G_{\rm flow}\cup\mathcal G_{\rm def}^{(\mathsf P)},
\qquad
\eps_{A_1}^{(\mathsf P)}
:=
\sup_{\mathcal G\in\mathcal G_{\rm ng}^{(\mathsf P)}}
\delta_1^\Psi(\mathcal G).
\label{eq:uniform-A1-error}
\end{equation}
By \eqref{eq:A1-purification-main} and data processing, every CMI used
for response $\mathsf P$ is at most $\eps_{A_1}^{(\mathsf P)}$.
Thus each error includes the common decomposition geometries but only
the deformation geometries needed for its own response.

\subsubsection{The two finite-time decomposition errors}

Define
\begin{align}
e_1(t)
:={}&
\norm{
\cI_{CDE}(t)\ket\Psi
-
\cI_D(-t)\cI_{DE}(t)\cI_{CD}(t)\ket\Psi
}_2,
\label{eq:e1-definition}\\
e_2(t)
:={}&
\norm{
\cI_{CD}(t)\ket\Psi
-
\cI_M(-t)\cI_C(t)\cI_{MD}(t)\ket\Psi
}_2.
\label{eq:e2-definition}
\end{align}
The order of the nested factors may be changed by
\eqref{eq:finite-commutation-move}.  The ordered tripartitions are $E|D|C$ and
$C'|M|D$, respectively.  Since these decomposition geometries belong to
both families, \cref{thm:weighted-IMF} gives, for each
$\mathsf P\in\{J,\Sigma\}$,
\begin{equation}
\pi\int_{-\infty}^{\infty}
\frac{e_\nu(t)^2}{\ell(t)^2[\cosh(2\pi t)-1]}
\,dt
\le
\eps_{A_1}^{(\mathsf P)},
\qquad \nu=1,2.
\label{eq:e1-e2-weighted}
\end{equation}

Using the flipping move, $\ket{\Psi(t)}=\cI_{CDE}(t)\ket\Psi$.  Define
\begin{equation}
\ket{\widetilde\Psi(t)}
:=
\cI_D(-t)\cI_{DE}(t)
\cI_M(-t)\cI_C(t)
\ket{\Psi'(t)},
\qquad
r(t):=\norm{\ket{\Psi(t)}-\ket{\widetilde\Psi(t)}}_2.
\label{eq:Psi-tilde}
\end{equation}
The outer composition $\cI_D(-t)\cI_{DE}(t)$ is
$\ell(t)^2$-Lipschitz, so
\begin{align}
r(t)
&\le e_1(t)
+\norm{
\cI_D(-t)\cI_{DE}(t)\cI_{CD}(t)\ket\Psi
-
\cI_D(-t)\cI_{DE}(t)\cI_M(-t)\cI_C(t)
\ket{\Psi'(t)}
}_2
\notag\\
&\le e_1(t)+\ell(t)^2e_2(t).
\label{eq:r-e1-e2}
\end{align}
Using $(a+b)^2\le2a^2+2b^2$ and $\ell(t)\ge1$, equations
\eqref{eq:e1-e2-weighted}--\eqref{eq:r-e1-e2} imply
\begin{align}
\pi\int_{-\infty}^{\infty}
\frac{r(t)^2}{\ell(t)^6[\cosh(2\pi t)-1]}
\,dt
&\le 2\pi\int_{-\infty}^{\infty}
\frac{e_1(t)^2+e_2(t)^2}{\ell(t)^2[\cosh(2\pi t)-1]}
\,dt
\notag\\
&\le4\eps_{A_1}^{(\mathsf P)}.
\label{eq:r-weighted}
\end{align}

\subsubsection{Exact common-unitary covariance}

The remaining IMF sequence in \eqref{eq:Psi-tilde} obeys the exact
common-unitary relation used in the strict no-go proof.  There is a unitary
$V(t)$, independent of the probe parameters, such that
\begin{align}
\cI_{AB}(x)\ket{\widetilde\Psi(t)}
&=
V(t)\cI_{AB}(x)\ket{\Psi'(t)},
\label{eq:V-AB}\\
\cI_{BC}(y)\ket{\widetilde\Psi(t)}
&=
V(t)\cI_{BC}(y)\ket{\Psi'(t)}.
\label{eq:V-BC}
\end{align}
One may take $V(t)$ to be the product of unitary extensions of the
$D,DE,M,C$ reduced-state powers computed from $\ket{\Psi'(t)}$:
\begin{equation}
V(t)
=
\bigl[\rho_D(\Psi'(t))\bigr]^{-it}
\bigl[\rho_{DE}(\Psi'(t))\bigr]^{it}
\bigl[\rho_M(\Psi'(t))\bigr]^{-it}
\bigl[\rho_C(\Psi'(t))\bigr]^{it}.
\label{eq:V-definition}
\end{equation}
The required reduced state is unchanged by all preceding factors because
their supports either contain the region in question or are disjoint from it.
Thus the same fixed unitary implements the residual IMF sequence on both probe
states.  It cancels in the overlap defining the modular commutator, giving
\begin{equation}
J(A,B,C)_{\widetilde\Psi(t)}
=
J(A,B,C)_{\Psi'(t)}.
\label{eq:J-common-unitary}
\end{equation}

For the Hall overlap, the second probe is the on-site charge rotation
$U_{BC}(y)=\exp(i yQ_{BC})$.  By \cref{lem:onsite-covariance}, it moves through the IMF sequence.  The
$D$ and $DE$ marginals are disjoint from $BC$, while exact symmetry gives
$[\rho_R,Q_R]=0$ for $R=M,C$; hence the same charge rotation also leaves the
$M$ and $C$ marginals unchanged.  Thus it does not change any reduced state
needed to construct $V(t)$.  Therefore
\begin{equation}
U_{BC}(y)\ket{\widetilde\Psi(t)}
=
V(t)U_{BC}(y)\ket{\Psi'(t)},
\end{equation}
and the same cancellation gives
\begin{equation}
\Sigma(A,B,C)_{\widetilde\Psi(t)}
=
\Sigma(A,B,C)_{\Psi'(t)}.
\label{eq:Sigma-common-unitary}
\end{equation}
No approximate Markov assumption is used in
\eqref{eq:J-common-unitary}--\eqref{eq:Sigma-common-unitary}; all approximation
has already been isolated in $r(t)$.

\subsection{Continuity and deformation moduli for the two responses}
\label{sec:response-moduli}

\subsubsection{Modular commutator}
We use $\Lambda_{X:Y:Z}$ as defined in \eqref{eq:Lambda-main}, with each
ordered deformation tripartition written explicitly in the subscript.
Assume henceforth
that $\eps_{A_1}^{(J)}<\frac12$.

\begin{proposition}[Modular commutator deformation modulus]
\label{prop:J-deformation}
Both on the initial state and on $\ket{\Psi'(t)}$,
\begin{equation}
\abs{J(A,B,C)-J(A,B,C')}
\le
\mathcal D_J(\eps_{A_1}^{(J)}).
\label{eq:J-deformation-bound}
\end{equation}
Here, for $0\le\eps<\frac12$,
\begin{equation}
\begin{aligned}
\mathcal D_J(\eps)
:={}&
2\log d_{AB}\Lambda_{B_1C'_1:B_2C'_2:M}(\eps)
+2(\log d_{B_2C'_2M}+\log d_{B_2C'_2})\Lambda_{A:B_1:B_2}(\eps)\\
&\hspace{-12pt}+2\min\{\log d_B,\log d_{B_2C'_2M}\}\Lambda_{B_1:B_2:C'_2M}(\eps)
+2\min\{\log d_B,\log d_{B_2C'_2}\}\Lambda_{B_1:B_2:C'_2}(\eps).
\end{aligned}
\label{eq:D-J}
\end{equation}
Moreover,
\begin{equation}
\abs{J(A,B,C)_{\Psi(t)}-J(A,B,C)_\Psi}
\le
\Omega_J(r(t);A,B,C)+2\mathcal D_J(\eps_{A_1}^{(J)}).
\label{eq:J-orbit-comparison}
\end{equation}
\end{proposition}

\begin{proof}
Although this deformation can formally be expressed using a single modular
defect, bounding it would require control of $I(B:M\mid C')$, which is not
provided by the\Aone{} condition because $B$ and $M$ contact each other.  The refinements $B=B_1B_2$ and $C'=C'_1C'_2$ resolve this
junction into the admissible local tripartitions above.
Direct expansion gives
\begin{align}
J(A,B,C)-J(A,B,C')
={}&
i\langle[K_{AB},K_{BC}-K_{BC'}]\rangle
\notag\\
={}&
-i\langle[K_{AB},\Dlt_{B_1C'_1:B_2C'_2:M}]\rangle
\notag\\
&-i\langle[\Dlt_{A:B_1:B_2},K_{B_2C'_2M}-K_{B_2C'_2}]\rangle
\notag\\
&+i\langle[K_B,K_{B_2C'_2M}-K_{B_2C'_2}]\rangle.
\label{eq:J-deformation-identity}
\end{align}
The weighted commutator inequality
\[
|i\langle[P,Q]\rangle|
\le2\norm{P\ket\phi}_2\norm{Q\ket\phi}_2
\]
combined with \cref{thm:main-cmi} and
\eqref{eq:universal-modular-moment-main} bounds the first two terms of
\eqref{eq:J-deformation-identity} by the first two terms in
\eqref{eq:D-J}.  The remaining term is
\[
i\langle[K_B,K_{B_2C'_2M}-K_{B_2C'_2}]\rangle
=J(B_1,B_2,C'_2M)-J(B_1,B_2,C'_2).
\]
Applying \cref{lem:small-J-main} to these two modular commutators gives
the last two terms in \eqref{eq:D-J}, with the smaller dimension factor
in each case.  On $\ket{\Psi'(t)}$, the needed CMIs are either unchanged or
bounded by enlarged CMIs containing $MD$.  For example,
\[
I(B_1C'_1:M|B_2C'_2)_{\Psi'(t)}
\le
I(B_1C'_1:MD|B_2C'_2)_{\Psi}
\le\eps_{A_1}^{(J)}.
\]
The other three cases are identical in structure.
The triangle inequality gives
\begin{align}
&\abs{J(A,B,C)_{\Psi(t)}-J(A,B,C)_\Psi}
\notag\\
&\quad\le
\abs{J(A,B,C)_{\Psi(t)}-J(A,B,C)_{\widetilde\Psi(t)}}
\notag\\
&\qquad+
\abs{J(A,B,C)_{\widetilde\Psi(t)}-J(A,B,C)_{\Psi'(t)}}
\notag\\
&\qquad+
\abs{J(A,B,C)_{\Psi'(t)}-J(A,B,C')_{\Psi'(t)}}
\notag\\
&\qquad+
\abs{J(A,B,C')_{\Psi'(t)}-J(A,B,C')_\Psi}
\notag\\
&\qquad+
\abs{J(A,B,C')_\Psi-J(A,B,C)_\Psi}.
\label{eq:J-orbit-triangle}
\end{align}
The first term is at most $\Omega_J(r(t);A,B,C)$ by
\cref{thm:modular-vector-continuity-phase} and the same inner-product
argument used in the proof of \cref{thm:J-trace-continuity-phase}.
The second term vanishes by the exact common-unitary
identity \eqref{eq:J-common-unitary}.  The third and fifth terms are each at
most $\mathcal D_J(\eps_{A_1}^{(J)})$ by
\eqref{eq:J-deformation-bound}.  Finally, the fourth term vanishes because
$\ket{\Psi'(t)}=\cI_{MD}(t)\ket\Psi$ is obtained by a unitary supported on
$MD$, which is disjoint from $ABC'$; hence
$\rho_{ABC'}(\Psi'(t))=\rho_{ABC'}(\Psi)$.  Substitution into
\eqref{eq:J-orbit-triangle} proves \eqref{eq:J-orbit-comparison}.
\end{proof}

For later integration, set
\begin{equation}
b_J
:=
\frac{
\log d_{BC}(\log d_{AB}+3)
+\log d_{AB}(\log d_{BC}+3)
}{\log d_{AB}+\log d_{BC}}.
\label{eq:bJ}
\end{equation}
For $0\le r\le2$, we then have
\begin{equation}
\Omega_J(r;A,B,C)
\le
2(\log d_{AB}+\log d_{BC})r
\left[b_J+2\log\frac1r\right].
\label{eq:Omega-J-scalar}
\end{equation}
The right-hand side is understood by continuity at $r=0$.

\subsubsection{Hall conductance estimator}

Assume the exact on-site $U(1)$ symmetry in \eqref{eq:U1-main}.
For this response, assume $\eps_{A_1}^{(\Sigma)}<\frac12$.
By \cref{lem:onsite-covariance}, all IMF states considered above remain in
the same total-charge sector.  Recall the Hall conductance estimator is
\begin{equation}
\Sigma(A,B,C)
:=
\frac{i}{2}\langle[K_{AB},Q_{BC}^2]\rangle
=
i\langle[K_{AB},Q_BQ_C]\rangle.
\label{eq:Sigma-definition}
\end{equation}

And \cref{thm:modular-vector-continuity-phase} gives the sharper same-purification continuity bound
\begin{equation}
\abs{\Sigma_\phi-\Sigma_\chi}
\le
2\norm{Q_BQ_C}_\infty
\left[\vartheta_{AB}(r)+\log d_{AB}r\right],
\qquad r\le2.
\label{eq:Sigma-vector-continuity}
\end{equation}
Indeed, write
$\Sigma=-2\operatorname{Im}\braket{K_{AB}\phi}{Q_BQ_C\phi}$,
add and subtract the mixed inner product, and apply Cauchy--Schwarz.  Define
\begin{equation}
b_\Sigma
:=2\log d_{AB}+3.
\label{eq:bSigma}
\end{equation}
With \eqref{eq:bSigma}, one then has, for $0\le r\le2$,
\begin{equation}
\Omega_\Sigma(r;A,B,C)
\le
2\norm{Q_BQ_C}_\infty r
\left[b_\Sigma+2\log\frac1r\right].
\label{eq:Omega-Sigma-scalar}
\end{equation}

The Hall conductance estimator deformation $C\leftrightarrow C'$ has a particularly short direct
bound, consistent with the general one-step formula in
\cref{thm:main-Sigma}.

\begin{proposition}[Hall conductance estimator deformation modulus]
\label{prop:Sigma-deformation}
Both on the initial state and on $\ket{\Psi'(t)}$,
\begin{equation}
\abs{\Sigma(A,B,C)-\Sigma(A,B,C')}
\le
\mathcal D_\Sigma(\eps_{A_1}^{(\Sigma)}),
\label{eq:Sigma-deformation-bound}
\end{equation}
where
\begin{equation}
\mathcal D_\Sigma(\eps)
:=
2\norm{Q_AQ_M}_\infty
\left[
\Lambda_{A:B_1:B_2}(\eps)
+
\Lambda_{B_1:A:M}(\eps)
\right].
\label{eq:D-Sigma}
\end{equation}
Moreover,
\begin{equation}
\abs{\Sigma(A,B,C)_{\Psi(t)}-\Sigma(A,B,C)_\Psi}
\le
\Omega_\Sigma(r(t);A,B,C)+2\mathcal D_\Sigma(\eps_{A_1}^{(\Sigma)}).
\label{eq:Sigma-orbit-comparison}
\end{equation}
\end{proposition}

\begin{proof}
Retain the split $B=B_1B_2$.  By
\cref{fact:basic} in \cref{app:Hall}, $[K_{AB},Q_{AB}]=0$.  Since
$Q_B=Q_{AB}-Q_A$,
\begin{align}
\Sigma(A,B,C)-\Sigma(A,B,C')
&=
i\langle[K_{AB},Q_BQ_M]\rangle
\notag\\
&=
-i\langle[K_{AB},Q_AQ_M]\rangle
\notag\\
&=
i\langle[\Dlt_{A:B_1:B_2}-\Dlt_{B_1:A:M},Q_AQ_M]\rangle.
\label{eq:Sigma-deformation-identity}
\end{align}
In the last line, $K_{AB}$ was replaced by
$K_{AB_1}+K_B-K_{B_1}-\Dlt_{A:B_1:B_2}$, while the remaining
$K_{AB_1}$ contribution was converted to the modular Markov defect
$\Dlt_{B_1:A:M}$.  The weighted commutator inequality together with
\cref{thm:main-cmi} proves \eqref{eq:Sigma-deformation-bound} on the initial
state.  The same bound holds on $\ket{\Psi'(t)}$: the first CMI is
unchanged because the $AB$ reduced state is unchanged, while
\[
I(B_1:M|A)_{\Psi'(t)}
\le
I(B_1:MD|A)_{\Psi}
\le\eps_{A_1}^{(\Sigma)}.
\]
Using \eqref{eq:Sigma-common-unitary} and follows the same series of triangle inquality \eqref{eq:J-orbit-triangle}. Utilize the fact that reduced state on $ABC'$ is unchanged and the two deformation bounds prove \eqref{eq:Sigma-orbit-comparison}.
\end{proof}

\subsection{A generic weighted drift obstruction}
\label{sec:generic-obstruction}

The modular commutator and Hall conductance estimator proofs now differ only in the scalar quantity that
has a finite drift budget.  We isolate the common bound.
Fix one response $\mathsf P\in\{J,\Sigma\}$.  In this subsection only,
write $\eps_{A_1}:=\eps_{A_1}^{(\mathsf P)}$ and take
$\mathcal D=\mathcal D_{\mathsf P}$ and $\Omega=\Omega_{\mathsf P}$.
The weighted estimate \eqref{eq:r-weighted} holds with this choice;
no supremum over the other response's deformation geometries is required for this step.

Write $\mathsf P(t):=\mathsf P(A,B,C)_{\Psi(t)}$ for the chosen response
along the modular-flow trajectory, and suppose it satisfies
\begin{equation}
|\mathsf P(t)-\mathsf P(0)|
\le
\Omega(r(t))+2\mathcal D(\eps_{A_1}),
\label{eq:generic-orbit-comparison}
\end{equation}
with
\begin{equation}
\Omega(r)\le ar\left[b+2\log\frac1r\right].
\label{eq:generic-Omega}
\end{equation}
For the modular-commutator application,
$a=2(\log d_{AB}+\log d_{BC})$ and $b=b_J$, with $b_J$ defined in
\eqref{eq:bJ}.  For the Hall application,
$a=2\norm{Q_BQ_C}_\infty$ and $b=b_\Sigma$, with $b_\Sigma$ defined in
\eqref{eq:bSigma}.

\begin{lemma}[Integrated continuity-modulus bound]
\label{lem:integrated-Omega-bound}
Assume $a>0$, $b\ge0$, and \eqref{eq:generic-Omega}.  Define
\begin{equation}
G(T)
:=
\frac1\pi\int_0^T
\ell(t)^6[\cosh(2\pi t)-1]\,dt.
\label{eq:G-definition}
\end{equation}
Then \eqref{eq:r-weighted} implies
\begin{equation}
\int_0^T r(t)\,dt
\le
2\sqrt{\eps_{A_1}G(T)}.
\label{eq:Q-bound}
\end{equation}
If $2\sqrt{\eps_{A_1}G(T)}\le T$, then
\begin{equation}
\int_0^T\Omega(r(t))\,dt
\le
2a\sqrt{\eps_{A_1}G(T)}
\left[
 b+2\log\frac{T}{2\sqrt{\eps_{A_1}G(T)}}+\frac2e
\right].
\label{eq:integrated-Omega}
\end{equation}
If $2\sqrt{\eps_{A_1}G(T)}>T$, then
\begin{equation}
\eps_{A_1}>\frac{T^2}{4G(T)}.
\label{eq:large-average-failure}
\end{equation}
\end{lemma}

\begin{proof}
The weight in \eqref{eq:r-weighted} is even.  Cauchy--Schwarz therefore gives
\begin{align}
\int_0^T r(t)\,dt
&\le
\left(
\int_0^T
\frac{r(t)^2}{\ell(t)^6[\cosh(2\pi t)-1]}\,dt
\right)^{1/2}
\left(
\int_0^T
\ell(t)^6[\cosh(2\pi t)-1]\,dt
\right)^{1/2}
\notag\\
&\le2\sqrt{\eps_{A_1}G(T)},
\end{align}
which proves \eqref{eq:Q-bound}.

For $x\ge0$ and $0<\alpha\le1$,
\begin{equation}
x\log\frac1x
\le
x\log\frac1\alpha+\frac{\alpha}{e}.
\label{eq:xlog-inequality}
\end{equation}
This follows from $x\log(\alpha/x)\le\alpha/e$ for $x>0$,
with $x\log(1/x)$ understood by continuity at $x=0$.
If $0<2\sqrt{\eps_{A_1}G(T)}\le T$, apply
\eqref{eq:xlog-inequality} pointwise with
$\alpha=2\sqrt{\eps_{A_1}G(T)}/T$.  Using \eqref{eq:Q-bound}, integration
of \eqref{eq:generic-Omega} yields
\begin{align}
\int_0^T\Omega(r(t))\,dt
&\le
a\left(\int_0^T r(t)\,dt\right)
\left[b+2\log\frac{T}{2\sqrt{\eps_{A_1}G(T)}}\right]
+\frac{4a\sqrt{\eps_{A_1}G(T)}}{e}
\notag\\
&\le
2a\sqrt{\eps_{A_1}G(T)}
\left[
b+2\log\frac{T}{2\sqrt{\eps_{A_1}G(T)}}+\frac2e
\right],
\end{align}
which is \eqref{eq:integrated-Omega}.  The case $\eps_{A_1}=0$ follows by
continuity.  Finally, $2\sqrt{\eps_{A_1}G(T)}>T$ rearranges directly to
\eqref{eq:large-average-failure}.
\end{proof}

\begin{lemma}[Explicit bound on the integrated strip weight]
\label{lem:G-upper}
For $T\ge0$, the function $G$ defined in \eqref{eq:G-definition} satisfies
\begin{equation}
G(T)
\le
\frac{2}{\pi^2}(1+4T^2)^3\e^{2\pi T}.
\label{eq:G-upper}
\end{equation}
\end{lemma}

\begin{proof}
Using $\ell(t)^6=8(1+4t^2)^3$ and
$\cosh(2\pi t)-1\le\frac12\e^{2\pi t}$ for $t\ge0$, we obtain
\begin{align*}
G(T)
&\le\frac4\pi(1+4T^2)^3\int_0^T\e^{2\pi t}\,dt\\
&\le\frac{2}{\pi^2}(1+4T^2)^3\e^{2\pi T}.
\end{align*}
\end{proof}

\begin{theorem}[Generic spectral-free drift obstruction]
\label{thm:generic-obstruction}
Assume that, for some $s>0$ and $R\ge0$,
\begin{equation}
sT-R
\le
\int_0^T
|\mathsf P(t)-\mathsf P(0)|\,dt
\label{eq:drift-budget}
\end{equation}
whenever $T>R/s$.  Let $a,b,\mathcal D, \mathsf P$ be as in
\eqref{eq:generic-orbit-comparison}--\eqref{eq:generic-Omega}, with $b\ge3$.
Then, for every $T>R/s$,
\begin{equation}
\eps_{A_1}
\ge
\min\left\{
\mathcal D^{\leftarrow}\!\left(\frac{sT-R}{4T}\right),
\frac{T^2}{4G(T)}
\chi_{b+2/e}\!\left(\frac{sT-R}{2aT}\right)^2
\right\}.
\label{eq:generic-lower-bound}
\end{equation}
where
\begin{equation}
\mathcal D^{\leftarrow}(y)
:=\inf\{\eps\ge0:\mathcal D(\eps)\ge y\},
\qquad
\chi_\gamma(z):=
\begin{cases}
0,&z=0,\\[0.2em]
-\displaystyle\frac{z}{
2W_{-1}\!\left(-\frac z2\e^{-\gamma/2}\right)},&0<z\le \gamma,\\[0.8em]
1,&z>\gamma,
\end{cases}
\qquad \gamma\ge2.
\label{eq:generic-bound-definitions}
\end{equation}
Here $\inf\varnothing:=+\infty$.  More explicitly, for each $\gamma\ge2$,
$\chi_\gamma:[0,\infty)\to[0,1]$ is defined as follows: for
$0<z\le\gamma$, $\chi_\gamma(z)$ is the unique $x\in(0,1]$ satisfying
\begin{equation}
x\left[\gamma+2\log\frac1x\right]=z.
\end{equation}
Uniqueness follows because the left-hand side is strictly increasing on
$[0,1]$, with endpoint values $0$ (by continuity) and $\gamma$.
We set $\chi_\gamma(0)=0$ and extend it by $\chi_\gamma(z)=1$ for
$z>\gamma$.  The formula in \eqref{eq:generic-bound-definitions} expresses
this inverse using the real Lambert-$W$ branch $W_{-1}\le-1$;
its scaling is given in \cref{lem:Lambert-W-scaling}.
\end{theorem}

\begin{proof}
If $2\sqrt{\eps_{A_1}G(T)}>T$,
\eqref{eq:large-average-failure} implies the second branch because
$\chi_{b+2/e}\le1$.  Suppose
$2\sqrt{\eps_{A_1}G(T)}\le T$.  Combining
\eqref{eq:drift-budget}, \eqref{eq:generic-orbit-comparison}, and
\eqref{eq:integrated-Omega} gives
\begin{align}
sT-R
\le{}
2a\sqrt{\eps_{A_1}G(T)}
\left[
 b+2\log\frac{T}{2\sqrt{\eps_{A_1}G(T)}}+\frac2e
\right] +2T\mathcal D(\eps_{A_1}).
\label{eq:generic-master}
\end{align}
At least one term is at least $(sT-R)/2$.  If it is the deformation term,
the first branch of \eqref{eq:generic-lower-bound} follows.  Otherwise,
\[
\frac{2\sqrt{\eps_{A_1}G(T)}}{T}
\left[
b+\frac2e
+2\log\frac{T}{2\sqrt{\eps_{A_1}G(T)}}
\right]
\ge
\frac{sT-R}{2aT}.
\]
The defining inverse relation for $\chi_{b+2/e}$ gives
\[
\frac{2\sqrt{\eps_{A_1}G(T)}}{T}
\ge
\chi_{b+2/e}\!\left(\frac{sT-R}{2aT}\right),
\]
which is the second branch.
\end{proof}

For an explicit bound, one can choose
\begin{equation}
T_*:=\frac Rs+\frac1\pi.
\label{eq:Tstar-generic}
\end{equation}
so that the leading factor
$(sT-R)^2\e^{-2\pi T}$ is maximized. Then $sT_*-R=s/\pi$, and \cref{thm:generic-obstruction,lem:G-upper} give
\begin{equation}
\begin{aligned}
\eps_{A_1}
\ge
\min\Bigg\{
\mathcal D^{\leftarrow}\!\left(\frac{s}{4\pi T_*}\right),
\frac{\pi^2\e^{-2}T_*^2}{8(1+4T_*^2)^3}
\chi_{b+2/e}\!\left(\frac{s}{2\pi aT_*}\right)^2
\exp\!\left(-\frac{2\pi R}{s}\right)
\Bigg\}.
\end{aligned}
\label{eq:generic-optimized}
\end{equation}

\subsection{Modular commutator: lower bound from entropy drift}
\label{sec:thermal-obstruction}

Let $F(t)$ denote the entropy itself along the modular-flow trajectory:
\begin{equation}
F(t):=S_{BC}(\Psi(t)),
\qquad
F'(t)=J(A,B,C)_{\Psi(t)}.
\label{eq:entropy-derivative}
\end{equation}
The key input is simply that the entropy stays within its allowed range,
\begin{equation}
0\le F(t)\le\log d_{BC}.
\label{eq:thermal-entropy-range}
\end{equation}
A constant nonzero modular commutator would make $F(t)$ grow or decrease
linearly until it leaves this interval.  Its deviation from the initial
value must therefore be large enough to prevent this.  To obtain the
sharpest entropy budget from this range, use the distance from $F(0)$ to
the nearer endpoint:
\begin{equation}
R_{BC}
:=
\min\{S_{BC}(\Psi),\log d_{BC}-S_{BC}(\Psi)\}.
\label{eq:RBC}
\end{equation}
Assume $J(A,B,C)_\Psi\ne0$ and choose the direction of modular time so
that the initial drift points toward that nearer endpoint (either one
in case of a tie).  Continue to write $\Psi(t)$ and $F(t)$ for the
oriented trajectory and its entropy.  Reversing time changes the sign
of the derivative in \eqref{eq:entropy-derivative}, but in either direction
\[
\abs{F'(t)-F'(0)}
=\abs{J(A,B,C)_{\Psi(t)}-J(A,B,C)_\Psi}.
\]
If the nearer endpoint is $\log d_{BC}$, then
$F'(0)=\abs{J(A,B,C)_\Psi}$ and $R_{BC}=\log d_{BC}-F(0)$.
Using $F(T)\le\log d_{BC}$ directly gives
\begin{align*}
\abs{J(A,B,C)_\Psi}T-R_{BC}
&\le F'(0)T-\bigl[F(T)-F(0)\bigr]\\
&=\int_0^T\bigl[F'(0)-F'(t)\bigr]\,dt\\
&\le\int_0^T\abs{F'(t)-F'(0)}\,dt.
\end{align*}
If the nearer endpoint is $0$, then $F'(0)=-\abs{J(A,B,C)_\Psi}$
and $R_{BC}=F(0)$.  The lower bound $F(T)\ge0$ similarly gives
\[
\abs{J(A,B,C)_\Psi}T-R_{BC}
\le F(T)-F(0)-F'(0)T
\le\int_0^T\abs{F'(t)-F'(0)}\,dt.
\]
Hence, for
$T>R_{BC}/\abs{J(A,B,C)_\Psi}$,
\begin{equation}
\abs{J(A,B,C)_\Psi}T-R_{BC}
\le
\int_0^T
\abs{J(A,B,C)_{\Psi(t)}-J(A,B,C)_\Psi}\,dt.
\label{eq:thermal-budget}
\end{equation}
Thus \cref{thm:generic-obstruction} applies with
$\eps_{A_1}=\eps_{A_1}^{(J)}$ and
\begin{equation}
(s,R,a,b,\mathcal D)
=
\left(
\abs{J(A,B,C)_\Psi},R_{BC},
2(\log d_{AB}+\log d_{BC}),b_J,\mathcal D_J
\right).
\label{eq:thermal-substitution}
\end{equation}

\begin{corollary}[Spectral-free modular commutator \Aone{} obstruction]
\label{cor:thermal-obstruction}
\label{cor:thermal-exp-L}
If $J(A,B,C)_\Psi\ne0$, define
\begin{equation}
T_J^*:=
\frac{R_{BC}}{\abs{J(A,B,C)_\Psi}}+\frac1\pi,
\label{eq:TJstar}
\end{equation}
Then $\eps_{A_1}^{(J)}$ is lower bounded by
\begin{equation}
\min\Bigg\{
\mathcal D_J^{\leftarrow}\!\left(
\frac{\abs{J(A,B,C)_\Psi}}{4\pi T_J^*}
\right),
\frac{\pi^2\e^{-2}(T_J^*)^2}
{8[1+4(T_J^*)^2]^3}
\chi_{b_J+2/e}\!\left(
\frac{\abs{J(A,B,C)_\Psi}}
{4\pi(\log d_{AB}d_{BC})T_J^*}
\right)^2
\exp\!\left(
-\frac{2\pi R_{BC}}{\abs{J(A,B,C)_\Psi}}
\right)
\Bigg\}.
\label{eq:thermal-optimized}
\end{equation}
In particular, consider regular two-dimensional response regions of scale
$L$, fixed on-site dimension, and polynomially many deformation moves.
If $J(A,B,C)_\Psi=\pi c_-/3$ with fixed $c_-\ne0$, then there is a
positive polynomial $P_J$ such that, for sufficiently large $L$,
\begin{equation}
\eps_{A_1}^{(J)}(L)
\ge
\frac1{P_J(L,|c_-|^{-1})}
\exp\!\left[-\frac{6R_{BC}(L)}{|c_-|}\right].
\label{eq:thermal-exp-L}
\end{equation}
\end{corollary}

\begin{proof}
Insert the modular commutator substitution \eqref{eq:thermal-substitution}
into \eqref{eq:generic-optimized} to obtain \eqref{eq:thermal-optimized}.
For fixed local dimension, all modular-moment and logarithmic dimension factors are
polynomial in $L$.  Moreover, $R_{BC}\le\log d_{BC}=O(L^2)$, so
$T_J^*$ is at most polynomial in $L$ and $|c_-|^{-1}$.  The deformation modulus
$\mathcal D_J(\eps)$ is a polynomial times
$\sqrt{\eps[\poly(L)+\log(1/\eps)]}$, so its generalized inverse at the
inverse-polynomial target in \eqref{eq:thermal-optimized} is
inverse-polynomial up to logarithms.  The Lambert-$W$ factor is likewise
inverse-polynomial up to logarithms by
\cref{lem:Lambert-W-scaling}.  Absorbing these logarithmic losses into
$P_J$, both branches in \eqref{eq:thermal-optimized} are bounded below by
the right-hand side of \eqref{eq:thermal-exp-L}, since
$\exp[-6R_{BC}/|c_-|]\le1$.
\end{proof}

An entanglement area law $S_{BC}(\Psi)=O(L)$ yields an exponentially small lower bound on the largest \Aone{} violation among the partitions used in the proof, up to a polynomial prefactor.  At fixed $R_{BC}$, the factor $\exp[-6R_{BC}/|c_-|]$ increases with $|c_-|$, so greater net chirality strengthens the exponential part of the bound. Volume-law entanglement can permit a weaker bound of order $\exp[-O(L^2)]$ when $R_{BC}=\Theta(L^2)$; however, the physical meaning of the modular commutator in this regime is not yet established. The argument applies to every response partition $A,B,C$: at least one of
its finitely many associated \Aone{} partitions satisfies the lower bound. If finite modular commutator and the entanglement area law both hold uniformly, the same lower bound applies to at least one \Aone{} partition in each construction.  Repeating the construction for different partitions therefore gives infinitely many distinct \Aone{} partitions satisfying the bound, though not necessarily every possible \Aone{} partition.

\subsection{Hall conductance estimator: lower bound from charge fluctuations}
\label{sec:Hall-obstruction}

The exact drift relation \eqref{eq:exact-review-Hall-drift} uses the second
moment $\langle Q_{BC}^2\rangle$.  Since the mean is constant, we instead
use the variance drift \eqref{eq:mean-charge-constant} to obtain a sharper
quantitative bound. Define the initial regional charge fluctuation
\begin{equation}
V_{BC}
:=
\operatorname{Var}_\Psi(Q_{BC})
=
\langle Q_{BC}^2\rangle_0-\langle Q_{BC}\rangle_0^2.
\label{eq:VBC}
\end{equation}
Orient modular time so that the variance initially decreases,
and denote the oriented trajectory again by $\Psi(t)$.  Positivity of
$\operatorname{Var}_{\Psi(t)}(Q_{BC})$ gives, whenever
$T>V_{BC}/(2\abs{\Sigma(A,B,C)_\Psi})$,
\begin{equation}
\abs{\Sigma(A,B,C)_\Psi}T-\frac{V_{BC}}2
\le
\int_0^T
\abs{
\Sigma(A,B,C)_{\Psi(t)}-
\Sigma(A,B,C)_\Psi
}\,dt.
\label{eq:Hall-budget}
\end{equation}
Thus \cref{thm:generic-obstruction} applies with
$\eps_{A_1}=\eps_{A_1}^{(\Sigma)}$ and
\begin{equation}
(s,R,a,b,\mathcal D)
=
\left(
\abs{\Sigma(A,B,C)_\Psi},
\frac{V_{BC}}2,
2\norm{Q_BQ_C}_\infty,
b_\Sigma,
\mathcal D_\Sigma
\right).
\label{eq:Hall-substitution}
\end{equation}

\begin{corollary}[Spectral-free Hall conductance estimator \Aone{} obstruction]
\label{cor:Hall-obstruction}
\label{cor:Hall-exp-L}
If $\Sigma(A,B,C)_\Psi\ne0$, define
\begin{equation}
T_\Sigma^*
:=
\frac{V_{BC}}{2\abs{\Sigma(A,B,C)_\Psi}}+\frac1\pi,
\label{eq:TSigmastar}
\end{equation}
Then $\eps_{A_1}^{(\Sigma)}$ is lower bounded by
\begin{equation}
\min\Bigg\{
\mathcal D_\Sigma^{\leftarrow}\!\left(
\frac{\abs{\Sigma(A,B,C)_\Psi}}{4\pi T_\Sigma^*}
\right),
\frac{\pi^2\e^{-2}(T_\Sigma^*)^2}
{8[1+4(T_\Sigma^*)^2]^3}
\chi_{b_\Sigma+2/e}\!\left(
\frac{\abs{\Sigma(A,B,C)_\Psi}}
{4\pi\norm{Q_BQ_C}_\infty T_\Sigma^*}
\right)^2
\exp\!\left(
-\frac{\pi V_{BC}}{\abs{\Sigma(A,B,C)_\Psi}}
\right)
\Bigg\}.
\label{eq:Hall-optimized}
\end{equation}
In particular, consider regular two-dimensional response regions of scale
$L$, fixed on-site dimension, uniformly bounded on-site charge, and
polynomially many deformation moves.  If
$\Sigma(A,B,C)_\Psi=\sigma_{xy}$ with fixed $\sigma_{xy}\ne0$, then there
is a positive polynomial $P_\Sigma$ such that, for sufficiently large $L$,
\begin{equation}
\eps_{A_1}^{(\Sigma)}(L)
\ge
\frac1{P_\Sigma(L,|\sigma_{xy}|^{-1})}
\exp\!\left[-\frac{\pi V_{BC}(L)}{|\sigma_{xy}|}\right].
\label{eq:Hall-exp-L}
\end{equation}
\end{corollary}

\begin{proof}
Insert the Hall substitution \eqref{eq:Hall-substitution} into
\eqref{eq:generic-optimized} to obtain \eqref{eq:Hall-optimized}.
For regular two-dimensional regions,
$\log d_{AB}=O(L^2)$ and
$\norm{Q_BQ_C}_\infty,\norm{Q_AQ_M}_\infty=\poly(L)$ under the bounded
on-site-charge assumption.  Hence $b_\Sigma$ and all charge prefactors in
$\mathcal D_\Sigma$ are polynomial.  If $\norm{Q_v}_\infty\le q_0$, then
$V_{BC}\le\norm{Q_{BC}-\langle Q_{BC}\rangle_0\id}_\infty^2\le4q_0^2|BC|^2=O(L^4)$;
thus $T_\Sigma^*$ is also at most polynomial in $L$ and $|\sigma_{xy}|^{-1}$.
As in the modular commutator case, the generalized
inverse of $\mathcal D_\Sigma$ at the relevant inverse-polynomial target is
inverse-polynomial up to logarithms, and the Lambert-$W$ factor is
inverse-polynomial up to logarithms by \cref{lem:Lambert-W-scaling}.
Absorbing logarithmic losses into $P_\Sigma$ and using
$\exp[-\pi V_{BC}/|\sigma_{xy}|]\le1$ gives \eqref{eq:Hall-exp-L}
from both branches of \eqref{eq:Hall-optimized}.
\end{proof}
For gapped ground states with conserved total charge and short-range charge
correlations, bulk charge fluctuations obey a boundary law
\cite{Song2012BipartiteFluctuations,Estienne2022Cornering}, including in quantum
Hall states \cite{Estienne2022Cornering}.  Thus $V_{BC}(L)=O(L)$ yields an
exponentially small lower bound on at least one \Aone{} violation for each
response partition, up to a polynomial prefactor.  At fixed $V_{BC}$, the factor
$\exp[-\pi V_{BC}/|\sigma_{xy}|]$ increases with $|\sigma_{xy}|$, so a larger
Hall response strengthens the exponential part of the bound.  More generally,
if $V_{BC}=\Theta(L^2)$, the bound can be as weak as $\exp[-O(L^2)]$.

\section{Discussion}
\label{sec:discussion}

\subsection{Summary}
\label{subsec:discussion-summary}

This work develops a finite-dimensional, spectral-free stability theory for the modular commutator and Hall conductance estimators under approximate \Aone{}.  The basic technical result, \cref{thm:main-cmi}, shows that conditional mutual information controls the modular Markov defect in the right-GNS norm.  The bound depends only logarithmically on subsystem dimensions and requires no lower bound on the nonzero eigenvalues of any reduced state.

Applied to the local moves of the entanglement bootstrap, this control yields quantitative deformation invariance of the modular commutator \cref{thm:main-J} and, with exact on-site \(U(1)\) symmetry, the Hall conductance estimator \cref{thm:main-Sigma}.  Exponentially
decaying \Aone{} violations make the dependence on topology- and orientation-preserving choices of the response partition exponentially small, up to polynomial geometric factors.

We also establish trace-norm continuity of the modular commutator and the
Hall conductance estimator \cref{thm:J-trace-continuity-phase,thm:Sigma-trace-continuity-phase}.
Combining this continuity with deformation invariance and quasi-local
truncation establishes invariance of the thermodynamic response limits
along \Aone-preserving quasi-local unitary paths
\cref{cor:response-phase-invariance-final}.  Thus the
paper proves both partition stability within a state and conditional phase
stability between states, while making the additional \Aone-preservation and
symmetry assumptions explicit.

The finite-time instantaneous-modular-flow analysis provides the converse
constraint.  The nonlinear current-state IMF map is globally Lipschitz
\cref{lem:IMF-Lipschitz}, and its Markov factorization obeys a weighted
finite-time bound determined solely by the initial CMI
\cref{thm:weighted-IMF}.  The result applies to arbitrary, possibly
nonfaithful finite-dimensional states, as detailed in
\cref{app:nonfaithful-IMF}.  Combining this flow control with the bound on how much the response
changes during modular evolution, \eqref{eq:generic-orbit-comparison},
converts the strict no-go argument into quantitative lower bounds:
for every response partition with a nonzero modular commutator or Hall
conductance estimator, at least one of the \Aone{} partitions used in the
construction obeys the lower bound
\cref{cor:thermal-obstruction,cor:Hall-obstruction}.  Under entanglement area law
or boundary-scale charge-fluctuation, these become the
\(\e^{-O(L)}\) obstructions in
\eqref{eq:thermal-exp-L} and \eqref{eq:Hall-exp-L}.

Taken together, the results resolve the apparent tension between asymptotic
deformation invariance and the exact-\Aone{} chirality no-go theorem.
Physically relevant chiral states may have \Aone{} violations that vanish with
scale, so their response becomes asymptotically independent of the auxiliary
partition, but the \Aone{} violations used in each construction cannot all
lie below the corresponding response-dependent lower bound.

\subsection{Open questions}
\label{subsec:discussion-open-questions}
\subsubsection{Robustness of R\'enyi-like probes of chirality and replica obstructions}

A natural extension is to establish quantitative robustness for finite-replica
probes of the chiral central charge.  Gass and Levin \cite{Gass_2026}
introduced the R\'enyi-like probe
\begin{equation}
Z_{\alpha,\beta}(A,B,C)_\rho
:=\Tr\rho_{ABC}\rho_{AB}^{\alpha}\rho_{BC}^{\beta},
\qquad
\omega_{\alpha,\beta}
:=\frac{Z_{\alpha,\beta}}{|Z_{\alpha,\beta}|},
\label{eq:discussion-renyi-probe}
\end{equation}
for $\alpha,\beta>0$ and $Z_{\alpha,\beta}\ne0$.
Here $A,B,C$ form the same cyclically ordered three-sector disk partition used to define the modular commutator $J(A,B,C)$ in \cref{fig:intro-geometries}(a).
For positive integers $\alpha,\beta$, the overlap is an expectation value of regional permutation operators on $\alpha+\beta+1$ replicas.
The phase is expected to encode the chiral central charge through the relation proposed in Ref.~\cite[Eq.~(5)]{Gass_2026},
\begin{equation}
\omega_{\alpha,\beta}
\longrightarrow
\exp\!\left[
-\frac{\pi i c_-}{12}
\left(
\frac{\alpha}{\alpha+1}
+\frac{\beta}{\beta+1}
-\frac{\alpha+\beta}{\alpha+\beta+1}
\right)
\right].
\label{eq:discussion-renyi-central-charge}
\end{equation}
Sheffer et al.\
\cite{sheffer2026probingchiraltopologicalstates} developed a broader
permutation-defect construction whose finite-replica probes can extract not
only the chiral central charge and Hall conductance but also higher central
charges through the lens-space multi-entropy.
These works motivate an approximate-Markov
stability theory for such observables.

The first question is whether small local \Aone{} violations imply
quantitative deformation invariance of $\log\omega_{\alpha,\beta}$ at fixed
replica number. Can this control be combined with trace-norm continuity
to obtain a conditional phase-invariance theorem analogous to ours?
An essential distinction is between controlling the overlap and controlling
its phase: a small absolute error in $Z_{\alpha,\beta}$ need not give a
small error in $\log\omega_{\alpha,\beta}$ on the chosen branch when the
overlap magnitude is small.
Thus a useful theorem would require relative-error control or suitable
lower bounds on the overlap magnitude, uniformly over the deformations
and state comparisons under consideration.

A complementary question is whether a nontrivial finite-replica phase forces a quantitative lower bound on the local \Aone{} violation.  Such a result would provide a replica-based counterpart of the obstruction proved here.  It remains open whether the appropriate quantity is the von Neumann \Aone{} defect itself or a R\'enyi analogue adapted to the permutation construction.  In the latter case, one must
identify a suitable definition and establish the nonnegativity and Markov or recovery properties needed by the argument, rather than assume that replacing each entropy by a R\'enyi entropy suffices.

\subsubsection{Infinite local dimension}

Our robustness bounds assume a fixed finite local Hilbert-space dimension,
independent of system size.  A natural extension is to bosonic lattice
systems, where each site has an infinite-dimensional Hilbert space.
In this setting, a small change in trace norm can move probability into
arbitrarily high levels and produce a large change in entropy or modular
fluctuations.  Quantitative continuity therefore requires additional
restrictions, typically an energy constraint
\cite{Winter_2016,Shirokov2016InfiniteCorrelations}.

For such bosonic systems, modular Hamiltonians are generally unbounded.
A useful starting point is a
positive reference Hamiltonian \(H_R\) with a finite partition function at
every positive inverse temperature, together with a uniform energy bound:
\begin{equation}
\Tr\e^{-\beta H_R}<\infty\quad(\beta>0),
\qquad
\Tr\rho_RH_R\le E_R.
\label{eq:discussion-energy-constraint}
\end{equation}
The maximum entropy at energy at most \(E\) is then finite.  This energy-dependent entropy budget replaces the dimension-dependent budget in entropy continuity bounds \cite{Winter_2016}. This extension would allow the modular response formulas to be studied in bosonic systems without imposing an artificial cutoff on the local occupation number but a more physically relevant regime of bounded energy.  More fundamentally, this would test whether the tension between a nonzero chiral response and approximate local Markovianity persists when the local Hilbert space is infinite-dimensional but the available energy is finite.  

A possible route is to apply the present bounds in low-energy subspaces and then remove the cutoff.  The central challenge, however, is to replace dimension-dependent estimates by bounds controlled by the available energy.  In the converse direction, it is interesting to ask whether the instantaneous modular-flow argument can yield an energy-dependent lower bound on local \Aone{} violations for states with a nonzero chiral response. An analogous Hall obstruction would require suitable control of regional charge fluctuations.  Resolving these questions would extend both the robustness of the modular response formulas and the obstruction to local Markovianity to physically relevant bosonic states, placing the two directions of the present theory in a setting beyond finite local dimension.

\section*{Acknowledgments}

I am especially grateful to Bowen Shi and for fruitful discussions and comments on this work. I also thank Jong Yeon Lee, Jia Wang and colleagues for helpful conversations. I am supported by Taiwan-UIUC fellowship program.
\section*{AI Disclosure Statement}

Generative artificial intelligence, specifically OpenAI's GPT-5.6 Sol model, played a substantial role in the development of the theoretical results presented in this work. The model was used to derive mathematical results and assist in organizing and writing the manuscript. The author developed the phase-equivalence argument.  GPT-6 Astra is only used to refine the writing. The author independently checked the derivations and proofs, assessed the scope of the results, and verified the conclusions reported in the paper. The author takes full responsibility for the correctness and content of the manuscript.

\appendix
\crefalias{section}{appendix}
\crefname{appendix}{appendix}{appendices}
\Crefname{appendix}{Appendix}{Appendices}

\section{Trace, information, and commutator identities}
\label{app:identities}

\begin{lemma}[Trace distance to purification distance]
\label{lem:purification-trace-distance-app}
Let $\rho$ and $\sigma$ be finite-dimensional states and set
\begin{equation}
\eps:=\frac12\norm{\rho-\sigma}_1.
\end{equation}
There exist purifications $\ket\psi$ of $\rho$ and $\ket\phi$ of $\sigma$
on a common auxiliary system such that
\begin{equation}
\norm{\ket\psi-\ket\phi}_2\le\sqrt{2\eps}.
\label{eq:purification-trace-distance-app}
\end{equation}
\end{lemma}

\begin{proof}
Define the root fidelity
\begin{equation}
F_{\mathrm r}(\rho,\sigma)
:=\norm{\sqrt{\rho}\sqrt{\sigma}}_1.
\end{equation}
By Uhlmann's theorem, the purifications and their phases may be chosen so
that
\begin{equation}
\braket{\psi}{\phi}=F_{\mathrm r}(\rho,\sigma)\ge0.
\end{equation}
The Fuchs--van de Graaf inequality gives
$1-F_{\mathrm r}(\rho,\sigma)\le\eps$.  Since both purifications are
normalized,
\begin{align}
\norm{\ket\psi-\ket\phi}_2^2
&=2-2\operatorname{Re}\braket{\psi}{\phi}
\notag\\
&=2\bigl(1-F_{\mathrm r}(\rho,\sigma)\bigr)
\le2\eps.
\end{align}
Taking the square root proves
\eqref{eq:purification-trace-distance-app}.
\end{proof}

\subsection{The subnormalized Markov operator}

\begin{lemma}[Trace bound]
\label{lem:tau-trace-app}
Define the subnormalized Markov operator
\begin{equation}
\tau_{ABC}
:=\exp\!\left(\log\rho_{AB}+\log\rho_{BC}-\log\rho_B\right).
\label{eq:tau-main}
\end{equation}
Then
\begin{equation}
\Tr\tau_{ABC}\le1.
\label{eq:tau-trace-app}
\end{equation}
\end{lemma}

\begin{proof}
For full-rank marginals, Lieb's triple-matrix inequality \cite{Lieb1973} gives
\begin{align}
&\Tr\exp\!\left(\log\rho_{AB}-\log\rho_B+\log\rho_{BC}\right)
\notag\\
&\quad\le
\Tr\int_0^\infty
\rho_{AB}(\rho_B+t\id_B)^{-1}
\rho_{BC}(\rho_B+t\id_B)^{-1}\,dt.
\label{eq:triple-matrix-app}
\end{align}
All operators are embedded in $ABC$ in the natural way.  Taking partial traces over $A$ and $C$ reduces the right-hand side to
\[
\Tr_B\int_0^\infty
\rho_B^2(\rho_B+t\id_B)^{-2}\,dt
=\Tr_B\rho_B=1.
\]
The general case follows by regularization and continuity \cite{SutterBertaTomamichel2017}.
\end{proof}

\begin{lemma}[Depolarized Markov spectral floor]
\label{lem:depolarized-markov-floor-app}
Let
\[
\rho_{ABC}^{(\eta)}
:=(1-\eta)\rho_{ABC}+\eta\frac{\id_{ABC}}{d_{ABC}},
\qquad 0<\eta\le1,
\]
and define
\[
\tau_\eta
:=\exp\!\left(
\log\rho_{AB}^{(\eta)}+
\log\rho_{BC}^{(\eta)}-
\log\rho_B^{(\eta)}
\right).
\]
Then
\begin{equation}
\rho_{ABC}^{(\eta)}\ge\frac{\eta}{d_{ABC}}\id_{ABC},
\qquad
\tau_\eta\ge\frac{\eta^2}{d_{AB}d_{BC}}\id_{ABC}.
\label{eq:depolarized-markov-floor-app}
\end{equation}
\end{lemma}

\begin{proof}
Every marginal is depolarized with the same parameter, so
\[
\rho_W^{(\eta)}
=(1-\eta)\rho_W+\eta\frac{\id_W}{d_W}
\ge\frac{\eta}{d_W}\id_W.
\]
Moreover, $\rho_B^{(\eta)}\le\id_B$.  After embedding all operators in $ABC$, operator monotonicity of the logarithm therefore gives
\begin{align*}
\log\rho_{AB}^{(\eta)}
&\ge\log\!\left(\frac{\eta}{d_{AB}}\right)\id_{ABC},
&
\log\rho_{BC}^{(\eta)}
&\ge\log\!\left(\frac{\eta}{d_{BC}}\right)\id_{ABC},
&
-\log\rho_B^{(\eta)}&\ge0.
\end{align*}
Hence the Hermitian exponent defining $\tau_\eta$ is bounded below by
\[
\log\!\left(\frac{\eta^2}{d_{AB}d_{BC}}\right)\id_{ABC}.
\]
Every eigenvalue of this Hermitian exponent is therefore at least
$\log[\eta^2/(d_{AB}d_{BC})]$.  Applying the scalar exponential to its
spectral decomposition gives the second inequality in
\eqref{eq:depolarized-markov-floor-app}.
\end{proof}

\section{Proof of the CMI-to-GNS theorem}
\label{app:cmi-proof}

\subsection{Depolarization smoothing}

\begin{proof}[Proof of \cref{lem:smoothing-app}]
Since $\sigma$ and
$\sigma^{(\eta)}=(1-\eta)\sigma+\eta\id/d$ commute, they have a common
orthonormal eigenbasis.  Let $p_i$ be the eigenvalues of $\sigma$ and set
\[
q_i:=(1-\eta)p_i+\frac{\eta}{d}.
\]
Then
\[
\Tr\sigma\bigl(\log\sigma-\log\sigma^{(\eta)}\bigr)^2
=
\sum_{i:p_i>0}p_i\left(\log\frac{p_i}{q_i}\right)^2.
\]

First consider the indices for which $p_i\ge d^{-1}$.  Since
$p_i\ge q_i\ge(1-\eta)p_i$,
\[
0\le\log\frac{p_i}{q_i}\le-\log(1-\eta).
\]
Their total contribution is therefore at most $\log^2(1-\eta)$.  Direct
differentiation shows that $\log^2(1-\eta)/\eta$ is increasing on
$(0,1/2]$, so
\[
\log^2(1-\eta)
\le2\log^2(2)\,\eta
<\eta.
\]

Now consider the indices for which $0<p_i<d^{-1}$ and put
$x_i:=dp_i\in(0,1)$.  Their individual contributions satisfy
\[
p_i\left(\log\frac{p_i}{q_i}\right)^2
=
\frac{x_i}{d}
\log^2\!\left(1-\eta+\frac{\eta}{x_i}\right).
\]
Writing $u_i:=\eta/x_i$ and using $1-\eta+u_i\le1+u_i$ gives
\[
x_i\log^2\!\left(1-\eta+\frac{\eta}{x_i}\right)
\le
\eta\frac{\log^2(1+u_i)}{u_i}
\le\eta.
\]
In the last step we used
$\log(1+u)\le\sqrt{u}$ for $u>0$, which follows by differentiating
$\sqrt{u}-\log(1+u)$.  Since there are at most $d$ such indices, their total
contribution is at most $\eta$.  Adding the two regimes proves
\[
\Tr\sigma\bigl(\log\sigma-\log\sigma^{(\eta)}\bigr)^2
\le2\eta.
\]
\end{proof}

\subsection{A scalar second-moment estimate}

\begin{proof}[Proof of \cref{lem:scalar-moment-app}]
Set
\[
g(x):=\e^{-x}-1+x\ge0.
\]
The assumptions give
\[
\mathbb E g(X)
=\mathbb E\e^{-X}-1+\mathbb E X
\le u.
\]
For $x\le0$, the Taylor lower bound for $\e^{-x}$ gives $g(x)\ge x^2/2$.  For $x\ge0$,
\begin{align*}
g(x)
&=\int_0^x(1-\e^{-s})\,ds
\ge\int_0^x\frac{s}{1+s}\,ds
=x-\log(1+x)
\ge\frac{x^2}{2(1+x)}.
\end{align*}
Hence $x^2\le2(1+L)g(x)$ for every $x\in[-L,L]$.  Taking expectations proves the claim.
\end{proof}

\subsection{Continuity of conditional mutual information}

\begin{proof}[Proof of \cref{lem:cmi-continuity-app}]
Winter's finite-dimensional conditional-entropy continuity bound
\cite[Lemma~2]{Winter_2016} states that, whenever
$\frac12\norm{\rho_{AB}-\sigma_{AB}}_1\le\eta$,
\begin{equation}
\abs{S(A|B)_\rho-S(A|B)_\sigma}
\le2\eta\log d_A
+(1+\eta)h_2\!\left(\frac{\eta}{1+\eta}\right).
\label{eq:conditional-entropy-continuity-app}
\end{equation}
Trace distance is contractive under partial trace, so the same $\eta$ applies to all marginals below.  Using
\[
I(A:C|B)=S(A|B)-S(A|BC)
\]
and applying \eqref{eq:conditional-entropy-continuity-app} twice gives
\[
\abs{I(A:C|B)_\rho-I(A:C|B)_\sigma}
\le4\eta\log d_A
+2(1+\eta)h_2\!\left(\frac{\eta}{1+\eta}\right).
\]
For $0\le\eta\le\frac12$,
$(1+\eta)h_2(\eta/(1+\eta))\le2h_2(\eta)$, and hence the right-hand side is at most
$4\eta\log d_A+4h_2(\eta)$.  Repeating the argument with
$I(A:C|B)=S(C|B)-S(C|AB)$ gives the same estimate with $d_C$ in place of $d_A$.  Taking the better of the two bounds proves \eqref{eq:cmi-continuity-app} with $d_*=\min\{d_A,d_C\}$.
\end{proof}
\subsection{Evaluation of the numerical constant}
\label{app:cmi-constant-evaluation}

We record the elementary estimates used at the end of the proof of
\cref{thm:main-cmi}.  Recall that
\[
\Gamma=1+\log(d_{AB}d_{BC})+\log\frac1u,
\qquad
\eta=\frac{u}{\Gamma},
\qquad
0<u\le\frac12.
\]
Since $\log(d_{AB}d_{BC})\ge0$ and
$\log(1/u)\ge\log2$, one has
\[
\Gamma\ge1+\log2>\frac53,
\qquad
\eta\le\frac{1/2}{1+\log2}<\frac3{10}.
\]
Moreover,
\[
\log(d_{AB}d_{BC})
=\log d_A+2\log d_B+\log d_C
\ge2\log d_*,
\qquad
d_*=\min\{d_A,d_C\}.
\]
Using $h_2(\eta)\le\eta[1+\log(1/\eta)]$ in
\eqref{eq:smoothed-cmi-main}, together with
$\log(1/\eta)=\log(1/u)+\log\Gamma$, gives
\begin{align}
u_\eta
&\le
u+\frac{4u}{\Gamma}\log d_*
+\frac{4u}{\Gamma}
\left(1+\log\frac1u+\log\Gamma\right)
\notag\\
&\le
u+2u+4u\left(1+\frac{\log\Gamma}{\Gamma}\right)
\le
\left(7+\frac4{\e}\right)u
<9u.
\label{eq:ueta-constant-app}
\end{align}
Here we used $\sup_{x>0}\log(x)/x=1/\e$.  Next,
$2\log\Gamma\le\Gamma$, so
\begin{equation}
L_\eta
=2\log\frac1u+\log(d_{AB}d_{BC})+2\log\Gamma
\le3\Gamma.
\label{eq:Leta-constant-app}
\end{equation}
Consequently, $1+L_\eta\le4\Gamma$ and
$(1-\eta)^{-1/2}\le\sqrt{10/7}<6/5$.  Therefore
\begin{align}
4\sqrt{2\eta}
&\le\frac{12\sqrt2}{5}\sqrt{u\Gamma},
\notag\\
(1-\eta)^{-1/2}\sqrt{2(1+L_\eta)u_\eta}
&\le\frac{36\sqrt2}{5}\sqrt{u\Gamma}.
\end{align}
Adding the two estimates and using $48\sqrt2/5<14$ proves the constant in
\eqref{eq:combined-cmi-bound-main}.

\section{A cluster chain with a persistent \Aone{} defect}
\label{app:spurious-A1}

\subsection{Embedding of the chain in the \Aone{} partition}

Let
\[
    \Gamma=\{-2N,-2N+1,\ldots,2N\},
    \qquad N\geq 2,
\]
be a chain of \(4N+1\) qubits.  Embed \(\Gamma\) along one radial
\(B\)--\(D\) interface and assign its sites by
\begin{equation}
\begin{aligned}
    \Gamma\cap C&=\{-2N\},
    &\qquad \Gamma\cap B&=\{-2N+2,-2N+4,\ldots,2N-2\},\\
    \Gamma\cap E&=\{2N\},
    &\qquad \Gamma\cap D&=\{-2N+1,-2N+3,\ldots,2N-1\}.
\end{aligned}
    \label{eq:spurious-chain-C}
\end{equation}
Thus the inner endpoint lies in \(C\), the outer endpoint lies in
\(E\), the internal even sites lie in \(B\), and the odd sites lie in
\(D\).  Adjacent chain sites can be placed on opposite sides of the
\(B\)--\(D\) interface at uniformly bounded geometric distance.

A schematic version of the required placement is shown in
\cref{fig:spurious-A1-radial-chain}.  The omitted bulk sites may be put in
a product state; they make no contribution to the \Aone{} entropy
combination.

\begin{figure}[t]
\centering
\begin{tikzpicture}[
    scale=0.93,
    line cap=round,
    line join=round
]
    \def\R{4.25}
    \def\r{1.72}

    \draw[red,line width=1.25pt] (0,0) circle (\R);
    \draw[red,line width=1.25pt] (0,0) circle (\r);
    \draw[red,line width=1.25pt] (0,\r)--(0,\R);
    \draw[red,line width=1.25pt] (0,-\r)--(0,-\R);

    \node[red,font=\Large] at (-2.7,0.45) {$B$};
    \node[red,font=\Large] at (0,0) {$C$};
    \node[red,font=\Large] at (2.7,0.45) {$D$};
    \node[red,font=\large] at (1.05,-4.75) {$E$};

    \coordinate (q0)  at ( 0.00,-1.35);
    \coordinate (q1)  at ( 0.55,-1.78);
    \coordinate (q2)  at (-0.55,-2.08);
    \coordinate (q3)  at ( 0.55,-2.38);
    \coordinate (q4)  at (-0.55,-2.68);
    \coordinate (q5)  at ( 0.55,-2.98);
    \coordinate (q6)  at (-0.55,-3.28);
    \coordinate (q7)  at ( 0.55,-3.58);
    \coordinate (q8)  at (-0.55,-3.88);
    \coordinate (q9)  at ( 0.55,-4.0);
    \coordinate (q10) at ( 0.00,-4.65);

    \draw[black,line width=1.55pt]
       (q0)--(q1)--(q2)--(q3)--(q4)--(q5)--(q6)--(q7)--(q8)--(q9)--(q10);

    \foreach \j in {0,...,10}{
        \fill[black] (q\j) circle (3.2pt);
    }

    \node[anchor=east] at ($(q0)+(-0.12,0.04)$) {$-2N$};
    \node[anchor=west] at ($(q1)+(0.12,0.02)$) {$-2N+1$};
    \node[anchor=east] at ($(q2)+(-0.12,0)$) {$\cdots$};
    \node[anchor=west] at ($(q3)+(0.12,0)$) {$-1$};
    \node[anchor=east] at ($(q4)+(-0.12,0)$) {$0$};
    \node[anchor=west] at ($(q5)+(0.12,0)$) {$1$};
    \node[anchor=east] at ($(q6)+(-0.12,0)$) {$2$};
    \node[anchor=west] at ($(q7)+(0.12,0)$) {$3$};
    \node[anchor=east] at ($(q8)+(-0.12,0)$) {$\cdots$};
    \node[anchor=west] at ($(q9)+(0.12,0)$) {$2N-1$};
    \node[anchor=east] at ($(q10)+(-0.12,0)$) {$2N$};

\end{tikzpicture}
\caption{
The cluster chain is placed along a radial \(B\)--\(D\) interface.
Its inner endpoint \(-2N\) lies in \(C\), its outer endpoint \(2N\)
lies in \(E\), its internal even sites lie in \(B\), and its odd
sites lie in \(D\).  All sites carry the usual cluster-chain terms,
with the standard endpoint truncations.
}
\label{fig:spurious-A1-radial-chain}
\end{figure}
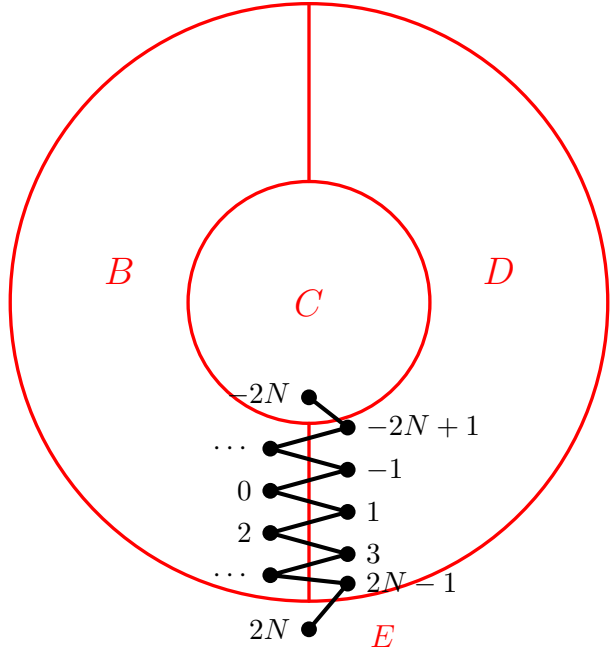

\subsection{Cluster Hamiltonian}

On the chain \(\Gamma\), consider
\begin{equation}
    H_\Gamma=-\sum_{i=-2N}^{2N}h_i.
    \label{eq:spurious-cluster-H}
\end{equation}
For every \(-2N+1\leq i\leq 2N-1\), let
\begin{equation}
    h_i=Z_{i-1}X_iZ_{i+1}.
    \label{eq:spurious-cluster-bulk}
\end{equation}
At the endpoints,
\begin{equation}
    h_{-2N}=X_{-2N}Z_{-2N+1},
    \qquad
    h_{2N}=Z_{2N-1}X_{2N}.
    \label{eq:spurious-cluster-endpoints}
\end{equation}
No bulk term is modified.  Equivalently, with
\(U_\Gamma:=\prod_{i=-2N}^{2N-1}\mathrm{CZ}_{i,i+1}\), one has
\(h_i=U_\Gamma X_iU_\Gamma^\dagger\) for every site, so
\(\lvert\psi_N\rangle=U_\Gamma\lvert+\rangle^{\otimes(4N+1)}\).
The operators \(h_i\) commute and obey \(h_i^2=\id\).  Their common
\(+1\) eigenspace is one-dimensional, so the ground-state density
operator is
\begin{equation}
    \rho_\Gamma
    =
    \lvert\psi_N\rangle\langle\psi_N\rvert
    =
    \frac{1}{2^{4N+1}}
    \prod_{i=-2N}^{2N}(\id+h_i).
    \label{eq:spurious-cluster-ground-state}
\end{equation}
To embed the example into a two-dimensional lattice, one may take
\begin{equation}
    H_{\mathrm{emb}}
    =
    H_\Gamma-\sum_{x\notin\Gamma}X_x,
    \qquad
    \lvert\Psi_N\rangle
    =
    \lvert\psi_N\rangle_\Gamma
    \otimes
    \lvert+\rangle_{\Gamma^c}.
    \label{eq:spurious-embedded-H}
\end{equation}
The off-chain product factors will be suppressed below.

We use the notation
\begin{equation}
    \str{n}{m}
    :=
    X_nX_{n+2}\cdots X_{m-2}X_m,
    \qquad
    n\equiv m\pmod 2.
    \label{eq:spurious-string-notation}
\end{equation}

\subsection{Reduced state after tracing out \(D\)}

Because \(D\cap\Gamma\) consists of all odd-numbered sites, tracing
out \(D\) retains only stabilizer products acting trivially on every odd
site.  No such product can contain an odd-site generator \(h_i\): its
\(X_i\) factor cannot be canceled by the neighboring generators, which
supply only \(Z_i\).  Among products of even-site generators, cancellation
of \(Z_i\) at each odd site requires either both adjacent generators or
neither.  Consequently, all even-site generators must be selected together,
or none of them.  The only surviving nonidentity stabilizer is therefore
\begin{equation}
    W:=\prod_{j=-N}^{N}h_{2j}
    =\str{-2N}{2N},
    \qquad W^2=\id,
    \qquad \Tr W=0.
    \label{eq:spurious-surviving-string}
\end{equation}
The reduced state on the \(2N+1\) even-site qubits in \(BCE\) is
\begin{equation}
    \rho_{BCE}
    =\Tr_D\rho_\Gamma
    =\frac{1}{2^{2N+1}}(\id+W).
    \label{eq:spurious-rho-BCE}
\end{equation}
This state has a single constraint: the product of all even-site
\(X\) eigenvalues is \(+1\). The string \(W\) acts nontrivially on both endpoints, in \(C\) and \(E\). Tracing out either endpoint therefore removes it, giving
\begin{equation}
    \rho_{BC}=\Tr_E\rho_{BCE}=\frac{\id_{BC}}{2^{2N}},
    \quad
    \rho_{BE}=\Tr_C\rho_{BCE}=\frac{\id_{BE}}{2^{2N}},
    \quad
    \rho_B=\Tr_{CE}\rho_{BCE}=\frac{\id_B}{2^{2N-1}}.
    \label{eq:spurious-rho-BC}
\end{equation}

\subsection{Entropy calculation}

The maximally mixed marginals immediately give
\begin{equation}
    S(BC)=2N\ln2,
    \qquad
    S(BE)=2N\ln2,
    \qquad
    S(B)=(2N-1)\ln2.
    \label{eq:spurious-SBC}
\end{equation}
Since \(W\) is a traceless Hermitian involution, its \(+1\) and \(-1\)
eigenspaces each have dimension \(2^{2N}\).  Thus
\eqref{eq:spurious-rho-BCE} has eigenvalue \(2^{-2N}\) with multiplicity
\(2^{2N}\), and all remaining eigenvalues are zero.  Hence
\begin{equation}
    S(BCE)=2N\ln2.
    \label{eq:spurious-SBCE}
\end{equation}

\begin{proposition}[Persistent spurious \Aone{} defect]
\label{prop:spurious-A1}
For the usual cluster-chain ground state embedded according to
\eqref{eq:spurious-chain-C},
\begin{equation}
    \delta_1(B,C,D)=I(C:E\mid B)=\ln2.
    \label{eq:spurious-exact-A1-defect}
\end{equation}
The defect is exactly one bit, independent of \(N\).
\end{proposition}

\begin{proof}
The global state on \(BCDE\) is pure, so
\(S(CD)=S(BE)\) and \(S(D)=S(BCE)\).  Therefore
\begin{align}
    \delta_1(B,C,D)
    &=S(BC)+S(BE)-S(B)-S(BCE)
    \nonumber\\
    &=\ln2.
    \label{eq:spurious-A1-defect-intermediate}
\end{align}
\end{proof}

The endpoint separation can grow proportionally to \(N\), while the
conditional mutual information remains \(\ln2\).  The surviving even-site
parity constraint explains this persistence: conditional on the internal
\(X\) eigenvalues in \(B\), the two endpoint eigenvalues share one bit of
correlation.  Nevertheless, \(U_\Gamma\) can be arranged in two layers of
nearest-neighbor controlled-\(Z\) gates, and \(H_\Gamma\) is unitarily
equivalent to \(-\sum_i X_i\), with a unique ground state and gap \(2\).
Thus the nonzero defect in this geometry already occurs for the ordinary
cluster chain, without any modified local term.

\subsection{The modified chain with a spurious modular commutator}
\label{app:modified-cluster-A1}

The ordinary stabilizer cluster state above has $J=0$; its \Aone{} defect
already demonstrates why finite-depth preparation need not preserve
admissibility.  To connect directly to the nonzero spurious $J$ of
Ref.~\cite{Gass_2024}, keep the same chain and \Aone{} partition, but replace
only the two central Hamiltonian terms by
\begin{equation}
  h_0=\frac{Z_{-1}X_0+Z_0Z_1+Z_{-1}Y_0Z_1}{\sqrt3},
  \qquad
  h_1=Z_{-1}X_0X_1Z_2.
\end{equation}
The remaining terms are unchanged.  Tracing out the odd sites gives the
reduced state calculated in Ref.~\cite{Gass_2024}, here labeled by $BCE$:
\begin{equation}
  \begin{gathered}
    \rho_{BCE}
    =2^{-(2N+1)}\left[\id+\frac{P_1+P_2+P_3}{\sqrt3}\right],\\
    P_1=\str{-2N}{0},\qquad
    P_2=Z_0\str{2}{2N},\qquad
    P_3=\str{-2N}{-2}Y_0\str{2}{2N}.
  \end{gathered}
  \label{eq:modified-cluster-rho-A1}
\end{equation}
The three Pauli strings anticommute pairwise, so their normalized sum
squares to $\id$.  Thus $S(BCE)=2N\log2$.  Tracing out $E$ retains only
$P_1$, tracing out $C$ retains only $P_2$, and tracing out both endpoints
leaves $B$ maximally mixed.  Writing
$q=(1+1/\sqrt3)/2$ and
$h_2(q)=-q\log q-(1-q)\log(1-q)$, we obtain
\begin{equation}
  S(BC)=S(BE)=(2N-1)\log2+h_2(q),
  \qquad S(B)=(2N-1)\log2.
\end{equation}
Purity therefore gives
\begin{equation}
  \delta_1(B,C,D)=I(C:E\mid B)
  =2h_2\!\left(\frac{1+1/\sqrt3}{2}\right)-\log2
  \approx0.338266>0,
  \label{eq:modified-cluster-A1-defect}
\end{equation}
independently of the chain length.  For the response partition of
Ref.~\cite{Gass_2024} (distinct from this \Aone{} partition), the same state
has the nonzero value
\begin{equation}
  J_{\mathrm{spur}}
  =\frac{1}{2\sqrt3}
  \left[\log\frac{1+1/\sqrt3}{1-1/\sqrt3}\right]^2.
\end{equation}
Hence the modified chain exhibits both a spurious modular commutator and a
persistent \Aone{} violation, and lies outside the uniformly admissible
state class.

\section{Bulk \Aone{} does not imply boundary \Aone{} after truncation}
\label{app:boundary-A1-counterexample}

We give a finite-range example in which the full trajectory is a stationary
product state, but spatial truncation creates a nondecaying \Aone{} defect
at the cut.  This separates bulk admissibility of a given trajectory from
the additional wall-crossing requirement in
\cref{def:A1-preserving-QLU-phase}.

\paragraph{Cancellation in the full generator.}
On the square lattice with the $\ell^1$ metric, take boxes
$\Lambda_N=[-N,N]^2\cap\mathbb Z^2$ and the truncation region
$\Omega_N=\Lambda_N\cap\{x\ge1\}$.  Place data qubits at
\begin{equation}
  v_j=\begin{cases}(1,j),&j\text{ even},\\(2,j),&j\text{ odd},\end{cases}
  \qquad w_j=(0,j),
  \label{eq:boundary-A1-sites}
\end{equation}
where $w_j$ are control qubits outside $\Omega_N$.  Set $g=\pi/4$ and
specify the interaction decomposition by
\begin{equation}
  \Phi(\{v_j,v_{j+1}\})=gZ_{v_j}Z_{v_{j+1}},
  \qquad
  \Phi(\{v_j,v_{j+1},w_j\})=-gZ_{v_j}Z_{v_{j+1}}Z_{w_j}.
  \label{eq:boundary-A1-interaction}
\end{equation}
Only terms contained in $\Lambda_N$ are included.  They commute, have
diameter at most $3$, and have uniformly bounded strength and incidence.
All polynomial interaction norms are therefore finite uniformly in $N$.
Prepare every data qubit in $\ket+$ and every other qubit in $\ket0$,
denoting the resulting product state by $\ket{\Psi_{0,N}}$.  Then
\begin{equation}
  G_{\Lambda_N}=g\sum_j Z_{v_j}Z_{v_{j+1}}(\id-Z_{w_j}),
  \qquad
  G_{\Lambda_N}\ket{\Psi_{0,N}}=0.
\end{equation}
Consequently $U_{\Lambda_N}(s)\ket{\Psi_{0,N}}=\ket{\Psi_{0,N}}$ for every
$s\in[0,1]$, and every bulk \Aone{} violation is exactly zero.

\paragraph{Truncation creates a boundary cluster chain.}
Spatial truncation is applied to the specified interaction terms in
\eqref{eq:boundary-A1-interaction}: it deletes every three-body term but
retains every two-body term.  Hence
\begin{equation}
  G_{\Lambda_N,\Omega_N}=\frac\pi4\sum_j Z_{v_j}Z_{v_{j+1}}.
\end{equation}
Using the two-qubit identity
\begin{equation}
  \e^{-i\pi Z_aZ_b/4}
  =\e^{i\pi/4}\e^{-i\pi Z_a/4}\e^{-i\pi Z_b/4}\mathrm{CZ}_{ab},
\end{equation}
we see that $U_{\Lambda_N,\Omega_N}(1)\ket{\Psi_{0,N}}$ is the cluster state on the data qubits, up to on-site unitaries and a global phase.
All other sites remain in a product state.  The on-site unitaries do not change subsystem entropies.

\paragraph{A wall-crossing \Aone{} geometry.}
Fix even $m\ge2$ and $N>3m+2$.  With $o=(3/2,0)$, let
$\mathbb B_r(o)=\{z\in\mathbb Z^2:\norm{z-o}_1\le r\}$ and define
\begin{equation}
  B_m=(\mathbb B_{3m+1/2}(o)\setminus C_m)\cap\{x<3/2\},\quad
  C_m=\mathbb B_{m+1/2}(o),\quad
  D_m=(\mathbb B_{3m+1/2}(o)\setminus C_m)\cap\{x>3/2\},\quad
\end{equation}
See
\cref{fig:boundary-A1-truncation} for the partition of the chain and the cancellation mechanism .

\begin{figure}[t]
\centering
\begin{minipage}[c]{0.47\textwidth}
\centering
\begin{tikzpicture}[x=0.9cm,y=0.9cm,font=\small,
  data/.style={circle,draw=blue!65!black,fill=blue!15,inner sep=2.8pt},
  control/.style={rectangle,draw,fill=white,inner sep=2.5pt}]
  \node[font=\small\bfseries] at (1.5,4.6) {(a) Removing the cancellation};
  \fill[blue!4] (0.5,1.55) rectangle (3.8,4.1);
  \draw[densely dashed,gray] (0.5,1.55)--(0.5,4.1);
  \node at (2.5,3.85) {$\Omega_N$};
  \fill[red!10] (-0.5,2.1)--(1.1,2.1)--(2.4,3.1)--cycle;
  \draw[red!65!black,dashed] (1.1,2.1)--(-0.5,2.1)--(2.4,3.1);
  \draw[blue!65!black,very thick] (1.1,2.1)--(2.4,3.1);
  \node[control,label=left:{$w_j$}] at (-0.5,2.1) {$0$};
  \node[data,label=below:{$v_j$}] at (1.1,2.1) {};
  \node[data,label=right:{$v_{j+1}$}] at (2.4,3.1) {};
  \node[red!65!black] at (1.4,1.05) {$-gZ_{v_j}Z_{v_{j+1}}Z_{w_j}$};
  \node at (1.4,0.5) {$Z_{w_j}\ket0=\ket0$: the two terms cancel};
  \draw[-{Latex},thick] (1.4,0.05)--(1.4,-0.8);
  \node[anchor=west,font=\footnotesize] at (1.65,-0.4) {truncate to $\Omega_N$};
  \fill[blue!4] (0.5,-3.0) rectangle (3.8,-1.0);
  \draw[densely dashed,gray] (0.5,-3.0)--(0.5,-1.0);
  \node[control,draw=gray,text=gray,label=left:{$w_j$}] at (-0.5,-2.5) {$0$};
  \draw[blue!65!black,very thick] (1.1,-2.5)--(2.4,-1.5);
  \node[data,label=below:{$v_j$}] at (1.1,-2.5) {};
  \node[data,label=right:{$v_{j+1}$}] at (2.4,-1.5) {};
  \node[blue!65!black] at (1.4,-3.45) {$+gZ_{v_j}Z_{v_{j+1}}$ remains};
\end{tikzpicture}
\end{minipage}\hfill
\begin{minipage}[c]{0.51\textwidth}
\centering
\begin{tikzpicture}[x=0.38cm,y=0.38cm,font=\small]
  \node[font=\small\bfseries] at (0,9.6) {(b) Boundary chain and \Aone{} regions};
  \fill[blue!12] (0,6.5)--(-6.5,0)--(0,-6.5)--cycle;
  \fill[orange!18] (0,6.5)--(6.5,0)--(0,-6.5)--cycle;
  \fill[green!15] (0,2.5)--(-2.5,0)--(0,-2.5)--(2.5,0)--cycle;
  \draw[thick] (0,6.5)--(-6.5,0)--(0,-6.5)--(6.5,0)--cycle;
  \draw[thick] (0,2.5)--(-2.5,0)--(0,-2.5)--(2.5,0)--cycle;
  \draw[gray,dotted] (0,6.5)--(0,2.5) (0,-2.5)--(0,-6.5);
  \draw[gray,densely dashed] (-1,-7.8)--(-1,7.8);
  \node[anchor=east,font=\footnotesize] at (-1.2,7.9) {cut $x=\tfrac12$};
  \foreach \siteindex in {-7,...,6} {
    \pgfmathsetmacro{\sitex}{0.5-mod(abs(\siteindex+1),2)}
    \pgfmathsetmacro{\nextx}{-\sitex}
    \draw[black!65,thick] (\sitex,\siteindex)--(\nextx,{\siteindex+1});
  }
  \foreach \siteindex in {-7,...,7} {
    \pgfmathsetmacro{\sitex}{0.5-mod(abs(\siteindex+1),2)}
    \def\sitecolor{gray!45}
    \ifnum\siteindex>-7
      \ifnum\siteindex<7
        \ifodd\siteindex\relax
          \def\sitecolor{orange!75}
        \else
          \def\sitecolor{blue!55}
        \fi
      \fi
    \fi
    \ifnum\siteindex>-3
      \ifnum\siteindex<3
        \def\sitecolor{green!55!black}
      \fi
    \fi
    \filldraw[fill=\sitecolor,draw=black!70] (\sitex,\siteindex) circle (0.13);
  }
  \node at (-3.5,0) {$B_m$};
  \node at (3.5,0) {$D_m$};
  \node at (1.45,0) {$C_m$};
  \node at (4.8,5.5) {$E_m$};
  \node at (0,-8.5) {$m=2,\quad \delta_1=2\log2$};
\end{tikzpicture}
\end{minipage}
\caption{Truncation destroys cancellation and creates a boundary \Aone{}
defect.  (a) The blue bond is the positive two-body term; the red triangle
indicates the negative three-body term involving the exterior control
$w_j$.  They cancel on the initial state with $w_j$ in $\ket0$.
Truncation deletes the three-body term, leaving the blue bond.
(b) The retained bonds generate a cluster chain, up to on-site unitaries.
For $m=2$, the inner and outer diamonds are $C_m$ and $B_mC_mD_m$; the annulus
is split into $B_m$ (blue) and $D_m$ (orange).   Each arm alternates between $D_m$ and $B_m$. Product-state sites are omitted.}
\label{fig:boundary-A1-truncation}
\end{figure}

We compute the entropies by counting stabilizers.  The product-state
qubits contribute no entropy, and the on-site unitaries may be omitted.
For a subsystem containing $n_T$ data qubits and $k_T$ independent
stabilizers supported entirely within it,
\begin{equation}
  S(T)=(n_T-k_T)\log2
\end{equation}  The stabilizer generators of the cluster-chain are
\begin{equation}
  \mathsf K_j=Z_{v_{j-1}}X_{v_j}Z_{v_{j+1}}.
\end{equation}
A product supported in $T$ can use only generators centered in $T$,
since the $X$ at a generator's center cannot be canceled by other
generators.  For each of the four regions below, no generator centered
on an annular arm can occur: the outermost such generator in a product
would leave an uncanceled $Z$ on the next site outward, which lies
outside $T$. Thus $B_m$ and $D_m$, each containing $2m$ data qubits, have no nontrivial
supported stabilizers.  The region $B_mC_m$ contains $4m+1$ data qubits
and has exactly the $2m-1$ independent generators
$\mathsf K_{-m+1},\ldots,\mathsf K_{m-1}$ supported in it.
For $C_mD_m$, the two additional generators $\mathsf K_{-m}$ and
$\mathsf K_m$ are supported because the adjacent annular sites belong to
$D_m$.  It therefore has $2m+1$ independent stabilizers on $4m+1$ data
qubits.  There is no extra independent global stabilizer: for example,
\begin{equation}
  \prod_{k=1}^{m}\mathsf K_{m+2k-1}
  =Z_{v_m}\left(\prod_{k=1}^{m}X_{v_{m+2k-1}}\right)Z_{v_{3m}}
\end{equation}
still acts on $v_{3m}\in B_m$, outside $C_mD_m$; the lower arm likewise
leaves a $Z$ at the distinct site $v_{-3m}\in B_m$.
The counts give
\begin{equation}
  \begin{array}{c|cccc}
    T&B_m&D_m&B_mC_m&C_mD_m\\\hline
    n_T&2m&2m&4m+1&4m+1\\
    k_T&0&0&2m-1&2m+1\\
    S(T)/\log2&2m&2m&2m+2&2m
  \end{array}
\end{equation}
Thus the truncated endpoint satisfies
\begin{equation}
  \delta_1(B_m,C_m,D_m)=2\log2
  \label{eq:boundary-A1-persistent-defect}
\end{equation}
Since $|C_m|=\Theta(m^2)$, this cannot obey a bound
$\delta_1\le P(|C_m|)\e^{-c\ell}$ with a fixed polynomial $P$ and $c>0$.
The failure comes from destroying cancellations at the cut, not from
long interaction tails or a bulk defect along the full trajectory.
This example concerns bulk preservation on the chosen input family; it
does not assert that the full unitary preserves \Aone{} on every admissible
input.  It suffices to show that uniform bulk \Aone{} for a full trajectory,
even together with finite-range dynamics, does not imply the boundary
admissibility required of its spatial truncations.

\section{Hall conductance estimator deformation identities}
\label{app:Hall}

\begin{fact}[Basic facts]
\label{fact:basic}
Let $\ket\Psi$ be a pure state satisfying the exact on-site $U(1)$ symmetry
\eqref{eq:U1-main}, and let $K_X=-\log\rho_X$ be the modular Hamiltonian of
its reduced state on $X$.  All unions below are unions of disjoint regions,
and regional operators are implicitly extended by the identity outside their
support.  Below we collect four basic facts about the local charge operators
and the modular Hamiltonians of the symmetric state:
\begin{enumerate}
\item $Q_{AB}=Q_A+Q_B$.

\item $[K_X,Q_X]=0$.  Moreover,
$[K_A,Q_{AB}]=0$ and $[K_A,Q_{AB}^2]=0$.

\item
$\langle[K_{AB},Q_A^2]\rangle
=\langle[K_A,Q_{AB}^2]\rangle=0$.

\item For a tripartition $ABC$ of a subregion of the full system,
\begin{equation}
\langle[K_{AB},Q_{BC}^2]\rangle
=\langle[K_{AB},Q_{\overline{BC}}^2]\rangle,
\label{eq:charge-complement-fact}
\end{equation}
where $\overline{BC}$ is the complement of $BC$ in the full system.
\end{enumerate}
These facts are taken from \cite[Appendix~C]{strict-J-2024}.
\end{fact}

\subsection{Deformation invariance of the Hall conductance estimator}

\begin{lemma}[On-site unitary covariance of IMF]
\label{lem:onsite-covariance}
Let $U=\prod_vU_v$ be any on-site unitary.  Then
\begin{equation}
U\cI_X(t)\ket\phi
=
\cI_X(t)U\ket\phi.
\label{eq:onsite-covariance}
\end{equation}
Consequently, $U$ can be moved freely through any sequence of IMF maps.  If
$U\ket\Psi$ differs from $\ket\Psi$ only by a phase, every state obtained
from $\ket\Psi$ by an IMF sequence has the same symmetry.
\end{lemma}

\begin{proof}
Write $U=U_XU_{\bar X}$.  The reduced state of $U\ket\phi$ on $X$ is
$U_X\rho_X(\phi)U_X^\dagger$.  Therefore
\begin{align*}
\cI_X(t)U\ket\phi
&=
(U_X\rho_XU_X^\dagger)^{i t}U_XU_{\bar X}\ket\phi\\
&=
U_XU_{\bar X}\rho_X^{i t}\ket\phi
=
U\cI_X(t)\ket\phi.
\end{align*}
Iteration proves the sequence statement.
\end{proof}

By \cref{fact:basic}, expanding $Q_{BC}^2=(Q_B+Q_C)^2$ in
\eqref{eq:intro-Sigma} shows that the $Q_B^2$ and $Q_C^2$ terms have zero
commutator expectation, and hence gives \eqref{eq:Sigma-cross-main}.

Now use
\[
K_{AB}=K_{ABC}-K_{BC}+K_B+\Dlt_{A:B:C}.
\]
The expectation of the commutator of $Q_BQ_C$ with $K_{ABC}$ or $K_{BC}$
vanishes by modular stationarity, while $[K_B,Q_BQ_C]=0$ by
\cref{fact:basic}.  This proves \eqref{eq:Hall-Markov-main}.  Applying
\eqref{eq:weighted-comm-main} and \cref{thm:main-cmi} proves
\eqref{eq:Hall-small-main}.

\subsection{Representative elementary moves}

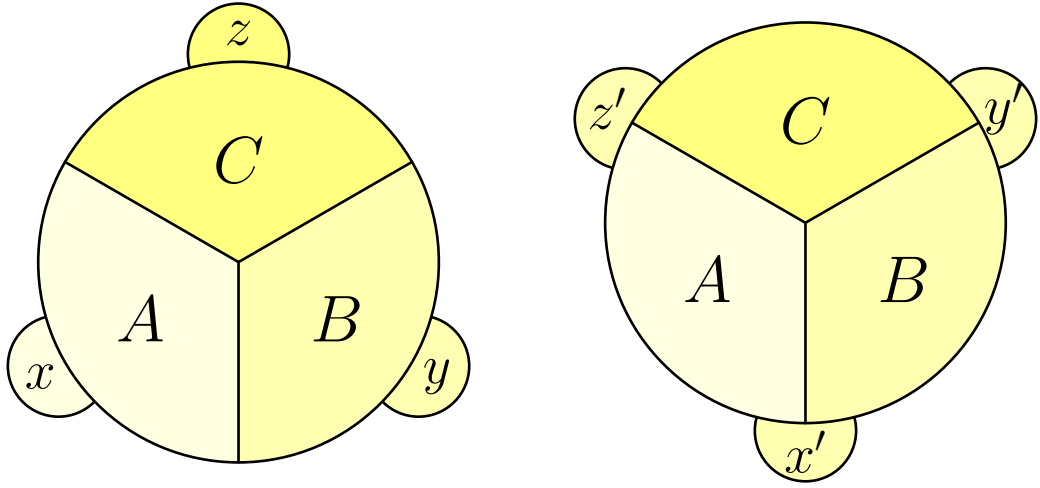
\begin{figure}[htbp]
\centering
\begin{tikzpicture}[line width=1pt,line cap=round,line join=round]
  \definecolor{sectorA}{RGB}{255,255,225}
  \definecolor{sectorB}{RGB}{255,255,178}
  \definecolor{sectorC}{RGB}{255,255,128}
  \coordinate (leftdisk) at (-3.75,-0.26);
  \coordinate (rightdisk) at (3.75,0.26);

  \foreach \diskcenter/\patchangle/\patchlabel/\patchcolor in {
    leftdisk/210/x/sectorA,
    leftdisk/330/y/sectorB,
    leftdisk/90/z/sectorC,
    rightdisk/150/z'/sectorB,
    rightdisk/30/y'/sectorB,
    rightdisk/270/x'/sectorB} {
    \filldraw[fill=\patchcolor]
      ($(\diskcenter)+(\patchangle:2.75)$) circle (0.67);
    \node[font=\fontsize{22}{24}\selectfont]
      at ($(\diskcenter)+(\patchangle:3.03)$) {$\patchlabel$};
  }

  \foreach \diskcenter in {leftdisk,rightdisk} {
    \begin{scope}[shift={(\diskcenter)}]
      \fill[sectorA] (0,0) -- (150:2.65)
        arc[start angle=150,end angle=270,radius=2.65] -- cycle;
      \fill[sectorB] (0,0) -- (270:2.65)
        arc[start angle=270,end angle=390,radius=2.65] -- cycle;
      \fill[sectorC] (0,0) -- (30:2.65)
        arc[start angle=30,end angle=150,radius=2.65] -- cycle;
      \draw (0,0) circle (2.65);
      \draw (150:2.65) -- (0,0) -- (30:2.65);
      \draw (0,0) -- (270:2.65);
      \node[font=\fontsize{32}{34}\selectfont] at (210:1.5) {$A$};
      \node[font=\fontsize{32}{34}\selectfont] at (330:1.5) {$B$};
      \node[font=\fontsize{32}{34}\selectfont] at (90:1.35) {$C$};
    \end{scope}
  }
\end{tikzpicture}
\caption{Three-sector disk partitions with boundary patches $x,y,z$
(left) and $x',y',z'$ (right).\textsuperscript{\ref{fn:boundary-patches-attribution}}}
\label{fig:three-sector-boundary-patches}
\end{figure}

The boundary patches in \cref{fig:three-sector-boundary-patches}\footnote{\label{fn:boundary-patches-attribution}The figure is adapted from \emph{Strict Area Law Entanglement versus Chirality}~\cite{strict-J-2024}.  The source figure image was supplied to a generative AI tool to recreate it in TikZ.} illustrate
the elementary moves considered below.  Unprimed patches meet a single
sector, while primed patches meet the outer endpoint of an interface between
two sectors.  For each move, let $X_m|Y_m|Z_m$ be the ordered tripartition
of the displayed CMI.  We write $\Dlt_{X|Y|Z}=\Dlt_{X:Y:Z}$ to display
this tripartition directly in the modular Markov defect.  The identities
have the uniform form
\begin{align}
\Sigma(\mathcal G_m')-\Sigma(\mathcal G_m)
&=-\frac{i}{2}\langle[\Dlt_{X_m|Y_m|Z_m},R_m]\rangle,
\notag\\
\abs{\Sigma(\mathcal G_m')-\Sigma(\mathcal G_m)}
&\le\Lambda_{X_m:Y_m:Z_m}\!\bigl(I(X_m:Z_m|Y_m)\bigr)
\norm{R_m\ket\Psi}_2.
\label{eq:Hall-generic-bound-app}
\end{align}
The six representative choices are
\begin{center}
\renewcommand{\arraystretch}{1.25}
\begin{tabularx}{\textwidth}{@{}l X X@{}}
\toprule
Move and CMI & Modular Markov defect & Charge polynomial $R_m$ \\
\midrule
$A\to Ax$, $I(x:B|A)$
& $\Dlt_{x|A|B}$
& $Q_{BC}^{2}$ \\
$B\to By$, $I(A:y|B)$
& $\Dlt_{A|B|y}$
& $Q_{BCy}^{2}-2Q_BQ_y$ \\
$C\to Cz$, $I(A:z|B)$
& $\Dlt_{A|B|z}$
& $-2Q_BQ_z$ \\
$A\to Ax'$, $I(x':C|AB)$
& $\Dlt_{x'|AB|C}$
& $-Q_{BC}^{2}$ \\
$B\to Bx'$, $I(x':C|AB)$
& $\Dlt_{x'|AB|C}$
& $2Q_{AB}Q_C-Q_{BC}^{2}$ \\
$C\to Cz'$, $I(z':B|A)$
& $\Dlt_{z'|A|B}$
& $2Q_AQ_{z'}$ \\
\bottomrule
\end{tabularx}
\end{center}
Cyclic permutations satisfy the same type of estimate.

\paragraph{$A\to Ax$.}
The defect relation
\[
K_{ABx}=K_{Ax}+K_{AB}-K_A-\Dlt_{x|A|B}
\]
shows that the $K_{Ax}$ and $K_A$ terms are disjoint from $Q_{BC}^2$.  Therefore
\[
\Sigma(Ax,B,C)-\Sigma(A,B,C)
=-\frac{i}{2}\langle[\Dlt_{x|A|B},Q_{BC}^2]\rangle.
\]

\paragraph{$B\to By$.}
Use
\[
K_{ABy}=K_{AB}+K_{By}-K_B-\Dlt_{A|B|y}.
\]
The $K_{By}$ and $K_B$ terms commute with $Q_{BCy}^2$.  Expanding the charge leaves
$i\langle[K_{AB},Q_BQ_y]\rangle$, which is the Hall conductance estimator for $A|B|y$.  Applying \eqref{eq:Hall-Markov-main} to that tripartition combines the terms into
\[
-\frac{i}{2}\langle[\Dlt_{A|B|y},Q_{BCy}^2-2Q_BQ_y]\rangle.
\]

\paragraph{$C\to Cz$.}
Expanding $Q_{BCz}^2-Q_{BC}^2$ leaves
$i\langle[K_{AB},Q_BQ_z]\rangle$.  Applying \eqref{eq:Hall-Markov-main} to $A|B|z$ gives
\[
i\langle[\Dlt_{A|B|z},Q_BQ_z]\rangle
=-\frac{i}{2}\langle[\Dlt_{A|B|z},-2Q_BQ_z]\rangle.
\]

\paragraph{$A\to Ax'$.}
The defect relation gives
\[
K_{ABx'}=K_{ABCx'}-K_{ABC}+K_{AB}+\Dlt_{x'|AB|C}.
\]
The nested-support commutator expectations involving $K_{ABCx'}$ and $K_{ABC}$ vanish, leaving
\[
\frac{i}{2}\langle[\Dlt_{x'|AB|C},Q_{BC}^2]\rangle.
\]

\paragraph{$B\to Bx'$.}
Compare assigning $x'$ to $B$ with assigning it to $A$.  Exact total-charge identities convert their difference to minus the Hall conductance estimator for $x'|AB|C$.  Combining it with the preceding move gives
\[
-\frac{i}{2}\langle[\Dlt_{x'|AB|C},2Q_{AB}Q_C-Q_{BC}^2]\rangle.
\]

\paragraph{$C\to Cz'$.}
Expanding the charge gives $i\langle[K_{AB},Q_BQ_{z'}]\rangle$.  By
\cref{fact:basic}, $[K_{AB},Q_{AB}]=0$, so this equals
$-i\langle[K_{AB},Q_AQ_{z'}]\rangle$.  With this choice, substituting
the defect for $z'|A|B$ makes all other terms vanish in expectation by
modular stationarity and charge symmetry, leaving
\[
-i\langle[\Dlt_{z'|A|B},Q_AQ_{z'}]\rangle.
\]

In each case, the weighted commutator inequality \eqref{eq:weighted-comm-main} and \cref{thm:main-cmi} turn the exact identity into the inequality in \eqref{eq:Hall-generic-bound-app}.  Summing over a deformation proves \eqref{eq:Sigma-global-main}.  If $\norm{Q_v}_\infty\le q_0$, then $\norm{Q_X}_\infty\le q_0|X|$, so all $\norm{R_m}_\infty$ are polynomial in the region sizes.  Approximate \Aone{} and polynomially many moves then prove \eqref{eq:Sigma-exp-main}.

\section{Singular-value bounds for modular continuity}
\label{app:phase-stability}

\subsection{Auxiliary singular-value estimates}
\label{app:singular-value-estimates}

For $L\ge2$, define
\begin{align}
  k_L(p)&:=\min\{-\log p,L\},
  &&0\le p\le1,
  \label{eq:truncated-log-active-app}\\
  f_L(s)&:=s\min\{-2\log s,L\},
  &&0\le s\le1,
  \qquad f_L(0):=0,
  \label{eq:singular-function-active-app}\\
  F_L(T)&:=k_L(TT^\dagger)T.
  \label{eq:matrix-function-active-app}
\end{align}

\begin{lemma}[Lipschitz truncated singular-value function]
\label{lem:fL-Lipschitz-app}
Let $\widetilde f_L$ be the odd extension of $f_L$ to $[-1,1]$.  Then
\begin{equation}
  \abs{\widetilde f_L(x)-\widetilde f_L(y)}
  \le L\abs{x-y},
  \qquad x,y\in[-1,1].
  \label{eq:fL-Lipschitz-active-app}
\end{equation}
\end{lemma}

\begin{proof}
Set $s_L:=\e^{-L/2}$.  On $[0,1]$,
\begin{equation}
  f_L(s)=
  \begin{cases}
    Ls, & 0\le s\le s_L,\\
    -2s\log s, & s_L\le s\le1.
  \end{cases}
  \label{eq:fL-pieces-active-app}
\end{equation}
The two branches agree at $s_L$, and on their interiors
\begin{equation}
  f_L'(s)=
  \begin{cases}
    L, & 0<s<s_L,\\
    -2\log s-2, & s_L<s<1.
  \end{cases}
  \label{eq:fL-derivative-active-app}
\end{equation}
The second branch ranges from $L-2$ to $-2$, so
$\abs{f_L'(s)}\le L$ because $L\ge2$.  Hence $f_L$ is $L$-Lipschitz on
$[0,1]$.  The odd extension
\begin{equation}
  \widetilde f_L(s)
  :=\operatorname{sgn}(s)f_L(\abs{s}),
  \qquad
  \widetilde f_L(0):=0,
\end{equation}
is continuous and piecewise differentiable on $[-1,1]$, with derivative of
magnitude at most $L$ on every smooth piece.  Integrating this bound across
the junctions at $0$ and $\pm s_L$ proves
\eqref{eq:fL-Lipschitz-active-app}.
\end{proof}

\begin{lemma}[Hilbert--Schmidt continuity of singular-value transforms]\cite{ArakiYamagami1981,AnderssonCarlssonPerfekt2016}.
\label{lem:truncated-modular-vector-app}
Let $f:[0,1]\to\mathbb R$ satisfy $f(0)=0$, and suppose its odd extension
$\widetilde f$ to $[-1,1]$ is $L_f$-Lipschitz.  For a singular-value
decomposition $R=\sum_i r_i\ket{a_i}\bra{b_i}$, define
\begin{equation}
  f^\diamond(R):=\sum_i f(r_i)\ket{a_i}\bra{b_i}.
\end{equation}
Then, for arbitrary rectangular matrices $T,S$ with operator norms at most
one,
\begin{equation}
  \norm{f^\diamond(T)-f^\diamond(S)}_2
  \le L_f\norm{T-S}_2.
  \label{eq:general-singular-transform-bound-app}
\end{equation}
In particular,
\begin{equation}
  \norm{F_L(T)-F_L(S)}_2
  \le L\norm{T-S}_2.
  \label{eq:FL-Lipschitz-app}
\end{equation}
For $f:[0,1]\to\mathbb C$ satisfying the same assumptions, the bound becomes
\begin{equation}
  \norm{f^\diamond(T)-f^\diamond(S)}_2
  \le\sqrt2\,L_f\norm{T-S}_2.
  \label{eq:complex-singular-transform-bound-app}
\end{equation}
\end{lemma}

\begin{proof}
Zero-padding allows us to regard $T$ and $S$ as square matrices without
changing their Hilbert--Schmidt norms or singular-value transforms.  For a
square matrix $R$, introduce
\begin{equation}
  \mathfrak d(R):=
  \begin{pmatrix}
    0&R\\
    R^\dagger&0
  \end{pmatrix}.
\end{equation}
Write $R=\sum_i r_i\ket{a_i}\bra{b_i}$.  For every positive singular value
$r_i$, the vectors
\begin{equation}
  \ket{w_i^\pm}
  :=\frac{1}{\sqrt2}
  \begin{pmatrix}
    \ket{a_i}\\
    \pm\ket{b_i}
  \end{pmatrix}
\end{equation}
are eigenvectors of $\mathfrak d(R)$ with eigenvalues $\pm r_i$.
The kernel contributes nothing because $\widetilde f(0)=0$.  Therefore,
using $\widetilde f(-r_i)=-\widetilde f(r_i)=-f(r_i)$, spectral calculus
gives
\begin{align}
  \widetilde f\bigl(\mathfrak d(R)\bigr)
  &=\sum_i f(r_i)
  \left(
    \ket{w_i^+}\!\bra{w_i^+}
    -\ket{w_i^-}\!\bra{w_i^-}
  \right)
  \notag\\
  &=
  \begin{pmatrix}
    0&f^\diamond(R)\\
    f^\diamond(R^\dagger)&0
  \end{pmatrix}.
  \label{eq:dilation-singular-transform-app}
\end{align}

For Hermitian contractions $A,B$, write their spectral resolutions as
$A=\sum_i\alpha_iP_i$ and $B=\sum_j\beta_jQ_j$.  Define
\begin{equation}
  \theta_{ij}:=
  \begin{cases}
    \displaystyle
    \frac{\widetilde f(\alpha_i)-\widetilde f(\beta_j)}
         {\alpha_i-\beta_j},
      &\alpha_i\ne\beta_j,\\[0.8em]
    0,&\alpha_i=\beta_j.
  \end{cases}
\end{equation}
Then $\abs{\theta_{ij}}\le L_f$ and
\begin{equation}
  P_i\bigl[\widetilde f(A)-\widetilde f(B)\bigr]Q_j
  =\theta_{ij}P_i(A-B)Q_j.
\end{equation}
The blocks $P_iXQ_j$ are mutually orthogonal in Hilbert--Schmidt inner
product as $(i,j)$ varies.  Consequently,
\begin{align}
  \norm{\widetilde f(A)-\widetilde f(B)}_2^2
  &=\sum_{i,j}\abs{\theta_{ij}}^2
    \norm{P_i(A-B)Q_j}_2^2
  \notag\\
  &\le L_f^2\sum_{i,j}\norm{P_i(A-B)Q_j}_2^2
  =L_f^2\norm{A-B}_2^2.
\end{align}
Thus
\begin{equation}
  \norm{\widetilde f(A)-\widetilde f(B)}_2
  \le L_f\norm{A-B}_2.
  \label{eq:HS-functional-calculus-app}
\end{equation}

Apply \eqref{eq:HS-functional-calculus-app} to $\mathfrak d(T)$ and
$\mathfrak d(S)$.  Since $f$ is real-valued,
$f^\diamond(R^\dagger)=f^\diamond(R)^\dagger$, so
\begin{equation}
  \norm{
    \widetilde f(\mathfrak d(T))
    -\widetilde f(\mathfrak d(S))
  }_2^2
  =2\norm{f^\diamond(T)-f^\diamond(S)}_2^2.
\end{equation}
On the other hand,
\begin{equation}
  \norm{\mathfrak d(T)-\mathfrak d(S)}_2^2
  =2\norm{T-S}_2^2
\end{equation}
The factors of $2$ cancel, proving
\eqref{eq:general-singular-transform-bound-app}. For $f=f_L$, the singular-value decomposition of $T$ shows that
\begin{align}
  F_L(T)
  &=k_L(TT^\dagger)T
  \notag\\
  &=\sum_i s_i k_L(s_i^2)\ket{u_i}\bra{v_i}
  =f_L^\diamond(T),
\end{align}
and similarly for $S$.  By
\cref{lem:fL-Lipschitz-app}, the odd extension of $f_L$ is
$L$-Lipschitz, so \eqref{eq:general-singular-transform-bound-app} proves
\eqref{eq:FL-Lipschitz-app}.

For complex-valued $f$, the dilation identity
\eqref{eq:dilation-singular-transform-app} and the spectral-block estimate
\eqref{eq:HS-functional-calculus-app} remain valid.  The upper-right block
of the difference of the dilated transforms is
$f^\diamond(T)-f^\diamond(S)$, so
\begin{align*}
  \norm{f^\diamond(T)-f^\diamond(S)}_2
  &\le\norm{\widetilde f(\mathfrak d(T))
             -\widetilde f(\mathfrak d(S))}_2\\
  &\le L_f\norm{\mathfrak d(T)-\mathfrak d(S)}_2
  =\sqrt2\,L_f\norm{T-S}_2.
\end{align*}
Explicitly, for $R=U\Sigma V^\dagger$,
\[
  f^\diamond(R^\dagger)=Vf(\Sigma)U^\dagger,
  \qquad
  f^\diamond(R)^\dagger=V\overline{f(\Sigma)}U^\dagger.
\]
Writing $A:=f^\diamond(T)-f^\diamond(S)$ and
$B:=f^\diamond(T^\dagger)-f^\diamond(S^\dagger)$, we have
\begin{align*}
  \widetilde f(\mathfrak d(T))-\widetilde f(\mathfrak d(S))
  &=\begin{pmatrix}0&A\\B&0\end{pmatrix},\\
  \norm{A}_2^2+\norm{B}_2^2
  &\le L_f^2\norm{\mathfrak d(T)-\mathfrak d(S)}_2^2
  =2L_f^2\norm{T-S}_2^2.
\end{align*}
For real-valued $f$, $B=A^\dagger$, so
\[
  2\norm{A}_2^2\le2L_f^2\norm{T-S}_2^2.
\]
For complex-valued $f$, we instead use
\[
  \norm{A}_2^2\le\norm{A}_2^2+\norm{B}_2^2
  \le2L_f^2\norm{T-S}_2^2,
  \qquad
  \norm{A}_2\le\sqrt2\,L_f\norm{T-S}_2.
\]
This proves \eqref{eq:complex-singular-transform-bound-app}.
\end{proof}

\section{Poisson analysis on a strip}
\label{app:strip-Poisson}

This appendix records the technique used in the proof of
\cref{thm:weighted-IMF}. 

\begin{lemma}[Hadamard three-lines theorem]
\label{lem:three-lines-app}
Let $g$ be bounded and continuous on the closed strip
$0\le\operatorname{Re}z\le a$ and analytic on $\mathcal S_a$.  Define
\begin{equation}
  M_0:=\sup_{s\in\mathbb R}\abs{g(is)},
  \qquad
  M_a:=\sup_{s\in\mathbb R}\abs{g(a+is)}.
\end{equation}
Then, for $z=x+iy$ with $0\le x\le a$,
\begin{equation}
  \abs{g(x+iy)}
  \le M_0^{1-x/a}M_a^{x/a}.
  \label{eq:three-lines-app}
\end{equation}
\end{lemma}

\begin{proof}
See Ref.~\cite[p.~39]{grafakos2014classical}.
\end{proof}

\begin{lemma}[Contraction of the interpolation function]
\label{lem:interpolation-contraction-app}
Let $\rho=\rho_{XYZ}$ be a faithful finite-dimensional state and define
\begin{equation}
  \mathcal W(z)
  :=\rho_{YZ}^{z}\rho_Y^{-z}\rho_{XY}^{z}\rho^{1/2-z},
  \qquad
  0\le\operatorname{Re}z\le\frac12.
  \label{eq:interpolation-function-app}
\end{equation}
Then
\begin{equation}
  \norm{\mathcal W(z)}_2\le1
  \qquad
  (0\le\operatorname{Re}z\le\tfrac12).
  \label{eq:interpolation-contraction-app}
\end{equation}
\end{lemma}

\begin{proof}
On the lower boundary, every factor preceding $\rho^{1/2}$ is unitary, so
\begin{equation}
  \norm{\mathcal W(is)}_2=1.
  \label{eq:interpolation-lower-app}
\end{equation}
On the upper boundary,
\begin{align}
  \norm{\mathcal W(\tfrac12+is)}_2^2
  &=
  \Tr\!\left[
    \rho_{XY}\rho_Y^{-1/2+is}
    \rho_{YZ}\rho_Y^{-1/2-is}
  \right]
  \notag\\
  &=
  \Tr\!\left[
    \rho_Y\rho_Y^{-1/2+is}
    \rho_{YZ}\rho_Y^{-1/2-is}
  \right]
  =1.
  \label{eq:interpolation-upper-app}
\end{align}
The second equality traces out $X$, after which cyclicity evaluates the
trace.  Now fix an operator $T$ with $\norm{T}_2=1$ and define
\begin{equation}
  g_T(z):=\Tr[T^\dagger\mathcal W(z)].
\end{equation}
The Hilbert--Schmidt Cauchy--Schwarz inequality and
\eqref{eq:interpolation-lower-app}--\eqref{eq:interpolation-upper-app} give
\begin{equation}
  \sup_{s\in\mathbb R}\abs{g_T(is)}\le1,
  \qquad
  \sup_{s\in\mathbb R}\abs{g_T(\tfrac12+is)}\le1.
\end{equation}
Applying \cref{lem:three-lines-app} with $a=1/2$ yields
\begin{equation}
  \abs{g_T(x+iy)}
  \le
  \left(\sup_s\abs{g_T(is)}\right)^{1-2x}
  \left(\sup_s\abs{g_T(\tfrac12+is)}\right)^{2x}
  \le1.
\end{equation}
Taking the supremum over all $T$ with $\norm{T}_2=1$ and using
Hilbert--Schmidt duality proves \eqref{eq:interpolation-contraction-app}.
This is the interpolation argument used in multivariate trace
inequalities \cite{SutterBertaTomamichel2017}.
\end{proof}

We next state the Poisson formula f.  
\begin{lemma}[Poisson integral formula on a strip]
\label{lem:strip-Poisson-representation-app}
Let $h$ be a bounded real-valued harmonic function on $\{z=x+iy:0<x<a\}$ with continuous boundary values. For $0<x<a$,
\begin{align}
h(x+iy)
={}&
\int_{-\infty}^{\infty}
P_0^{(a)}(x,y-s)h(is)\,ds +
\int_{-\infty}^{\infty}
P_a^{(a)}(x,y-s)h(a+is)\,ds,
\label{eq:strip-Poisson-representation-app}
\end{align}
where
\begin{align}
P_0^{(a)}(x,s)
&:=
\frac{1}{2a}
\frac{\sin(\pi x/a)}
{\cosh(\pi s/a)-\cos(\pi x/a)},
\label{eq:strip-P0-app}\\
P_a^{(a)}(x,s)
&:=
\frac{1}{2a}
\frac{\sin(\pi x/a)}
{\cosh(\pi s/a)+\cos(\pi x/a)}.
\label{eq:strip-Pa-app}
\end{align}
\end{lemma}
\begin{proof} 
See Refs.~\cite{widder1961functions} 
\end{proof}

\begin{lemma}[Inward normal derivative at the strip boundary]
\label{lem:strip-Poisson-normal-app}
Suppose in addition that $h$ is twice continuously differentiable near
$z=0$, that $h\ge0$ on the closed strip, and that $h(0)=0$.  Then
\begin{align}
\partial_xh(0)
={}&
\frac{\pi}{2a^2}
\int_{-\infty}^{\infty}
\frac{h(is)}{\cosh(\pi s/a)-1}\,ds +
\frac{\pi}{2a^2}
\int_{-\infty}^{\infty}
\frac{h(a+is)}{\cosh(\pi s/a)+1}\,ds.
\label{eq:strip-normal-general-app}
\end{align}
\end{lemma}

\begin{proof}
The only issue is justifying differentiation under the Poisson integrals.
Set $q(s):=h(is)$.  Since $q(s)\ge0=q(0)$, one has $q'(0)=0$, and Taylor's
theorem gives $h(is)=O(s^2)$ as $s\to0$.

For $s\neq0$, direct differentiation of
\eqref{eq:strip-P0-app}--\eqref{eq:strip-Pa-app} at $x=0$ gives
\begin{align}
\partial_xP_0^{(a)}(0,s)
&=
\frac{\pi}{2a^2}
\frac{1}{\cosh(\pi s/a)-1},
\label{eq:strip-P0-normal-app}\\
\partial_xP_a^{(a)}(0,s)
&=
\frac{\pi}{2a^2}
\frac{1}{\cosh(\pi s/a)+1}.
\label{eq:strip-Pa-normal-app}
\end{align}
Near $(x,s)=(0,0)$, Taylor expansion gives
$|\partial_xP_0^{(a)}(x,s)|\le C/(s^2+x^2)$.  Thus its product with
$h(is)=O(s^2)$ is locally bounded, cancelling the otherwise nonintegrable
$1/s^2$ singularity.  The upper-boundary term is nonsingular at $s=0$, and
both differentiated kernels decay exponentially at infinity.  Dominated
convergence theorem can therefore be applied to move the $x\downarrow0$ limit inside the integrals, thus permitting differentiation of
\eqref{eq:strip-Poisson-representation-app} at $x=0$, $y=0$; inserting
\eqref{eq:strip-P0-normal-app}--\eqref{eq:strip-Pa-normal-app} proves
\eqref{eq:strip-normal-general-app}.
\end{proof}

\section{Nonfaithful states in finite-time instantaneous modular flow}
\label{app:nonfaithful-IMF}

In the proof of the weighted finite-time Markov-decomposition theorem,
\cref{thm:weighted-IMF}, it is convenient to first assume that the reduced
state on the tripartite region is faithful.  This assumption is only a
technical device for the complex strip-interpolation argument.  It is
\emph{not} required for the definition of instantaneous modular flow itself,
nor does it modify any of the bounds stated in the main text.  In this
appendix we explain this point carefully and prove that
\cref{eq:weighted-IMF}, together with all of its consequences, remains valid
for arbitrary finite-dimensional states without any lower bound on their
nonzero eigenvalues.

\subsection{Support preservation under real-time instantaneous modular flow}
\label{app:nonfaithful-support}

Let $\ket{\psi}$ be a normalized pure state on
$\cH_X\otimes\cH_{\bar X}$, and let
\begin{equation}
\rho_X
=
\sum_{a=1}^{r}p_a\ket a\!\bra a,
\qquad
p_a>0,
\label{eq:nonfaithful-rhoX}
\end{equation}
be its reduced density matrix on $X$.  We denote by
\begin{equation}
\Pi_X:=\supp\rho_X=\sum_{a=1}^{r}\ket a\!\bra a
\label{eq:nonfaithful-support-projector}
\end{equation}
the orthogonal projector onto its support.  Imaginary powers are understood
on this support:
\begin{equation}
\rho_X^{it}
:=
\sum_{a=1}^{r}p_a^{it}\ket a\!\bra a.
\label{eq:nonfaithful-imaginary-power}
\end{equation}
Thus $\rho_X^{it}$ is generally only a partial unitary on $\cH_X$,
\begin{equation}
\rho_X^{it}\rho_X^{-it}
=
\rho_X^{-it}\rho_X^{it}
=
\Pi_X.
\label{eq:nonfaithful-partial-unitary}
\end{equation}

\begin{lemma}[Support preservation and unitary extension of IMF]
\label{lem:nonfaithful-support-preservation}
For every $t\in\mathbb R$, the instantaneous modular flow
\begin{equation}
\cI_X(t)\ket\psi
:=
\rho_X^{it}\ket\psi
\label{eq:nonfaithful-IMF-definition}
\end{equation}
is a normalized pure state.  Moreover,
\begin{align}
\rho_X\bigl(\cI_X(t)\psi\bigr)
&=\rho_X(\psi),
\label{eq:nonfaithful-X-marginal-invariant}\\
\rho_{\bar X}\bigl(\cI_X(t)\psi\bigr)
&=\rho_{\bar X}(\psi),
\label{eq:nonfaithful-complement-marginal-invariant}
\end{align}
and hence the support and rank of the two Schmidt marginals are unchanged.
In particular,
\begin{equation}
\cI_X(-t)\cI_X(t)\ket\psi=\ket\psi.
\label{eq:nonfaithful-IMF-inverse}
\end{equation}

Furthermore, the action of $\rho_X^{it}$ on the current state can always be
implemented by a genuine unitary on the full Hilbert space $\cH_X$.  One
choice is
\begin{equation}
\widetilde U_X(t)
:=
\rho_X^{it}+(\id_X-\Pi_X)
=
\sum_{a=1}^{r}p_a^{it}\ket a\!\bra a
+
(\id_X-\Pi_X),
\label{eq:nonfaithful-unitary-extension}
\end{equation}
for which
\begin{equation}
\widetilde U_X(t)\ket\psi
=
\rho_X^{it}\ket\psi
=
\cI_X(t)\ket\psi.
\label{eq:nonfaithful-unitary-extension-action}
\end{equation}
\end{lemma}

\begin{proof}
Take the Schmidt decomposition across $X:\bar X$,
\begin{equation}
\ket\psi
=
\sum_{a=1}^{r}
\sqrt{p_a}\,
\ket a_X\ket{\widehat a}_{\bar X},
\label{eq:nonfaithful-Schmidt}
\end{equation}
where only the nonzero Schmidt values are included.  Then
\begin{equation}
\cI_X(t)\ket\psi
=
\sum_{a=1}^{r}
\sqrt{p_a}\,
e^{it\log p_a}
\ket a_X\ket{\widehat a}_{\bar X}.
\label{eq:nonfaithful-Schmidt-flow}
\end{equation}
The flow only changes Schmidt phases.  Consequently its norm is one, and
tracing either side of the Schmidt decomposition gives
\eqref{eq:nonfaithful-X-marginal-invariant} and
\eqref{eq:nonfaithful-complement-marginal-invariant}.  Applying the opposite
flow multiplies each Schmidt term by
$e^{-it\log p_a}e^{it\log p_a}=1$, proving
\eqref{eq:nonfaithful-IMF-inverse}.

Finally, $\widetilde U_X(t)$ is unitary because it acts by phases on
$\Pi_X\cH_X$ and by the identity on its orthogonal complement.  Since
\begin{equation}
\Pi_X\ket\psi=\ket\psi,
\end{equation}
its action on $\ket\psi$ agrees with the support imaginary power, proving
\eqref{eq:nonfaithful-unitary-extension-action}.
\end{proof}

More generally, if $X\subseteq Y$ and $\rho_Y$ has marginal $\rho_X$, then
\begin{equation}
\Pi_X\rho_Y=\rho_Y=\rho_Y\Pi_X.
\label{eq:nonfaithful-projector-absorbed}
\end{equation}
Indeed, purify $\rho_Y$ and apply the same Schmidt-support argument to the
purification.  Consequently,
\begin{align}
\rho_X^{it}\rho_Y
&=
\widetilde U_X(t)\rho_Y,
\label{eq:nonfaithful-extension-left}\\
\rho_Y\rho_X^{it}
&=
\rho_Y\widetilde U_X(t).
\label{eq:nonfaithful-extension-right}
\end{align}
This is the support-projector mechanism underlying the flipping, inverse, and
nested commutation moves in the presence of zero eigenvalues.

\subsection{Why faithfulness appears in the strip proof}
\label{app:nonfaithful-strip-role}

The preceding discussion shows that real-time IMF has no singularity at zero
Schmidt values.  Faithfulness enters the proof of
\cref{thm:weighted-IMF} for a different reason: the proof analytically
continues the real-time vector families into complex modular time.

For a faithful state $\rho=\rho_{XYZ}$, the proof introduces
\begin{equation}
\mathcal W(z)
=
\rho_{YZ}^{z}
\rho_Y^{-z}
\rho_{XY}^{z}
\rho^{1/2-z},
\qquad
0\le\operatorname{Re}z\le\frac12.
\label{eq:nonfaithful-interpolation-function}
\end{equation}
Faithfulness makes all positive and negative complex powers in
\eqref{eq:nonfaithful-interpolation-function} bounded operators on the full finite-dimensional
Hilbert space.  In particular, on the imaginary axis the relative factor
\begin{equation}
U_s:=\rho_{YZ}^{-is}\rho_Y^{is}
\label{eq:nonfaithful-Us}
\end{equation}
is a genuine unitary.  This gives the exact norm identity
\begin{equation}
d_\phi(s)
=
\norm{
\mathcal W(-is)-\rho^{1/2}
}_2
\label{eq:nonfaithful-d-interpolation-faithful}
\end{equation}
used in \eqref{eq:d-interpolation-function}, and it also gives the unit-norm boundary
condition
\begin{equation}
\norm{\mathcal W(is)}_2=1.
\label{eq:nonfaithful-interpolation-lower-faithful}
\end{equation}

For a nonfaithful state, support imaginary powers are perfectly well defined,
but they are only partial unitaries.  If one formally keeps the same
expression \eqref{eq:nonfaithful-Us}, then in general
\begin{equation}
U_s^\dagger U_s\neq\id
\end{equation}
on the full Hilbert space.  Therefore unitary invariance of the norm cannot
be used directly to infer \eqref{eq:nonfaithful-d-interpolation-faithful} for the
difference of the two matrix amplitudes.  Moreover, the upper strip boundary
contains negative real powers such as $\rho_Y^{-1/2+is}$, whose treatment
requires careful support bookkeeping when different noncommuting marginals
have different kernels.  Thus the faithfulness assumption is a convenience
for the \emph{complex interpolation}; it is not an assumption required by
the physical IMF evolution.

Rather than defining every inverse, logarithm, and modular power only on the
support of the corresponding density operator, and tracking the resulting
support projectors throughout the strip proof, we remove this auxiliary
assumption by a faithful approximation of the initial state.

\subsection{Faithful regularization and convergence of IMF trajectories}
\label{app:nonfaithful-regularization}

Let
\begin{equation}
\rho^{(\eta)}
:=
(1-\eta)\rho
+
\eta\frac{\id_{XYZ}}{d_{XYZ}},
\qquad
0<\eta<1.
\label{eq:nonfaithful-depolarization}
\end{equation}
Every marginal is then faithful; for example,
\begin{equation}
\rho_R^{(\eta)}
=
(1-\eta)\rho_R
+
\eta\frac{\id_R}{d_R}
>
0.
\label{eq:nonfaithful-depolarized-marginal}
\end{equation}
Set $T_\eta:=(\rho^{(\eta)})^{1/2}$ and $T:=\rho^{1/2}$.
Since $\rho^{(\eta)}$ and $\rho$ commute, the scalar inequality
$(\sqrt a-\sqrt b)^2\le|a-b|$ gives
\begin{equation}
\norm{T_\eta-T}_2^2
\le\norm{\rho^{(\eta)}-\rho}_1
=\eta\norm{\id_{XYZ}/d_{XYZ}-\rho}_1
\le2\eta.
\label{eq:nonfaithful-amplitude-convergence}
\end{equation}
Thus the matrix amplitudes converge in Hilbert--Schmidt norm, even though
imaginary powers need not converge in operator norm on the kernel of $\rho$.

For a matrix amplitude $T$ with $\norm{T}_2=1$, write the current-state IMF
in matrix form as
\begin{equation}
\cI_R(t)T
:=\bigl[\Tr_{XYZ\setminus R}(TT^\dagger)\bigr]^{it}T,
\label{eq:nonfaithful-matrix-IMF}
\end{equation}
with support imaginary powers as before.  Regrouping the row indices in
$R$ against the remaining row and column indices preserves the
Hilbert--Schmidt norm and turns this map into the singular-value transform
used in the proof of \cref{lem:IMF-Lipschitz}.  Hence
\begin{equation}
\norm{\cI_R(t)T-\cI_R(t)S}_2
\le\ell(t)\norm{T-S}_2.
\label{eq:nonfaithful-matrix-IMF-Lipschitz}
\end{equation}
Each map uses the marginal of its current matrix amplitude, not that of
the initial state.

Define the matrix factorization error
\begin{equation}
\mathfrak E_\rho^{X:Y:Z}(t)
:=\cI_{XYZ}(t)\rho^{1/2}
-\cI_{XY}(t)\cI_{YZ}(t)\cI_Y(-t)\rho^{1/2}.
\label{eq:nonfaithful-matrix-error}
\end{equation}
For any original state vector $\ket\phi$ with reduced state $\rho$ on
$XYZ$, the vector and matrix sequences have identical current reduced
states and therefore identical left-acting operators.  Consequently,
\begin{equation}
\norm{\mathfrak E_\phi^{X:Y:Z}(t)}_2
=\norm{\mathfrak E_\rho^{X:Y:Z}(t)}_2.
\label{eq:nonfaithful-matrix-error-norm}
\end{equation}

\begin{lemma}[Convergence of the finite Markov-factorization error]
\label{lem:nonfaithful-E-convergence}
For every fixed $t\in\mathbb R$,
\begin{equation}
\begin{aligned}
\norm{\mathfrak E_{\rho^{(\eta)}}^{X:Y:Z}(t)
      -\mathfrak E_\rho^{X:Y:Z}(t)}_2
&\le[\ell(t)+\ell(t)^3]\norm{T_\eta-T}_2\\
&\le\sqrt{2\eta}\,[\ell(t)+\ell(t)^3].
\end{aligned}
\label{eq:nonfaithful-E-convergence}
\end{equation}
In particular, the matrix errors converge in Hilbert--Schmidt norm as
$\eta\downarrow0$.
\end{lemma}

\begin{proof}
By \eqref{eq:nonfaithful-matrix-IMF-Lipschitz}, the first branch is
$\ell(t)$-Lipschitz and the second, a composition of three current-state
IMF maps, is $\ell(t)^3$-Lipschitz.  The triangle inequality and
\eqref{eq:nonfaithful-amplitude-convergence} give the result.
\end{proof}

\subsection{Removal of the faithfulness assumption in the weighted theorem}
\label{app:nonfaithful-Fatou}

We can now pass to the nonfaithful limit without changing the bound.

\begin{proposition}[The weighted IMF bound for arbitrary finite-dimensional states]
\label{prop:nonfaithful-weighted-IMF}
The weighted finite-time Markov-decomposition bound
\begin{equation}
\pi
\int_{-\infty}^{\infty}
\frac{
\norm{\mathfrak E_{\phi}^{X:Y:Z}(t)}_2^2
}{
\ell(t)^2[\cosh(2\pi t)-1]
}
\,dt
\le
I(X:Z|Y)_\rho
\label{eq:nonfaithful-weighted-IMF}
\end{equation}
holds for every finite-dimensional tripartite state $\rho_{XYZ}$, whether
or not $\rho$ or its marginals are faithful.
\end{proposition}

\begin{proof}
For every $\eta>0$, the state $\rho^{(\eta)}$ is faithful, so the faithful
case of \cref{thm:weighted-IMF}, expressed in matrix form using
\eqref{eq:nonfaithful-matrix-error-norm}, gives
\begin{equation}
\pi
\int_{-\infty}^{\infty}
\frac{
\norm{
\mathfrak E_{\rho^{(\eta)}}^{X:Y:Z}(t)
}_2^2
}{
\ell(t)^2[\cosh(2\pi t)-1]
}
\,dt
\le
I(X:Z|Y)_{\rho^{(\eta)}}.
\label{eq:nonfaithful-weighted-eta}
\end{equation}
For $t\neq0$, define
\begin{equation}
g_\eta(t)
:=
\frac{
\pi
\norm{
\mathfrak E_{\rho^{(\eta)}}^{X:Y:Z}(t)
}_2^2
}{
\ell(t)^2[\cosh(2\pi t)-1]
},
\qquad
g(t)
:=
\frac{
\pi
\norm{
\mathfrak E_{\rho}^{X:Y:Z}(t)
}_2^2
}{
\ell(t)^2[\cosh(2\pi t)-1]
}.
\label{eq:nonfaithful-g-eta}
\end{equation}
The value assigned at the single point $t=0$ is irrelevant for the
integral.  By
\cref{lem:nonfaithful-E-convergence},
\begin{equation}
g_\eta(t)\longrightarrow g(t)
\qquad
(t\neq0).
\label{eq:nonfaithful-g-pointwise}
\end{equation}
All of these functions are nonnegative.  Fatou's lemma therefore gives
\begin{align}
\int_{-\infty}^{\infty}g(t)\,dt
&\le
\liminf_{\eta\downarrow0}
\int_{-\infty}^{\infty}g_\eta(t)\,dt
\notag\\
&\le
\liminf_{\eta\downarrow0}
I(X:Z|Y)_{\rho^{(\eta)}}.
\label{eq:nonfaithful-Fatou-step}
\end{align}
Conditional mutual information is continuous in fixed finite dimension, and
$\rho^{(\eta)}\to\rho$ in trace norm.  Hence
\begin{equation}
I(X:Z|Y)_{\rho^{(\eta)}}
\longrightarrow
I(X:Z|Y)_\rho.
\label{eq:nonfaithful-CMI-convergence}
\end{equation}
Substitution into \eqref{eq:nonfaithful-Fatou-step} and
\eqref{eq:nonfaithful-matrix-error-norm} prove
\eqref{eq:nonfaithful-weighted-IMF}.
\end{proof}

\begin{remark}[No change to the finite-time or no-go bounds]
\label{rem:nonfaithful-no-bound-change}
The regularization parameter $\eta$ is removed before any finite-time
bound is used.  Therefore no factor depending on $\eta$, on
$\lambda_{\min}(\rho_R)$, or on the rank of any reduced density matrix enters
the final result.  In particular:

\begin{enumerate}
\item the weighted bound \eqref{eq:weighted-IMF} holds with exactly the
same right-hand side $I(X:Z|Y)_\rho$;

\item the pointwise finite-time bound in
\cref{cor:pointwise-IMF} follows from the same weighted bound and the same
spectrum-free modular-moment bound, with unchanged constants;

\item the weighted comparison of the no-go trajectories and all subsequent
continuity and deformation bounds are unchanged;

\item the quantitative lower bounds on the \Aone{} violation forced by a
nonzero modular commutator or Hall conductance estimator therefore retain exactly the
same dependence on time, subsystem dimensions, $c_-$, and $\sigma_{xy}$.
\end{enumerate}

Thus faithfulness is an auxiliary assumption of the analytic interpolation
proof only.  The physical instantaneous modular flow and every theorem stated
in the main text remain valid for arbitrary finite-dimensional reduced
states.
\end{remark}

\section{Pointwise finite-time factorization error of the IMF}
\label{app:pointwise-IMF}

Combining the weighted estimate with a time-Lipschitz bound gives the
following explicit pointwise estimate.

\begin{corollary}[Explicit pointwise finite-time error]
\label{cor:pointwise-IMF}
Define
\begin{equation}
\mathfrak m_{X:Y:Z}
:=
\log(d_{XYZ}d_{XY}d_{YZ}d_Y).
\label{eq:m-definition}
\end{equation}
Then
\begin{equation}
\begin{aligned}
\norm{\mathfrak E_{\phi}^{X:Y:Z}(t)}_2
\le
\min\Bigg\{
2,
|t|\mathfrak m_{X:Y:Z},
\ell(t)
\left[
\frac{4\mathfrak m_{X:Y:Z}}{\pi}
\left(
\cosh\!\left(2\pi\left[|t|+\mathfrak m_{X:Y:Z}^{-1}\right]\right)-1
\right)
I(X:Z|Y)_\phi
\right]^{1/3}
\Bigg\}.
\end{aligned}
\label{eq:pointwise-IMF}
\end{equation}
\end{corollary}

\begin{proof}
First assume that $\rho:=\rho_{XYZ}$ is faithful.  Work directly with the
matrix amplitudes
\begin{equation}
P_t:=\rho_{XY}^{-it}\rho^{1/2+it},
\qquad
Q_t:=\rho_Y^{-it}\rho_{YZ}^{it}\rho^{1/2}.
\label{eq:trajectory-matrix-representations}
\end{equation}
Then $d_\phi(t)=\norm{P_t-Q_t}_2$ by \eqref{eq:dphi-definition},
and $\norm{P_t}_2=\norm{Q_t}_2=1$.  With $K_R=-\log\rho_R$,
differentiation gives
\begin{align*}
\dot P_t
&=iK_{XY}P_t-i\rho_{XY}^{-it}K_{XYZ}\rho^{1/2+it},\\
\dot Q_t
&=iK_YQ_t-i\rho_Y^{-it}K_{YZ}\rho_{YZ}^{it}\rho^{1/2}.
\end{align*}
The identities
$\Tr_Z(P_tP_t^\dagger)=\rho_{XY}$ and
$\Tr_{XZ}(Q_tQ_t^\dagger)=\rho_Y$, together with unitary invariance of
the Hilbert--Schmidt norm, therefore yield
\begin{align}
\norm{\dot P_t}_2
&\le\norm{K_{XY}\rho_{XY}^{1/2}}_2
   +\norm{K_{XYZ}\rho^{1/2}}_2
\le\log d_{XY}+\log d_{XYZ},\notag\\
\norm{\dot Q_t}_2
&\le\norm{K_Y\rho_Y^{1/2}}_2
   +\norm{K_{YZ}\rho_{YZ}^{1/2}}_2
\le\log d_Y+\log d_{YZ},
\label{eq:trajectory-derivative-bounds}
\end{align}
where the final inequalities use
\eqref{eq:universal-modular-moment-main}.  Hence
\begin{equation}
\abs{d_\phi(s)-d_\phi(t)}
\le\norm{P_s-P_t}_2+\norm{Q_s-Q_t}_2
\le\mathfrak m_{X:Y:Z}\abs{s-t}.
\label{eq:dphi-Lipschitz}
\end{equation}
We next convert the weighted integral estimate into a pointwise estimate.
Fix $t\in\mathbb R$.  The number $\mathfrak m_{X:Y:Z}$ is strictly positive
by its definition in \eqref{eq:m-definition}.  If $d_\phi(t)=0$, then
\eqref{eq:E-by-d} already gives
$\mathfrak E_\phi^{X:Y:Z}(t)=0$, so the third branch is immediate.  We may
therefore assume $d_\phi(t)>0$.  By \eqref{eq:dphi-Lipschitz},
\begin{equation}
d_\phi(s)
\ge d_\phi(t)-\mathfrak m_{X:Y:Z}\abs{s-t}
\ge\frac{d_\phi(t)}{2}
\qquad
\text{whenever }
\abs{s-t}\le\frac{d_\phi(t)}{2\mathfrak m_{X:Y:Z}}.
\label{eq:dphi-local-lower-bound}
\end{equation}
Since $\norm{P_t}_2=\norm{Q_t}_2=1$, we have $d_\phi(t)\le2$.  Hence
$|s|\le|t|+\mathfrak m_{X:Y:Z}^{-1}$ throughout this interval.
Because
$x\mapsto\cosh(2\pi x)-1$ is increasing for $x\ge0$, the weighted estimate
\eqref{eq:weighted-d} yields
\begin{align}
I(X:Z|Y)_\phi
&\ge
\pi\int_{t-d_\phi(t)/(2\mathfrak m_{X:Y:Z})}^{t+d_\phi(t)/(2\mathfrak m_{X:Y:Z})}
\frac{d_\phi(s)^2}{\cosh(2\pi s)-1}\,ds
\notag\\
&\ge
\frac{\pi d_\phi(t)^3}
{4\mathfrak m_{X:Y:Z}
\left[\cosh\!\left(2\pi[|t|+\mathfrak m_{X:Y:Z}^{-1}]\right)-1\right]}.
\label{eq:weighted-local-pointwise}
\end{align}
Equivalently,
\begin{equation}
d_\phi(t)^3
\le
\frac{4\mathfrak m_{X:Y:Z}}{\pi}
\left(
\cosh\!\left(2\pi[|t|+\mathfrak m_{X:Y:Z}^{-1}]\right)-1
\right)I(X:Z|Y)_\phi.
\label{eq:dphi-pointwise-cubic}
\end{equation}
Taking the cube root and using
\eqref{eq:E-by-d} proves the third branch of
\eqref{eq:pointwise-IMF}.

For completeness, we derive the other two branches.  Every IMF map preserves
the norm of its input state.  The two vectors whose difference defines
$\mathfrak E_\phi^{X:Y:Z}(t)$ are therefore normalized, and the triangle
inequality immediately gives
\begin{equation}
\norm{\mathfrak E_\phi^{X:Y:Z}(t)}_2\le2.
\label{eq:pointwise-trivial-branch}
\end{equation}
For any normalized $\ket\omega$ and any region $R$, the marginal used during
the single map $\cI_R(\tau)\ket\omega$ is fixed by its input.  Hence
\begin{align}
\norm{\cI_R(\tau)\ket\omega-\ket\omega}_2
&=
\norm{
\int_0^\tau
(-iK_R(\omega))e^{-isK_R(\omega)}\ket\omega\,ds
}_2
\notag\\
&\le
|\tau|\norm{K_R(\omega)\ket\omega}_2
\le |\tau|\log d_R.
\label{eq:single-IMF-identity-bound}
\end{align}
The last inequality follows from
\eqref{eq:universal-modular-moment-main} applied to the current marginal
$\rho_R(\omega)$.  Applying \eqref{eq:single-IMF-identity-bound} at each
current input state and telescoping directly gives
\begin{align}
&\norm{
\cI_{XY}(t)\cI_{YZ}(t)\cI_Y(-t)\ket\phi-\ket\phi
}_2
\notag\\
&\le
\norm{
\cI_{XY}(t)\cI_{YZ}(t)\cI_Y(-t)\ket\phi
-\cI_{YZ}(t)\cI_Y(-t)\ket\phi
}_2
\notag\\
&\quad+
\norm{
\cI_{YZ}(t)\cI_Y(-t)\ket\phi
-\cI_Y(-t)\ket\phi
}_2
\notag\\
&\quad+
\norm{\cI_Y(-t)\ket\phi-\ket\phi}_2
\notag\\
&\le
|t|(\log d_{XY}+\log d_{YZ}+\log d_Y),
\label{eq:three-IMF-identity-bound}
\end{align}
while the one-map branch obeys
\begin{equation}
\norm{\cI_{XYZ}(t)\ket\phi-\ket\phi}_2
\le |t|\log d_{XYZ}.
\label{eq:full-IMF-identity-bound}
\end{equation}
A final triangle inequality between the two branches yields the second branch.

The calculation above is valid directly on the supports of nonfaithful
marginals.  Equivalently, one may apply it to the faithful regularizations in
\cref{app:nonfaithful-IMF} and pass to the limit; all constants are uniform in
the regularization parameter.
\end{proof}

\section{Lambert-\texorpdfstring{$W$}{W} inversion and scaling}
\label{app:Lambert-W}

The Lambert function is defined implicitly by
\begin{equation}
W(x)\e^{W(x)}=x.
\end{equation}
On $[-\e^{-1},0)$ it has two real branches.  We use the negative branch
$W_{-1}$, characterized by
\begin{equation}
W_{-1}(x)\le-1,
\qquad
W_{-1}(x)\e^{W_{-1}(x)}=x.
\label{eq:Lambert-W-minus-one-definition}
\end{equation}

\begin{lemma}[Lambert-$W$ representation of $\chi_\gamma$]
\label{lem:chi-Lambert-representation}
For $\gamma\ge2$ and $0<z\le\gamma$, the unique solution
$x=\chi_\gamma(z)\in(0,1]$ of $x[\gamma+2\log(1/x)]=z$ is
\begin{equation}
\chi_\gamma(z)
=-\frac{z}{2W_{-1}\!\left(-\frac z2\e^{-\gamma/2}\right)}.
\label{eq:chi-Lambert-representation}
\end{equation}
\end{lemma}

\begin{proof}
Set $v:=\gamma/2+\log(1/x)$.  Then $v\ge\gamma/2\ge1$,
$x=\e^{\gamma/2-v}$, and the defining equation becomes
\begin{equation}
z=2xv=2v\e^{\gamma/2-v},
\qquad
(-v)\e^{-v}=-\frac z2\e^{-\gamma/2}.
\end{equation}
Since $-v\le-1$, the relevant branch is $W_{-1}$, so
$v=-W_{-1}(-(z/2)\e^{-\gamma/2})$.  Substituting into $x=z/(2v)$
gives \eqref{eq:chi-Lambert-representation}.
\end{proof}

The standard asymptotic expansion of $W_{-1}$ gives the following scaling;
see Ref.~\cite[Sec.~4, pp.~349--350]{Corless1996LambertW}.

\begin{lemma}[Lambert-$W$ scaling]
\label{lem:Lambert-W-scaling}
As $y$ approaches $0$ from above,
\begin{equation}
-W_{-1}(-y)
=
\log\frac1y
+\log\log\frac1y
+O\!\left(
\frac{\log\log(1/y)}{\log(1/y)}
\right).
\label{eq:Lambert-W-asymptotic}
\end{equation}
Consequently, for $\gamma\ge2$ and $0<z\le\gamma$, set
\begin{equation}
L_{\gamma,z}:=\frac\gamma2+\log\frac2z.
\end{equation}
Then, as $L_{\gamma,z}\to\infty$,
\begin{equation}
\chi_\gamma(z)
=
\frac{z}{
2\left[
L_{\gamma,z}+\log L_{\gamma,z}
+O\!\left(\frac{\log L_{\gamma,z}}{L_{\gamma,z}}\right)
\right]}.
\label{eq:chi-Lambert-scaling}
\end{equation}
\end{lemma}

\bibliographystyle{unsrturl}
\bibliography{bibilography}

\end{document}